\documentclass[12pt,english]{article}
\usepackage[markup=underlined]{changes}

\usepackage{comment}
\usepackage{fourier}
\usepackage{newcent}
\usepackage{ulem}
\DeclareMathSizes{12}{13}{10}{7}

\usepackage[multiple]{footmisc}
\usepackage{etex}
\usepackage{graphicx}
\usepackage{setspace}
\usepackage{amsmath}
\usepackage{amsfonts}
\usepackage{titlesec}
\usepackage{pictex}
\usepackage{comment}
\usepackage{amsthm}
\usepackage{graphics,epsfig,verbatim,bm,latexsym,url,amsbsy}
\usepackage{rotating}
\usepackage[authoryear,round]{natbib}
\usepackage{float}
\usepackage[position=bottom]{subfig}
\usepackage{mathrsfs}
\usepackage{multirow}
\usepackage{array}
\usepackage{bigints}
\usepackage{bbm}
\usepackage{geometry}
\usepackage{bbm}
\usepackage[ruled,vlined]{algorithm2e}
\usepackage{soul}
\usepackage{mathtools}
\usepackage{amssymb}
\usepackage{enumitem}
\usepackage{graphicx}
\usepackage{subcaption}
\usepackage{subfig}

\usepackage[status=draft,inline,index]{fixme}
\fxsetup{theme=color,margin=false,mode=multiuser}
\FXRegisterAuthor{mf}{mfa}{\color{orange}MF}
\FXRegisterAuthor{ak}{aka}{\color{cyan}AK}
\FXRegisterAuthor{bx}{bxa}{\color{blue}BX}

\DeclarePairedDelimiterX{\inp}[2]{\langle}{\rangle}{#1, #2}
\usepackage{hyperref}
\hypersetup{
	colorlinks=true,
	linkcolor=red!60!black,
	citecolor=blue!60!black,
	filecolor=magenta,      
	urlcolor=blue!60!black,
}
\usepackage[nameinlink,noabbrev,sort,capitalise]{cleveref}

\newcommand{\de}{\mathop{}\!\mathrm{d}}

\crefname{thm}{Theorem}{Theorems}

\theoremstyle{definition} \newtheorem{example}{Example}
\theoremstyle{definition} 
\theoremstyle{definition} 
\theoremstyle{definition} 
\theoremstyle{definition} \newtheorem{definition}{Definition}
\theoremstyle{definition} 
\theoremstyle{definition} 
\theoremstyle{definition} \newtheorem{lemma}{Lemma}
\theoremstyle{plain} 
\theoremstyle{definition} \newtheorem{assumption}{Assumption}
\theoremstyle{definition} 
\theoremstyle{definition} 
\theoremstyle{definition} 
\theoremstyle{definition}\newtheorem{proposition}{Proposition}
\theoremstyle{definition} 
\theoremstyle{definition} 
\theoremstyle{definition} 
\theoremstyle{definition} 
\theoremstyle{definition} 
\theoremstyle{definition}

\theoremstyle{definition} \newtheorem{remark}{Remark}
\usepackage{mathtools}
\titlespacing*{\paragraph}{0pt}{1.25ex plus 1ex minus .2ex}{0.5em}
        {}% formatting commands applied just to subsection title
\definecolor{mypurple}{rgb}{0.6, 0.4, 0.8}
\definecolor{easygreen}{RGB}{0, 176, 80}
\definecolor{lightblue}{rgb}{0.4, 0.6, 0.8}
\begin{document}

\title{\Large{\textbf{\MakeUppercase{Data-Driven Automation}}}}

\author{\makebox[.25\linewidth]{{Maryam Farboodi}\thanks{MIT Sloan School of Management; email: \protect\texttt{farboodi@mit.edu}}}\\{MIT} 
\and 
\makebox[.25\linewidth]{Andrew Koh\thanks{MIT Department of Economics; email: \protect\texttt{ajkoh@mit.edu}}} \\ {MIT} 
\and \makebox[.25\linewidth]{Anchi Xia\thanks{MIT Department of Economics; email: \protect\texttt{bryantx@mit.edu} \\~\\
\emph{First posted draft: November 2025.} We thank Daron Acemoglu, David Autor, Tom Cunningham, Eric Gao, Basil Halperin, Fredric Kong, Anton Korinek, Stephen Morris, Arjun Ramani, Pascual Restrepo, Phil Trammell, and Iv\'an Werning for helpful conversations. 
}} \\
{MIT}}

\date{
\today 
%\\ \small{\emph{latest version \href{https://www.dropbox.com/scl/fo/si0ovhb85drsr9gx1pr92/AMC3N_UUb6TCcEy-6FIxLGk?rlkey=khaihrbl7d62mm5xespe35xm9&st=qle83c70&dl=0}{\underline{here}}}}
}

\maketitle 
\setstretch{1.1}

\begin{abstract} 
We build a dynamic model of data-driven automation in which data (i)  is  {heterogeneous} and task-specific; (ii) {accumulates endogenously} as a byproduct of economic activity; and (iii)  {exhibits spillovers} such that data generated by one task can augment the productivity of another. Along the transition path of automation, data plays a dual role in simultaneously augmenting the productivity of already-automated tasks and expanding the automation frontier. We derive tight conditions for the economy to be partially versus fully automated in the long-run. In the latter case, automation exhibits rich short-run dynamics that depend on the pattern of data spillovers but is always slow in the long-run: the share of tasks produced by labor decays asymptotically as a power law in time. We show that the economy is generically inefficient and analyze how a planner optimally tilts the direction of data accumulation. With endogenous capital accumulation, data-driven automation generates explosive growth but stagnant long-run wages. 
\end{abstract}

\thispagestyle{empty}

\clearpage 
\begin{comment}
SUBSTANTIVE: 
- \textbf{Foundation for aggregating data}: 
    \begin{itemize}
        \item Emphasize that $i \in X$ is the effective unit that DATA is accumulated. But it need not be the most 'primitive' i.e., tasks can be downstream of 'skills'. 
        \item Data on task $i$ $\to$ get better at skills $\to$ improve task $j$. Cobb Douglas gets linear (?) 
        \item RL vs pretraining: look into the literature to motivate functional forms 
\end{itemize}
\end{comment}

\setcounter{page}{1}
\section{Introduction}
Rapid advancements in artificial intelligence have sparked intense debate about the future of work and the trajectory of automation across the economy. For instance, Dario Amodei, CEO of Anthropic has warned of a `white collar bloodbath' \citep{axios2025amodei}. A recent US Senate report has likewise warned that nearly $100$ million US jobs might disappear over the next decade \citep{us_senate_2025_ai_jobs}. What are we to make of such claims? While previous waves of technological change primarily displaced routine tasks \citep{autor2003skill}, modern AI systems, trained on vast corpora of human media, increasingly demonstrate capabilities that rival human performance across a growing range of cognitively demanding tasks. For instance, the ability of AI systems to complete long software tasks---those that would take $t$ units of time for human software engineers---has been growing exponentially since the release of GPT-2 in 2019 \citep{kwa2025measuring}. Such technologies raise fundamental questions about which tasks of the economy will be automated, how quickly this transition will occur, whether automation will be limited to certain domains or become pervasive across all aspects of economic life, and what this might mean for growth and wages in the short and long run. 

Central to understanding the process of automation is recognizing that modern machine learning systems are fundamentally data-driven: AI systems require extensive training data to perform tasks, and their performance continues to improve even in the `big data' regime of tens of trillions of tokens. Put differently, AI and data are strong complements. Moreover, there is a growing trend of finetuning models on task-specific data.\footnote{Examples include  \citet{capoot2025harvey} for legal tasks, \citet{jiang2023health} for medicine, and (\cite{wu2023bloomberggpt} and `BloombergGPT' for financial analysis.} All of this data does not arise in a vacuum, but was rather generated as a byproduct of economic activity, giving rise to the following \emph{data-automation feedback loop}: 

\emph{Data $\to$ Automation.} Modern AI systems rely on staggering amounts of data. Pre-trained models ingest virtually the entire corpus of human work, and are increasingly fine-tuned on task-specific data---workflows, norms, and job-specific details e.g., call transcripts, a law firm's legal cases, or a hedge fund's investment process---to improve their performance in specific domains. A striking empirical regularity---so-called `neural scaling laws'---is that more abundant training data and/or model parameters leads to predictable decreases in the loss across a wide suite of model architectures.\footnote{Some examples of losses are coding mistakes,  wrong medical diagnoses, malfunctioning of robots or autonomous vehicles, or wrong mathematical proofs.}  This shapes \emph{task-specific} capital productivity and, in turn, drives the path of equilibrium automation. 

\emph{Automation $\to$ Data.} Where does this data come from? Up to now, large-language models have been largely pre-trained on the internet. % and subsequently post-trained to perform specific economic tasks (e.g., coding, medicine, law). 
However, a growing view within the computer science community is that in order for AI models to perform economically valuable tasks presently done by humans, simply training larger models with more compute may not be enough: they need to be trained on \emph{task-specific data}\footnote{For an articulation of the view that better task-specific data is needed for automation, see \url{https://www.mechanize.work/blog/sweatshop-data-is-over/}} that might encode tacit know-how and contextual knowledge that can only be gained from learning-by-doing. For instance, this is the explicit ambition of firms that build self-driving vehicles, and have become increasingly prevalent---OpenAI reportedly expects the “entire economy" to become a “Reinforcement Learning Machine,” where AI systems might train on recordings of how professionals in all fields handle day-to-day work on their devices \citep{information2025reinforcement}. Indeed, there are increasing attempts to collect data in the workplace \citep{paul2026meta} and Silicon Valley is beginning to build virtual work environments --- a `virtual machine outfitted with an email inbox, a Slack account, some coding tools and a web browser' \citep{nyt2025mechanize} --- in which firms might train an AI system to automate the \emph{entirety} of a white collar job. Hence, the intensity and composition of production---and automation---across the economy will shape task production and, in so doing, drive differential data accumulation. 

Motivated by these developments, we build a simple macroeconomic model to incorporate this data-automation feedback loop. Our model takes three distinctive features of data seriously. First, data is a  byproduct of economic activity---different kinds of data are differentially valuable across tasks. Second,  data is a long-lived asset that accumulates endogenously as a by product of economics activity. Third, data   exhibits spillovers---data generated by one task can augment productivity of another task. 

We characterize the scale and speed of automation and capital concentration in equilibrium,  as a function of the structural parameters of the economy. Our model delivers the following key insights:

%\sout{We show that endogenous capital and data accumulation interact to generate unbounded long-run growth but stagnant long-run wages. Furthermore, we illustrate that the equilibrium is generically inefficient, and the planner accelerates the equilibrium process of data convergence or divergence depending on the degree of cross-task complementarity. Finally, we show the combination of data-driven automation and capital accumulation result in increasing returns to scale, giving rise to unbounded output and complete automation in finite time.}
\begin{enumerate}
    \item \emph{Long-run automation:} The degree of automation in the long-run depends jointly on cross-task substitutability and the extent of diminishing returns to data.

    \item \emph{Wages:}   In equilibrium labor `trains its own replacement'.  Long-run wages stagnate whenever the economy is fully-automated in the limit e.g., whenever tasks are gross complements.

            \item \emph{Speed of automation:} The path of equilibrium automation can be rapid or imbalanced, and depends richly on the network of cross-task spillovers. But automation eventually slows down: sans capital accumulation, the proportion of tasks produced by labor decays according to a power law in time. 
            
                \item \emph{Contagion}: Automation can spread contagiously across the economy in the presence of cross-task data spillovers driven by either transfer learning or overlaps in primitive `skills'. 

    \item \emph{Core-periphery automation:} The network structure of data spillovers can generate imbalanced automation in the short-to-medium run in which a small number of full-automated sectors (`core') are overproduced and generate data that feedback into those sectors while bottlenecks to spillovers slow automation in the wider economy (`periphery').
        \item \emph{(In)efficiency:} The composition of data accumulation is generically inefficient, and the direction of inefficiency i.e., whether data-rich sectors are accumulating data too slowly or quickly depends on the degree of cross-task substitution.
    \item \emph{Explosive growth:} Data-automation feedback loops drive long-run explosive growth in the presence of endogenous capital accumulation. 
\end{enumerate}
We now explain these results in a more detail. We first shut down capital accumulation and analyze an economy in  data autarky where data is purely task-specific and there are no cross-task spillovers. We describe conditions under which the economy achieves full versus partial automation in the long-run, and study the transition path of automation. We show that the direction and limit of automation are determined by two key parameters: $\sigma$, which governs the substitutability across tasks, and $\eta$, which governs the returns to scale of data on capital productivity.

When $\sigma \leq 1/\eta$, the degree to which tasks are complementary exceeds the speed at which capital productivity runs into diminishing returns from additional data. We show that the economy is then fully automated in the long run (\cref{prop:complements}). The economy asymptotes towards a \emph{balanced data path} while task-specific productivities continuously improve. Intuitively, the capital allocation is not too skewed towards data-rich tasks  and hence the ratio of task capital productivities approaches a constant. In such economies, data-driven automation ensures that the set of tasks produced by labor decline quickly enough so that our economy escapes the cost-disease trap \citep{baumol1967macroeconomics}. 

Alternatively, when $\sigma > 1/\eta$, the degree of task complementarity is dominated by the speed at which capital productivity runs into diminishing returns. In this case, we show that the economy can remain perpetually on an \emph{imbalanced data path}---equilibrium capital allocation tilts toward data-rich tasks which, in turn, makes them more data-rich. In such cases, the ratio of data in task $i$ (data-rich) on task $j$ (data-poor) approaches infinity so that automation can, in fact, fall along the equilibrium path and the economy is only partially automated in the limit, with capital concentrated in a small set of highly productive tasks (\cref{prop:substitutes}). 

What does this mean for wages? We show that in the case $\sigma <1/\eta$ regime, wages initially rise but eventually stagnate as the productivity gains from data-augmented capital is offset by endogenous displacement from a growing automation boundary (\cref{prop:wages}). By contrast, when tasks are sufficiently substitutable and automation is partial ($\sigma > 1/\eta$), long-run wages grow unboundedly. Together, this implies, somewhat surprisingly, that cross-task complementarities hurt long-run wages. Indeed, the `partial equilibirum' logic is that for a \emph{fixed} vector of capital productivities, decreasing $\sigma$ (thus making tasks more complementary) increases wages. But when data is dynamically accumulated in equilibrium, an increase in cross-task complementarity means that production of `data-poor' tasks---those that are yet to be automated---are higher in equilibrium, and this, in turn, displaces labor and drives wages down until the `productivity effect' (from automated sectors being more productive) and the `displacement effect' (from labor performing a shrinking share of tasks) exactly offset each other.

Furthermore, we develop tight non-asymptotic bounds on the speed of data-driven automation (\cref{prop:heterogeneous-speed}). We show that although equilibrium automation can be potentially rapid in the short term, it takes a long time for full equilibrium automation i.e., for (almost) all tasks to be performed by capital in equilibrium. In particular, we show---sans capital accumulation---the share of tasks produced by labor decays as a power law in time. Thus, although data-driven automation is powerful in the limit, on its own it is also \emph{slow}: there remains a `fat tail' of tasks that continue to be performed by labor---not because capital is incapable of doing so, but because, given the prevailing price of compute and the limited efficiency of capital at producing those tasks, it is more worthwhile to produce them with labor. 

%\question{the next two paragraphs go} 
We then generalize the model by introducing cross-task data spillovers into our economy. We say that the economy is \emph{connected} if, for any two tasks $i,j$, there exists a `positive measure path' of tasks through which data can spillover from $i$ to $j$. Connectedness reflects the idea that even though tasks $i$ and $j$ may not directly benefit from each other's data, we can find a chain of intermediate tasks $i := i_1 \to i_2 \to i_3 \to  \ldots \to j$ that are beneficiaries of the preceding task's data. This series of spillovers ensure a kind of contagion through quantities produced in general equilibrium---an increase in $i = i_1$'s data augments $i_2$'s factor productivity that raises its equilibrium quantity. This increases the speed at which data on task $i_2$ accumulates, which then augments $i_3$'s factor productivity, and so on.\footnote{We model data spillovers via a graphon (a dense continuous graph) so connectedness comprises of a `positive measure path'.} We show  that connectedness is sufficient to guarantee full automation in the long-run independent of the elasticity of substitution $\sigma$, return-to-scale $\eta$, and the initial condition on the data stock across tasks (\cref{prop:generalW_full_automation}). These, `local data spillovers' can, in effect, drive automation by substituting for global cross-task complementarities. 
We further show that connectedness is sufficient to guarantee a balanced data path and a balanced composition of tasks produced in the long run. 

The network structure of cross-task spillovers also matter enormously for automation dynamics in the short-run. We illustrate this by specializing our model to \emph{core-periphery automation} in which a subset of densely connected sectors (`core') are data rich, and as they are produced in equilibrium, data generated feeds back into the core, further augmenting their capital productivity. By contrast, there are only weak spillovers to remaining sectors that comprise the broader economy (`periphery'). Although our previous long-run results state that this is sufficient to guarantee full long-run automation, how long might this take? We show that patterns of spillovers can generate imbalanced short-run automation in which a subset of tasks/sectors are rapidly automated but the wider economy is still dominated by labor. We analytically characterize the time at which the periphery sectors begin to be automated: it hinges on both the (i) strength of within-core linkages that reduce equilibrium incentives to automate peripheral sectors; and (ii) the severity of the bottlenecks in how quickly data flows from the core to periphery. %\question{up to here}

%\textcolor{blue}{BX: question - do we want to give a more narrative take to these propositions? for instance, the contagion story in the real world would be nail CS --> hoping for sufficient overlap --> automate less similar tasks.}

Next, we address the efficiency of data-driven automation and show that equilibrium automation is generically inefficient (\cref{prop:generic_inefficiency}). We show that along the transition path, the planner emphasizes the equilibrium direction of data convergence or divergence.  When tasks are sufficiently complementary $(\sigma \leq 1/\eta)$, the planner wishes to initially tilt the allocation of capital further \emph{toward} data-poor tasks relative to the equilibrium allocation---this speeds up data accumulation for bottlenecked tasks at the cost of changing the equilibrium allocation of capital. Conversely, when tasks are strong substitutes $(\sigma > 1/\eta)$ the planner wishes to tilt the allocation of capital further \emph{away from} data-poor tasks relative to the equilibrium allocation. This offers a new perspective on how market forces can distort the \emph{direction} of data production because firms do not internalize the impact of their task production on the evolution of the aggregate data stock.

Finally, we introduce endogenous capital accumulation into our model.  We show that even in the case in which tasks are strongly complementary, additional data delivers steeply diminishing returns, and capital depreciates quickly, the fact that data (i) {augments productivity} of already-automated tasks; (ii) {expands the automation boundary};\footnote{See, e.g. \cite{solow1957technical,acemoglu2018race,aghion2017artificial}.} and (iii) {does not depreciate} imply a growth explosion in finite time (\cref{prop:singularity}). In contrast to the `supply-side singularity' result of \cite{nordhaus2021we}, labor and capital are not necessarily gross substitutes at any point in time, but automation will eventually give rise to a representation of the economy where the effective substitutability is governed by the perfect substitutability between these two factors at the task level. 

\textbf{Connection to the literature.} Our work relates closely to the task model \citep{zeira1998workers} and recent work analyzing the trajectory as well as impact of automation \citep{acemoglu2018race,aghion2017artificial,korinek2024scenarios,trammell2025workflows,demirer2025economic}. Different from these papers, our focus is on how task-specific data accumulation---determined both in general equilibrium, as well as via local spillovers---drives the \emph{speed}, \emph{direction}, and \emph{limit} of automation across the economy. More related is \cite{jones2024framework} who show that expanding automation and growing capital productivity can jointly produce balanced growth if intermediate goods are gross complements. \cite{jones2024framework} embody technological change via an R\&D process in the tradition of \cite{romer1986increasing}---by contrast, our model takes seriously the idea that data rather than labor (in the form of human scientists) is the key input---and bottleneck---to modern AI systems, and it is the ensuing path of equilibrium data accumulation that delivers rising capital productivity.\footnote{Our focus is on the data dimension of modern AI systems and, in so doing, abstract unemployment; see \cite{costinot2023robots,guerreiro2022should,beraja2025inefficient} for analyses of policy remedies.
} 

Data-automation feedback loops are qualitatively distinct from previous analyses of big data as being used to predict demand or improve product quality \citep{jones2020nonrivalry,farboodi2021model}.\footnote{These models build upon the pioneering work of \cite{wilson1975informational}.} This is because the data-automation feedback loop emerges when data is specifically used for  \emph{training} AI models that can be duplicated at scale both within and across firms, and does not depreciate over time.\footnote{See, \cite*{chatterji2025transformative,hadfield2025agents} for possibilities for how AI agents might be incorporated within firms.} We think this qualitatively distinguishes data-automation feedback loops from previous analyses e.g., predicting demand \citep{bajari2019impact}.\footnote{Of course, on a high level all data is used for prediction; indeed, even with large-language models or reinforcement learning algorithms, training data serves to augment next-token prediction. However, we think this is economically distinct from previous analyses of big data.}

Finally, our analysis of local spillovers across nearby products, produced by tasks that share a subset of their tasks, is informed by modern machine learning systems that exhibit a degree of local `transfer learning'.\footnote{See, e.g., \cite{dwivedi2019representation} who develop a measure of task similarity by assessing the performance of a model trained on different tasks. A recent interview by employees of Anthropic echo this ambition for large-language models \citep{patel2025scaling}.} 
%If, for instance, a reinforcement learning algorithm can perform task $x$ efficiently (because it has been trained on ample data on $x$) then we might expect it to also do well at a sufficiently nearby task $x'$.  
To capture this, we endow our task space with a measure of distance, and tasks that are `nearby'  have more similarities in their production routine and thus exhibit more spillovers. Our formalization employs the language of graphons developed in mathematics \citep{borgs2008convergent}, which has recently been drawn on to study network games \citep{parise2023graphon} and contagion \citep*{erol2023contagion,koh2022speed}. Our setting is substantially different because the intensity of spillover is jointly determined by `local' and `global' general equilibrium forces, and---as is typical in macroeconomics---in the long-run data for all tasks diverge to infinity. Thus, it is their \emph{relative speeds} that determines equilibrium automation.

\section{A Data-Augmented Task Model} \label{model section}

%This section is organized as follows. Section \ref{model} introduces our modifications to the task model to capture the role of data. Section \ref{spillovers} details how data spillovers across tasks. Section \ref{market clearing} characterizes the market clearing prices as well as the degree of automation along the equilibrium path.

%\subsection{A data-augmented task model} \label{model}
We begin by developing a simple model of production with task-specific data accumulation. 

\paragraph{Final good and intermediate tasks.} There is a task space $X = [0,1]$. A final good is competitively produced by aggregating intermediate tasks according to the usual CES technology
\begin{equation} \label{GDP}
Y= \Big(\int_{i \in X} y_i^{\frac{\sigma-1}{\sigma}}\Big)^{\frac{\sigma}{\sigma-1}} 
\end{equation}
where $y_i$ is the quantity produced of task $i \in X$. $\sigma > 1$ corresponds to gross substitutes, $\sigma < 1$ corresponds to gross complements, and the limit of $\sigma = 1$ corresponds to the neutral, Cobb-Douglas case. 

\paragraph{Factors of production.} We consider two factors of production, capital and labor. In our baseline model, capital and labor are fixed at $K > 0$ and $L> 0$ respectively.\footnote{In \cref{sec:extensions} we present an extension with capital accumulation, where we introduce a standard household block that governs the consumption-savings decision in our economy. There we analyze how endogenous capital accumulation interacts with data-driven automation.} We emphasize that the `labor block' in our model is deliberately stylized such as to focus on the impact of how data accumulation might augment capital productivity hetrogenously across tasks; see, e.g., \cite*{acemoglu2025tasks,autor2025expertise} for analyses of how automation can have differential occupational impacts.

\paragraph{Data.} Data is a state variable that governs how productive capital is across tasks. Data is task-specific: the data that is used to train an AI model to build financial models is different from that used to train a surgery robot. We use $\bm{D}_t \in \mathcal{D}$ to denote the (infinite-dimensional) vector of time-$t$ data across tasks, and write $D_{it}$ to denote the cumulative data at time $t$ from producing task $i$.\footnote{$\mathcal{D}$ is the space of (infinite-dimensional) data stock that is the space of measurable functions from $X \to \mathbb{R}_+$ and for $\pmb{D} \in \mathcal{D}$, $i \in X$ we write $\pmb{D}_i := \pmb{D}(i)$ to denote the stock of data for task $i$.}
%We assume that the initial stock of data $\bm{D}_0$ is weakly decreasing in task indices i.e., $D_{0i} \geq D_{0j}$ whenever $i \leq j$. 

\paragraph{Tasks.} The production function of task $i$ is linear in capital and labor, and is given by: 
\[
y_i = \psi^{L}l+ \psi^K_i ( \mathcal{A}_i(\bm{D}) )\cdot k 
\]
where $\mathcal{D} \mapsto \mathcal{A}_i(\bm{D}) \in \mathbb{R}_+$ is a \emph{data aggregator} tracking how the full vector of data translates into the \emph{effective data} available for task $i$. We expand on properties of $\mathcal{A}_i(\bm{D})$ in detail in Section \ref{subsec:dataProd}.

As capital and labor are perfect substitutes, each task will be performed using the input that is more productive for that specific task. When a task is produced with capital in equilibrium, we say that it is {\it automated.} We emphasize that automation is an equilibrium notion which depends jointly on technology ($\psi$) and factor prices ($r$ and $w$). To keep track of the degree of automation in our economy, it is useful to introduce the following equilibrium object. 

\begin{definition}[Automation boundary]\label{defn:boundaryAuto}
    The time-$t$ \emph{automation boundary} $\gamma_t \in \text{Int}(X)$ denotes the threshold task that is automated---tasks $i \leq \gamma_t$ are produced via capital, and tasks $j > \gamma_t$ are produced via labor.\footnote{We break ties in favor of automation without loss of generality.}
\end{definition}

Finally, to incorporate diminishing return to data, we impose the following assumptions on how data shape the marginal productivity of capital.

\begin{assumption}[Productivity of data within and across tasks] \label{ass:DRS}
The vector of {\it effective data} determines the marginal productivity of capital in each task $i \in X$ according to the following functional form,
    \[
    \psi^{K}_i\Big(\underbrace{\mathcal A_i \left(\boldsymbol{D}\right)}_{\substack{\text{effective} \\ \text{data}}} \Big)= f_i \cdot \left(\mathcal A_i\left(\boldsymbol{D}\right)\right)^{\eta}  
    \]
where $\eta \in (0,1)$, common across all tasks, captures the degree of diminishing returns of data in capital productivity and $f_i > 0$ is a task-specific scalar coefficient that weakly decreases in the task index $i$. We use $\pmb{f}$ to denote the vector of task-specific scalars, and assume that it is bounded above and below. 
\end{assumption}

An alternate interpretation of $f_i$ is that it encodes \emph{differential speeds} of data accumulation per unit of task produced. We discuss this interpretation momentarily along with introducing spillover of data across tasks to capture the degree of transfer learning as documented in computer science.

\subsection{Evolution of Data and Cross-Task Spillovers} \label{subsec:dataProd}

The initial stock of data is given by $\bm{D}_0$, which is bounded, continuous, and strictly decreasing in task index. The stock of data evolves according to the following law of motion: 
\begin{align}
    D_{it} = D_{i0} + \int_0^t y_{is} ds.\label{eq:DlawM}
\end{align}

Thus, the accumulation of data generated by task $i$ is simply the production of $i$. This reflects the key channel through which economic forces that determine equilibrium task quantities feed back into data accumulation. Indeed, in addition to text and image inputs (e.g., a legal firm's past cases), modern AI training procedures are increasingly adept at taking in video inputs of the physical tasks (e.g., surgical tasks \citep{kim2024surgical}) being performed and learning how to imitate accordingly. Note here that, different from \cite{farboodi2021model,jones2020nonrivalry}, data does not depreciate---if an AI system augments the productivity of a given task, it can always do so.\footnote{For instance, the depreciation of data in \cite{farboodi2021model} arises from a changing state so past data is less predictive of future states. However, we view a given task $i \in X$ as fixed; we discuss how our model might accommodate changing composition of tasks across the economy a la \cite{acemoglu2018race,korinek2024scenarios} in \cref{sec:extensions}.}

%begin{comment} \question{in my suggestion, things are removed from here on. We just say in this paper we restrict attention to each task onlyusing its own data and in the companion paper, we consider the spillover.} \end{comment}

Note that \cref{eq:DlawM} specifies the  \emph{task-specific data accumulation}, i.e., how $\bm{D}_t$ is determined. Next, we explain the determination  of $\mathcal{A}_i(\bm{D})$, i.e., how data is employed in production. Here, we assume that although data accumulation is task-specific, the data accumulated for a given task might be valuable for other tasks---$\mathcal{A}_i(\bm{D})$ can depend on $D_j$ for $j\neq i$. That is, our economy might exhibit \emph{cross-task spillovers}. 

%To flexibly capture heterogeneity in how data for a given task spills over to other tasks, we endow our task space $X$ with the Euclidean distance i.e., $\|i - j\| := |i - j|$. Data spillovers are \emph{local} in the sense that task $i$ benefits more from data generated by the production of `nearby tasks'. We think of an AI system's performance on each task $i \in X$ as being downstream of some set of latent and more primitive `capabilities' that might be broad (e.g., verbal, numerical, or spatial reasoning which have emerged via pre-training), or narrow (e.g., at specific coding or writing tasks that improve with finetuning). Hence, data for task $i$ can be thought of as augmenting the attendant capabilities that constitute task $i$. For another task $j$ whose requisite capabilities overlap with that of $i$'s, task $i$'s data $D_{it}$ is, in effect, valuable for the production of $j$.

To describe patterns of spillovers across tasks, we define the directed graphon as the measurable function\footnote{Graphons are sometimes defined with the condition of symmetry, but things are well-defined without it.}  $W: [0,1]^2 \to \mathbb{R}_+$. The data aggregator for task $i$ under $W$ is: 
\[
\mathcal{A}_i(\bm{D}) = \int_{j \in X} W(i,j) \cdot D_{j} dj. 
\]
$W(i,j)$ denotes how much the productivity of task $i$ is influenced by the data of task $j$ and the `effective data' that feeds into task $i$'s capital productivity  simply integrates over tasks $j \in X$. There are two ways to interpret our graphon $W$. First, it might capture the idea that two tasks $i,j \in X$ might share underlying skills such that data on how task $i$ is performed might improve the performance of AI systems on those underlying skill that, in turn, improves the productivity of task $j$. Second, the graphon $W$ might simply capture the degree of `transfer learning'---the ability of AI systems to generalize to new domains. This remains very much an empirical question fiercely debated among computer scientists.\footnote{For example, Demis Hassabis the CEO of DeepMind recently said in an interview that: ``For example, when these large models improve at coding, that can actually improve their general reasoning. There is evidence of some transfer, although we would like a lot more evidence of that.'' \citep{patel2025scaling}.} Our model of graphons keeps the shape and magnitude of data spillovers flexible, allowing us to analyze implications for the direction and limit of automation without taking a particular stance on the exact form of transfer learning.

\begin{figure}[t]
\caption{Examples of cross-task spillovers}
    \centering
    \subfloat[Geometric spillovers]{
        \includegraphics[width=0.33\textwidth]{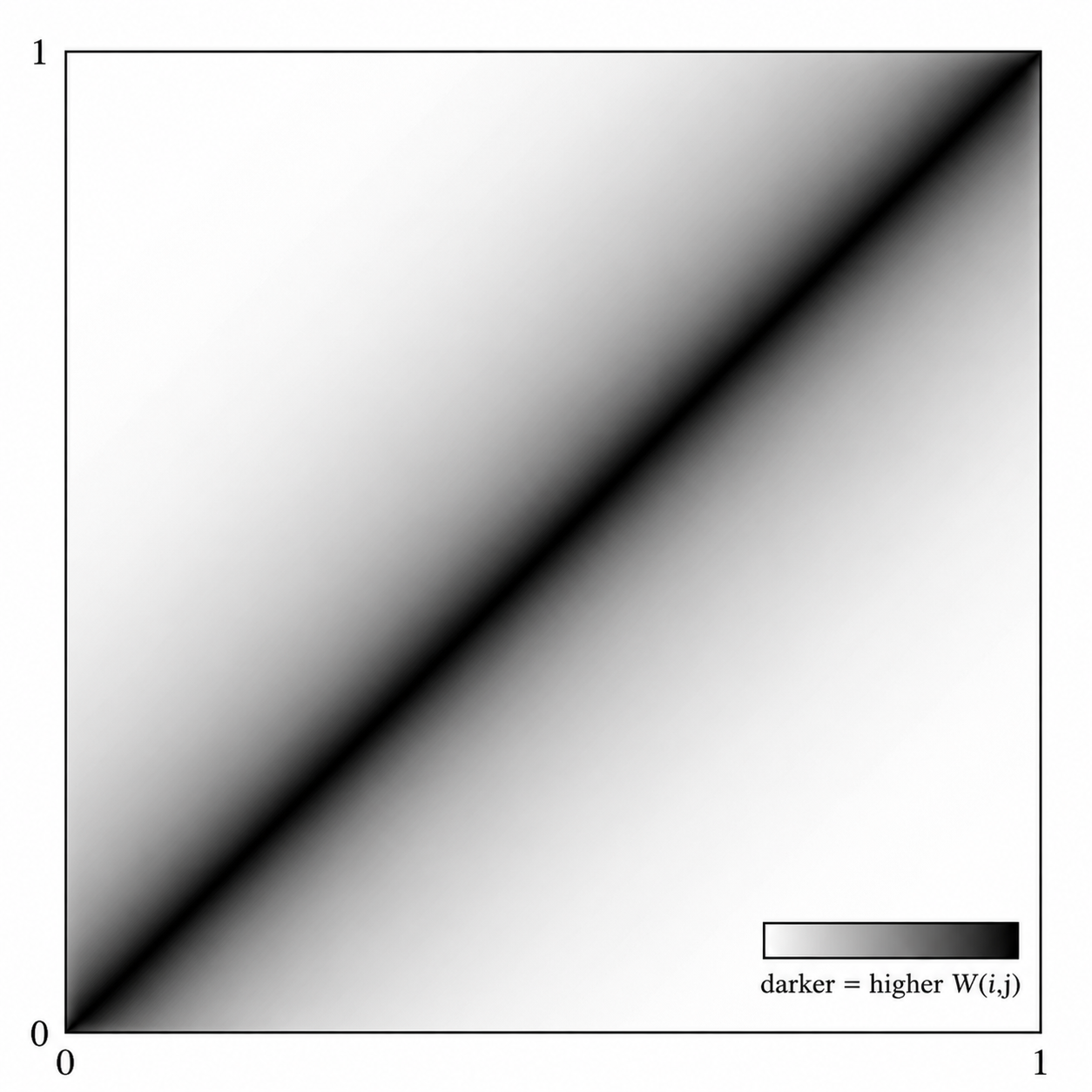}}
        \subfloat[Block structure]{
        \includegraphics[width=0.33\textwidth]{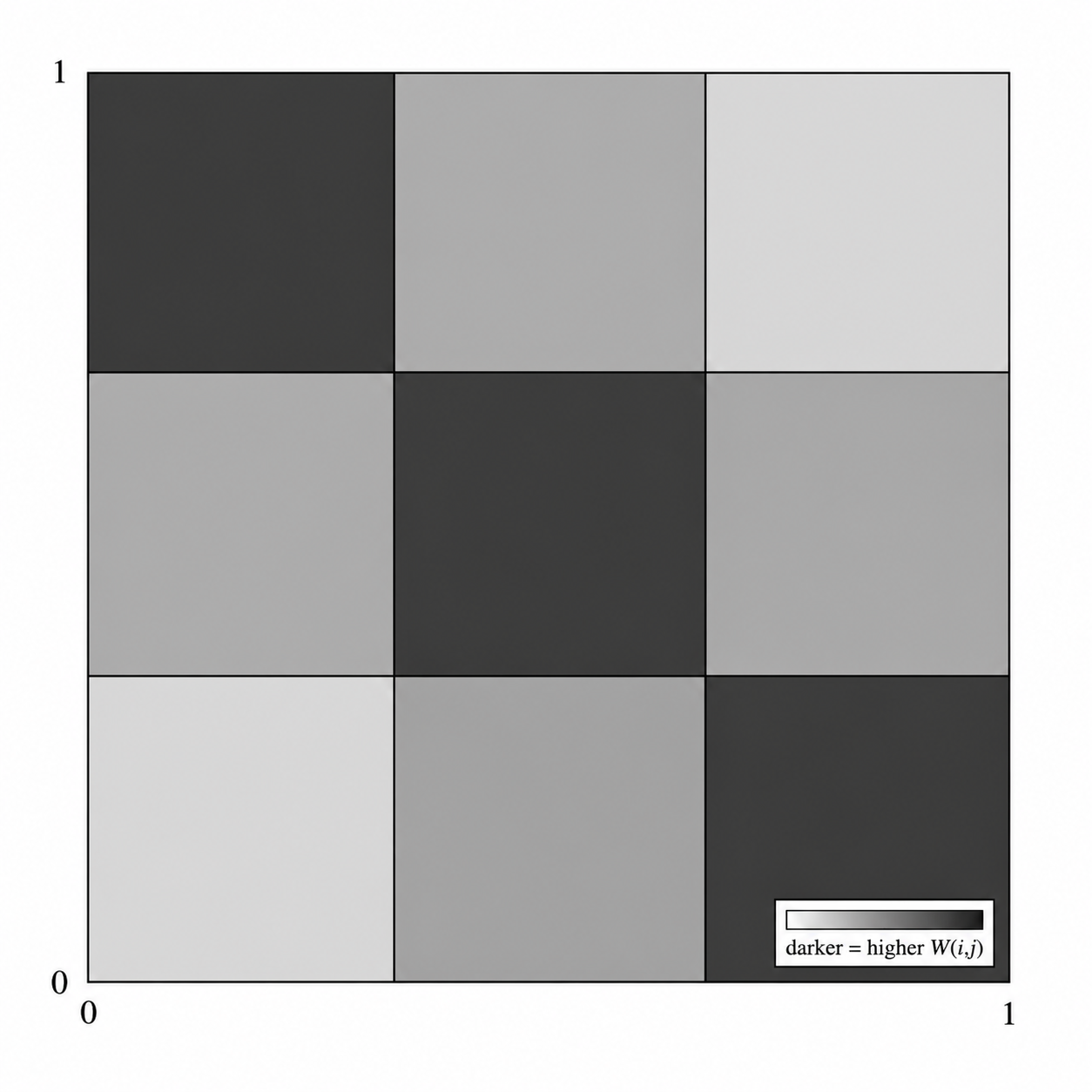}%
    }
        \subfloat[Rank 1 structure]{
        \includegraphics[width=0.33\textwidth]{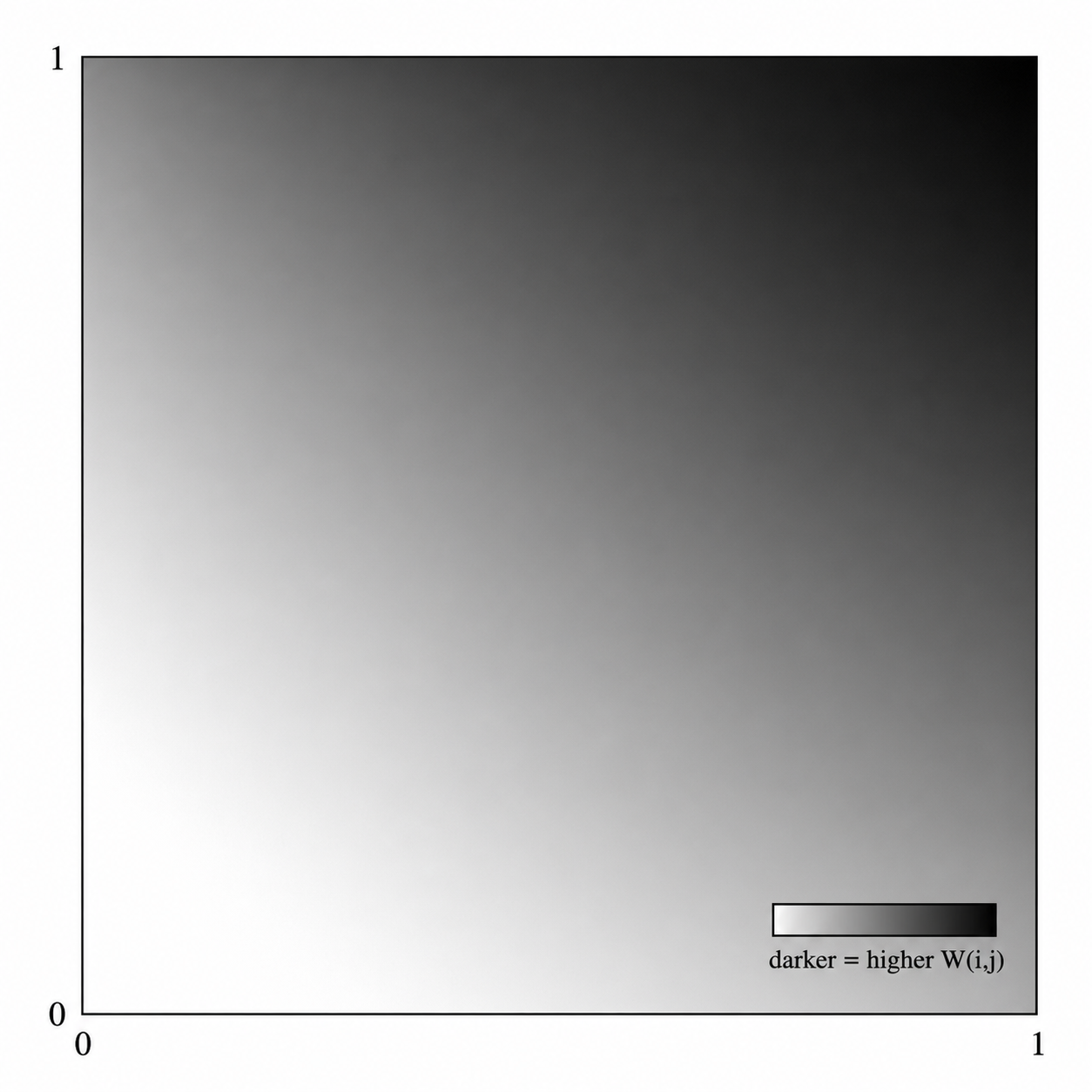}%
    }
   \label{fig:W_examples}
\end{figure}

\cref{fig:W_examples} illustrates three examples of graphons: 
\begin{itemize}
    \item[(a)] \emph{Geometric spillovers}: endow the task space $X$ with the Euclidean distance and data spillovers are \emph{local} in the sense that task $i$ benefits more from data generated by the production of nearby tasks than those further away.\footnote{See \cite{koh2022speed}, who also endow the graphon with geometry to study the speed of game-theoretic contagion, and more broadly the literature in probability theory on random geometric graphs.}
    This specification is natural when the index \(i\) orders tasks by technological or occupational similarity. Data generated in one task is then most useful for nearby tasks whose inputs, workflows, or prediction problems are similar.

    \item[(b)] \emph{Block structure}: tasks are partitioned into subsets (`blocks'). Link strength is constant across each pair of blocks, so that tasks in the same origin block have identical spillovers to tasks in the same destination block. This nests finite-sector models as a special case.
    This specification is natural when tasks can be grouped into sectors or jobs that are broadly similar. For example, data may travel easily within a sector, while cross-sector spillovers depend mainly on the pair of sectors involved.

    \item[(c)] \emph{Rank-one structure}: some tasks are systematically stronger sources of outflows, while others are systematically stronger than inflows. Then our graphon is $W(i,j) = u(i) v(j)$. For instance, `mathematical reasoning' might be upstream of many other tasks (stronger outflows, high $v$) while `writing' might be receptive to a broad source of data (stronger inflows, high $u$). 
\end{itemize}

We now have all the necessary elements to formally define the equilibrium in  this economy, provided in Definition \ref{market clearing}.

\begin{definition}[Equilibrium] \label{market clearing}
At each time $t$, given the data stock $\bm{D}_t$, an equilibrium comprises of a rental rate $r_t$, a wage $w_t$, and allocations of capital and labor across tasks such that (i) the final good producer is maximizing her time-$t$ profit; and (ii) intermediate task producers are maximizing their time-$t$ profits.
\end{definition}

\subsection{Discussion} \label{subsection:discussion} Before presenting our results, it will be helpful to discuss our main assumptions governing data-driven automation.

\paragraph{D1: Data as a source of heterogeneous capital productivity.} The performance of modern AI systems can, to a first order, be broadly attributed to (i) advances in compute; (ii) algorithmic improvements; and (iii) increases in training data \citep{ho2024algorithmic}. While cheapening compute and algorithmic improvements impact capital productivity in a homogeneous manner (via the price of capital $r$), differences in the availability of task-specific data might explain differences in performance of AI systems across economic tasks. 

Indeed, an interpretation of $f = \left(f_i\right)_{i \in X}$ (recall $f_i$ is the constant governing how well data maps to capital productivity at task $i$) is that it encodes differential rates of data accumulation.\footnote{That is, our economy is \emph{identical} to another economy in which the functional form of capital productivity is the same across tasks, but different tasks accumulate data at different rates per unit of output.} There are several reasons why tasks might differ in their rates of data accumulation. For instance, tasks for which $f_i$ is low might correspond to those that can only be done by AI systems trained on long multi-step procedures with sparse reward signals (e.g., low `bits per sample') \citep{ord2025inefficiency,patel2025inefficiency}. Tasks for which $f_i$ is high might correspond to those that can be completed by AI models trained on denser e.g., video data \citep{wiedemer2025video}. Likewise, some tasks might be more easily verifiable than others, thereby corresponding to a quicker rate of data accumulation per unit produced \citep{karpathy2025verifitable}. Our model offers a language to formalize this `generation-verification gap' in which equilibrium forces determine the cost of generating data for different tasks, as well as the quantities actually generated.

\paragraph{D2: Productivity improvements via data.} We make three remarks about \cref{ass:DRS}, illustrated in Figure \ref{fig:psi_illust}. First, our functional form is motivated by the literature in computer science (see, e.g., \cite{hoffmann2022training} and a large body of subsequent work) demonstrating that the loss---a measure of imperfection for a task---diminishes polynomially with the amount of data used for training large language models. Heuristically, they find that the loss approximately takes the functional form: $\mathsf{LOSS} \simeq \frac{1}{N^{\alpha}} + \frac{1}{D^{\beta}}$ for $\alpha, \beta > 0$ where $D$ is data and $N$ is the number of parameters which is typically linear in the amount of floating point operations (`compute').\footnote{\cite{hoffmann2022training} estimate this for the Huber loss; see also \cite{besiroglu2024chinchilla}.} A growing view from computer science is that, given compute has been growing far more quickly than data (on which models can be trained),\footnote{For instance, 
the training compute of frontier models have been growing at $4-5$ times per year \citep{sevilla2024training}.} the performance of AI systems---especially on specific economic tasks---might be increasingly bottlenecked by the scarcity of high quality data \citep{kim2025pretraininginfinitecompute,narayanan2025ai}. Moreover, building competent AI systems for the real economy (as opposed to, say, benchmarks) might not only require large and powerful pretrained models, but also data on the specific economic context within which such systems are deployed. Our model takes this seriously, and approximates the `data-constrained' regime in which $N$ is large so $\mathsf{LOSS} \approx 1 /D^{\beta}$. Under the (\emph{substantial!}) assumption that economic productivity of a given task is the reciprocal of the statistical loss (perhaps raised to some power) we exactly recover the assumed functional form.\footnote{See also \cite{farboodi2021model} where the quality of data is a strictly decreasing function of the quadratic loss from prediction.}

\begin{figure}[t]
\centering
\includegraphics[width=0.40\textwidth]{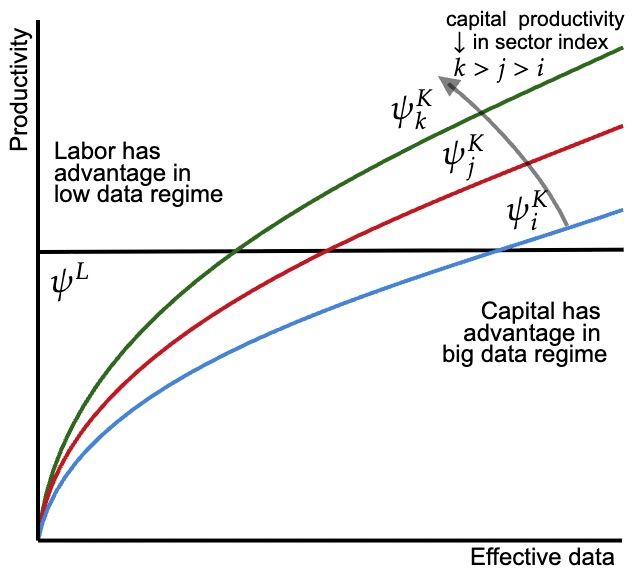}
    \vspace{-0.5em}
    \caption{Illustration of $\psi^K_i$ vs $\psi^L$}\label{fig:psi_illust}
\end{figure}

Second, the economic content of \cref{ass:DRS} reflects an extreme version of the ideas that (i) humans are more \emph{sample efficient} than modern AI systems----in the `low data' regime where training data for tasks is low, labor is more productive than capital; and that (ii) AI systems \emph{saturate more slowly}---they are able to exploit vast amounts of data to drive loss down so that in the `big data' regime, the opposite holds. 

Third, we have chosen to let capital productivity grow unboundedly with data because what ultimately matters for equilibrium automation is the task cost $r/\psi^K$, which would also go to $0$ with capital accumulation (as in \cite{acemoglu2018race,aghion2017artificial}) holding capital productivity fixed. The key difference is that our model with hetrogeneous and endogenous data accumulation allows us to speak to the \emph{ratio} of cross-task productivity, i.e., $\psi^K_i / \psi^K_j$ which drives the speed and limit of automation.  

\paragraph{D3: Raw versus labeled data} Our model is one of endogenous data accumulation as a byproduct of economic activity. This is already the case for certain industries such as self-driving cars in which data accumulates with production. We think our model of endogenous data accumulation well-approximates such tasks. For other tasks, however, current AI systems are not yet able to perform them---no matter how much computational power is thrown at the problem. For such tasks, the AI industry has resorted to hiring humans to produce data on which such systems can be trained \emph{separately} rather than as a byproduct. For instance a wave of new tech firms are hiring lawyers, doctors, bankers, and journalists to lend their expertise in labeling datasets AI systems can be trained on \citep{wsj2025mercor}. We think this might simply be a `cold start' problem that will be eventually overcome with enough human labeling and/or as workplace surveillance technologies improve and AI systems become increasingly able to learn from multimodel domains e.g., video (see \cite*{cullen2025labor}). But as soon as the task can be performed with capital, we expect equilibrium quantities to drive how task data is accumulated in equilibrium.\footnote{Our model can accommodate this `cold start' in which, in the absence of enough data capital productivity is zero by writing $\psi^K_i(\mathcal{A}) = \max\{0, f_i \cdot \mathcal{A}^{\eta} - C_i\}$ where $C_i > 0$ is some task-specific constant such that capital  requires $[C_i/f_i]^{1/\eta}$ units of effective data to be able to produce it.} 

% \question{goes from here}
\paragraph{D4: Cross-task data spillovers.}  A striking example of cross-task spillover is the impressive performance of pre-trained models on multitask language understanding. A single model is capable of answering questions from computer security to business ethics \citep{hendrycks2020measuring}. However, we might expect a large-language model to be fine-tuned on data on how to compose emails to also excel at copywriting, but not at mathematical proofs or steering an autonomous vehicle. %\citep{dwivedi2019representation}. 
This is consistent with recent evidence that shows models finetuned on a classification task exhibits positive transfer within the same `kind' of task (but in out-of-domain data), with more mixed performance across different kinds of tasks \citep{yang2024unveiling}. By varying the network of transfer learning $W$, we will be able to link the speed, direction, and limit of automation under different degrees of transfer learning. 
%\question{to here}

% Panel (b) of \cref{fig:graph_illust} illustrates three of the many possibilities for $g$. If $g$ is `broad' (\textcolor{red}{red}), this means there are ample cross-task spillovers. Hence, $g$ captures the degree to which data in a given domain generalizes to other domains---this is known as \emph{transfer learning} within computer science, and the extent of it remains an empirical question that is widely debated among computer scientists. If it is `narrow' (\textcolor{blue}{blue}), this means that cross-task data spillovers across tasks are local, though they might be quite strong e.g., we might expect a lot of spillovers between email support and chat support tasks. Finally, an extreme case is when $g$ exhibits no spillovers (\textcolor{green!30!black}{green}) which we call \emph{no spillovers}. 
%%%%%%%%%%%%%%%%%%%%%%%%%%%%%%%%%%%

% \paragraph{D5: Richness of tasks and spillovers.} \question{1 things: 1. the simplified geometry is now generalized, so I think D5 should be removed, 2. If we go ahead with breaking the paper, D5 is irrelevant for the paper without spillovers anyways.} Our main model works within a one-dimensional task space $X = [0,1]$ and spillovers modeled via the graphon $W$ induced by $g$ that imposed a sort of `geometry' so nearby tasks exhibit more spillovers. This is largely for expositional reasons that simplifies the proof and statement of our results. We expect analogous results to hold for richer spaces of (high-dimensional) tasks and richer network structures of spillovers.

%%%%%%%%%%%%%%%%%%%%%%%%%%%%%%%

\paragraph{D5: Data-driven technological change.} A distinctive feature of our model is the vector of task-specific data determine equilibrium allocation. It will, however, be helpful to connect our model to more `standard' models of technological change in the Hicksian tradition \citep{hicks1932theory} formalized by \cite{acemoglu2002directed}.

Define the following \emph{capital productivity index}: 
\begin{equation} \label{capital_prod_index}
\Psi_{K}\left(\gamma_{t},\boldsymbol{\psi^K_t}\right):= \left(\frac{1}{\gamma_{t}}\int_{0}^{\gamma_{t}}\big[\psi_{jt}^K\big]^{\sigma-1}dj\right)^{\frac{1}{\sigma-1}}
\end{equation}
 which is a function of the automation boundary $\gamma_t$ and the profile of task productivities $\boldsymbol{\psi^K_t} := (\psi_{it}^K)_i$. By substituting our expressions for capital and labor productions into (\ref{GDP}), output can be written in terms of our capital productivity index as follows: 
\begin{equation} \label{eqn:GDP_KL}
    Y_{t}=\left(\gamma_{t}^{\frac{1}{\sigma}}\left(\Psi_{K}\left(\gamma_{t},\boldsymbol{\psi}_{t}\right)\cdot K\right)^{\frac{\sigma-1}{\sigma}}+\left(1-\gamma_{t}\right)^{\frac{1}{\sigma}}\left(\psi_{L}\cdot L \right)^{\frac{\sigma-1}{\sigma}}\right)^{\frac{\sigma}{\sigma-1}}.
\end{equation}

Labor and capital augmenting technological change correspond to improvements in $\psi_L$ and $\Psi_K$, holding the automation boundary $\gamma_t$ fixed. A key feature of our model is that data accumulation simultaneously (i) increases capital productivities so that \emph{fixing} the automation boundary $\gamma_t$, it is capital-augmenting; and (ii) \emph{endogenously changes} the automation boundary $\gamma_t$ which must equilibrate labor and capital markets. This connects our model to the important work of \cite{acemoglu2018race,aghion2017artificial} who model automation only on the extensive margin (capital productivity is either $0$ or $1$), and with an exogenous movement of the capital boundary given by technology. Our model is also distinct from a recent paper of \cite{jones2024framework} who also allow both `intensive' and `extensive' margins for automation via an R\&D process a la \cite{romer1986increasing}.\footnote{We might view a special case of our model with uniform spillovers i.e., for each stock of data $\pmb{D}$, $\mathcal{A}_i(\pmb{D}) = \mathcal{A}_j(\pmb{D})$ for all $i,j \in X$ as recovering a version of `learning by doing' a la \cite{arrow1962economic,romer1986increasing} in which productivity is a one-dimensional function of past production. By taking heterogeneous and task-specific data seriously, we will move between the `data autarky' (\cref{sec:direction_limit}) with no spillovers, and economies with richer patterns of spillovers (\cref{section:contagious}).} Differently, our model takes seriously recent trends in AI progress that (i) automation is at least, in part, \emph{data-driven}; (ii) data is heterogeneous and \emph{task-specific}, and (iii) data is \emph{produced endogenously} as a by-product of production.

%\footnote{Both \cite{acemoglu2018race,aghion2017artificial} model automation as an exogenous movement of $\gamma_t$, and as capital productivity as either $0$ (cannot be automated) or $1$.}
% Data accumulation mechanically increases task productivities, thereby increasing the productivity index for a $\gamma_t$. In general equilibrium, however, $\gamma_t$ needs to further adjust to equilibrate labor and capital markets. This adjustment both directly manifest in of our output expression and indirectly through the productivity index by changing the composition of automated tasks.

%\bxnote{We are doing this under no data depreciation and no spillovers. Need to generalize. Also, Maryam mentioned it is nicer to write the weak inequality in data. For some reason, with the infimum, it is so annoying to show it with weak inequality, even though it is clearly true. Maybe I am missing something extremely basic. I am writing the strict inequality so I can move on for now...}

\section{Automation in Data Autarky} \label{sec:direction_limit}

We start by considering the case where there are no data spillovers across tasks i.e., $\mathcal{A}_i(\pmb{D}) = D_i$ for all $i \in X$.\footnote{Although the data-autarky case is technically not nested within our graphon framework, it can be viewed as the limit of local averaging kernels: for
\(\epsilon>0\), define $W^\epsilon(i,j):=
\frac{\mathbf 1\{|i-j|\le \epsilon\}}
{\mu\{k\in[0,1]: |i-k|\le \epsilon\}}$.
Then for every \(D\in L^1([0,1])\), $\mathcal A_i^\epsilon(\bm D)\to D_i$ as $\epsilon\downarrow 0$ for a.e. $i \in [0,1]$. Thus \(\mathcal A_i(\bm D)=D_i\) is the local-kernel limit of the aggregators
\(\mathcal A_i^\epsilon\).} 
\cref{lemma:static eq lemma}  characterizes the static equilibrium---the automation boundary, factor prices, and allocations---and establishes that the equilibrium path is  regular in the sense that at each time $t \geq 0$ task capital productivities are indeed decreasing in task index.

\begin{lemma}[Static equilibrium characterization and regular paths] \label{lemma:static eq lemma}
    Suppose that there are no data spillovers. If capital productivity is weakly decreasing in task index, then:
    \begin{itemize}
        \item[(i)] \textbf{Automation boundary.} There exists a unique task $\gamma_t \in (0, 1)$ satisfying Equation \eqref{gamma def} such that tasks in $[0, \gamma_t]$ produce with capital and tasks in $(\gamma_t, 1]$ produce with labor.
        \begin{equation} \label{gamma def}
\psi_{\gamma_{t}}^{K}=\psi^{L}\underbrace{\left(\frac{\int_{0}^{\gamma_{t}}\left(\psi_{j}^{K}\right)^{\sigma-1}dj}{\left(1-\gamma_{t}\right)\left(\psi^{L}\right)^{\sigma-1}}\frac{L}{K}\right)^{\frac{1}{\sigma}}}_{r/w}
\end{equation}
        \item[(ii)] \textbf{Factor prices and allocations.} Capital and labor allocations are given by Equations \eqref{eqn:capital allocation} and \eqref{eqn:labor allocation}:
            \begin{align} 
        K_{i}& =\frac{\left(\psi_{i}^{K}\right)^{\sigma-1}}{\int_{0}^{\gamma_{t}}\left(\psi_{j}^{K}\right)^{\sigma-1}dj}K \quad \text{for all $i \in [0, \gamma_t]$} \label{eqn:capital allocation}\\
        L_{i}&=\frac{L}{1-\gamma_{t}} \quad \text{for all } i \in (\gamma_t, 1]
        \label{eqn:labor allocation}
        \end{align}
         Rental rate of capital and wages are given by Equations \eqref{eqn:rental rate} and \eqref{eqn:wage}:
    \begin{align} 
        r &= Y^{\frac{1}{\sigma}} \left(\frac{\int_{0}^{\gamma_{t}}\left(\psi_{j}^{K}\right)^{\sigma-1}dj}{K}\right)^{\frac{1}{\sigma}} \label{eqn:rental rate} \\
        w &= Y^{\frac{1}{\sigma}} \left(\frac{\left(1-\gamma_{t}\right)\left(\psi^{L}\right)^{\sigma-1}}{L}\right)^{\frac{1}{\sigma}}. \label{eqn:wage}
        \end{align}
    \end{itemize}  
    If, in addition, the initial data stock $\bm{D}_0$ is continuous and strictly decreasing in task index, then there exists an equilibrium path such that for all times $t$, capital productivity is weakly decreasing in task index. 
\end{lemma}

%\begin{lemma}[Existence of a regular equilibrium path] \label{lemma:mono lemma}
%    If the initial data stock $\bm{D}_0$ is continuous and strictly decreasing in task index, then there exists an equilibrium path such that for all times $t$, capital productivity is weakly decreasing in task index. 
%\end{lemma}

Note that these equilibrium allocations and factor price characterizations apply in more general settings when monotonicity of capital productivity across tasks may not hold. In such cases, simply replace the integrating region with the set of tasks that are cheaper to produce by capital.
We next characterize how the extent of automation in this economy changes over time. 
\begin{lemma}[Direction of the automation boundary] \label{lemma:direction lemma}
    The automation boundary, $\gamma_t$, is weakly increasing in time if and only if the following inequality holds
    \begin{equation} \label{eq:gammatPrelim}
     \left(\sigma-1\right)\int_{0}^{\gamma_{t}}\left(\frac{f\left(j\right)}{f\left(\gamma_{t}\right)}\right)^{1/\eta}\left(\frac{\psi_{j}^{K}}{\psi_{\gamma_{t}}^{K}}\right)^{2\sigma-1-1/\eta}dj
     \leq
     \sigma\int_{0}^{\gamma_{t}}\left(\frac{\psi_{j}^{K}}{\psi_{\gamma_{t}}^{K}}\right)^{\sigma-1}dj    \end{equation}
\end{lemma}

Notice that Equation \eqref{eq:gammatPrelim} is more slack for small $\sigma$. In particular, it holds trivially when $\sigma \leq 1$ as the left hand side is weakly negative while the right hand side is positive. In other words, if tasks are gross complements, the automation boundary moves to the right monotonically and more tasks are automated over time. 

To gain some preliminary intuition for this result, start from a given automation boundary $\hat \gamma$. Let time progress bur hold fixed the automation boundary. Data accumulation acts as a purely capital augmenting technological change, so capital becomes relatively cheaper compared to labor due to standard Baumol forces when tasks are gross complements \citep{baumol1967macroeconomics,nordhaus2008baumol}.
Consequently, capital productivity of the cutoff task must decrease so as to equilibrate (recall $r/w=\psi^K_\gamma/\psi^L$). On the other hand, we know data accumulation improves capital productivity of all tasks over time, so the only way for this relation to hold is for $\gamma_t$ to increase, as higher indexed tasks are less productive for capital.

In Subsections \ref{subsec:subsL} and \ref{subsec:subsH} we leverage \cref{lemma:direction lemma} to precisely characterize the evolution of the automation boundary, before discussing the implications of data-driven automation for equilibrium wages in Subsection \ref{subsec:wages}.

\subsection{Low Substitutability Across Tasks} \label{subsec:subsL} 

Let us first analyze the case where the elasticity of substitution across tasks $\sigma$ is smaller than the reciprocal of the degree of diminishing returns to data, $1/\eta$. Note that $1 / \eta > 1$ so this nests both the case where tasks are gross complements and the case where tasks are gross substitutes but not too similar. 

Our first result is that in this case, full automation is inevitable---for any initial condition on data $\pmb{D_0}$ and any heterogeneity in how tasks might differentially benefit from data $\pmb{f}$, so long as \cref{ass:DRS} holds, the economy is always automated in the limit. 

\begin{proposition}[Full automation and balanced data] \label{prop:complements}
    If $\sigma \leq 1/\eta$: %then for any pattern of spillovers $W \in \mathcal{W}$:
    \begin{enumerate}
        \item[(i)] \textbf{Full automation}: The economy is fully automated in the limit, that is, $\lim_{t \to \infty} \gamma_t = 1$. If tasks are additionally gross complements (i.e., $\sigma \leq 1$), then automation increases monotonically.
        \item[(ii)] \textbf{Balanced data}: For any $i,j \in \text{Int}(X)$ where $i < j$:
        \[
        \lim_{t \to \infty} \cfrac{D_{it}}{D_{jt}} = \left(\cfrac{f_i}{f_j}\right)^{\frac{\sigma}{1- \sigma\eta}}.
        \]
    \end{enumerate}
    %for any stock of initial data stock $\bm{D}_0$ that is continuous and strictly decreasing in task index [[I NOW MAKE THIS A MAINTAINED ASSUMPTION]]
    % \begin{itemize}
    %     \item[(i)] \emph{Automation monotonically increases:} $\gamma_t$ increases with time; and 
    %     \item[(ii)] \emph{Full limit automation:} $\lim_{t \to \infty} \gamma_t = 1$.
    % \end{itemize}  
\end{proposition}

%\begin{proof}
%   See \cref{sec:gross comp pf}.
%\end{proof} 

The intuition behind \cref{prop:complements} is quite different for the gross complements case $(\sigma \leq 1)$ and for the gross substitutes case with $(1 \leq \sigma \leq 1/\eta)$. We discuss each in turn.

First consider the case when tasks are gross complements $(\sigma \leq 1)$. Towards a contradiction, suppose automation is upper-bounded by some level $\bar \gamma \in \text{Int}(X)$. This implies data accumulation must become purely capital augmenting after some large time. By the usual Baumol logic, the capital share of the economy as well as the relative price of capital to labor will both collapse to zero. On the other hand, we know data for all tasks grows unboundedly so capital becomes infinitely productive which yields the contradiction 
\[
\underbrace{\psi^K_{\bar \gamma}}_{\to \infty} = \underbrace{(r_t/w_t)}_{\substack{\to 0 \\ \text{(Baumol)}}} \cdot \psi^L
\]
since the per unit cost for producing any task $i \in \text{Int}(X)$ cannot equilibrate between capital and labor. Indeed, the logic of this argument is quite general and applies more broadly: 
\begin{remark} \label{remark:porp 1}
    \textit{When tasks are gross complements $(\sigma \leq 1)$ our economy is fully automated in the limit for any functional form for capital productivity $\psi^K_i: \mathbb{R}_+ \to \mathbb{R}_+$ that grows unboundedly as (effective) data grows unboundedly.}
\end{remark}

\begin{figure}[t]
    \centering
    \subfloat[capital and labor\label{fig:shares_complements}]{
        \includegraphics[width=0.48\textwidth]{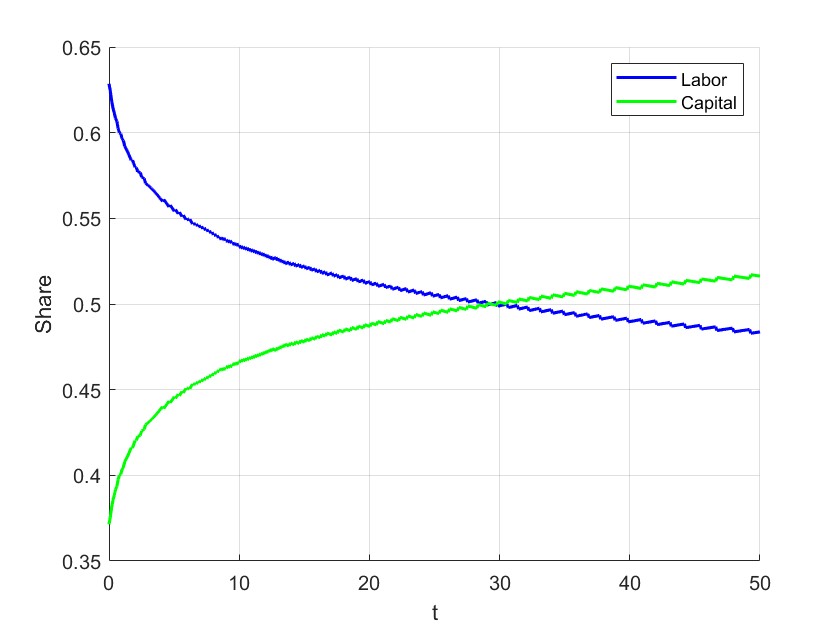}}
    \hfill  
        \subfloat[Automation boundary \label{fig:psi_vs_ratio_A}]{
        \includegraphics[width=0.45\textwidth]{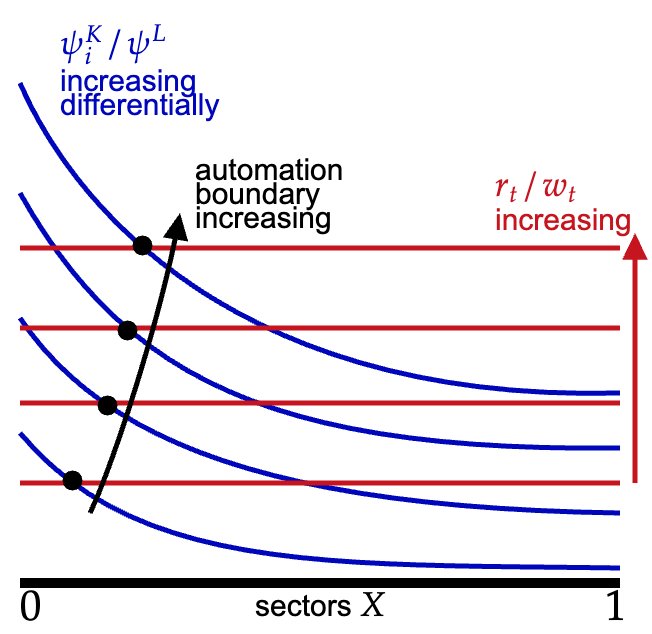}%
    }
       \caption{Dynamics of equilibrium outcomes when the degree of substitutability across tasks is low ($\sigma < 1 / \eta$). In this case, capital allocation across tasks is biased toward data-poor tasks.}
   \label{fig:capitalLaborShare}
\end{figure}

In addition to the economy being fully automated in the limit, the capital share will also converges to $1$. This is the result of a tug of war between intensive capital productivity improvements and the extensive expansion of the automation boundary. Whereas the former drives the usual Baumol cost disease, shrinking the size of the economy attributable to capital, the latter directly enlarges the capital share which is most readily seen from \cref{eqn:GDP_KL}. The previous argument implies that automation dominates the  Baumol effect in the limit; Figure \ref{fig:shares_complements} illustrates numerically that this is also the case along the transition. Intuitively, automation occurs quickly enough because capital productivity growth in the marginal task is keeping pace with growth among inframarginal tasks. At each time $t$, capital allocation is \emph{biased towards data-poor tasks} (Figure \ref{fig:complements_alloc}), partially offsetting inequities in task production due to current levels of capital productivity and resulting in a balanced data path in the long run.

\begin{figure}[t]
    \centering
    \subfloat[Capital biased towards data-poor tasks\label{fig:complements_alloc}]{%
        \includegraphics[width=0.5\textwidth]{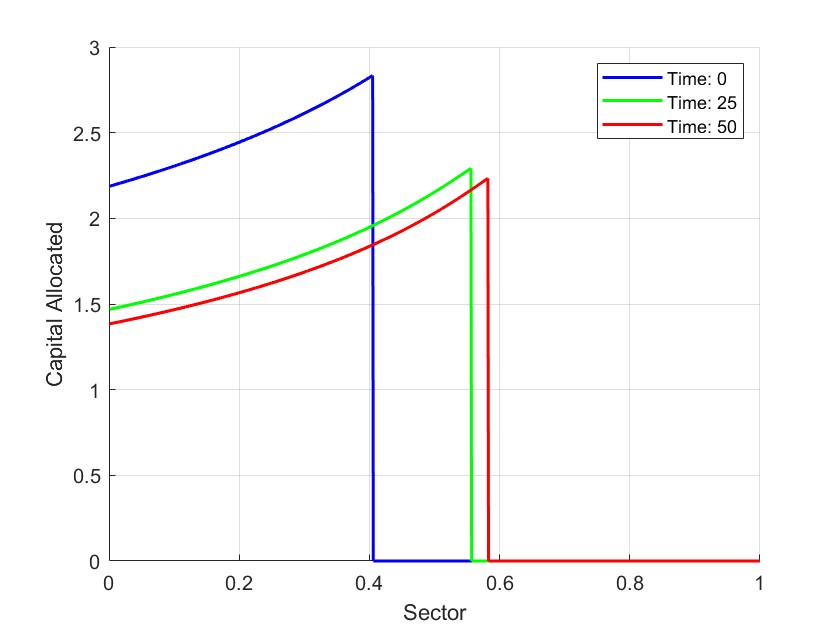}%
    }
    \hfill
    \subfloat[Data stock for different tasks \label{fig:complements_abs}]{%
        \includegraphics[width=0.5\textwidth]{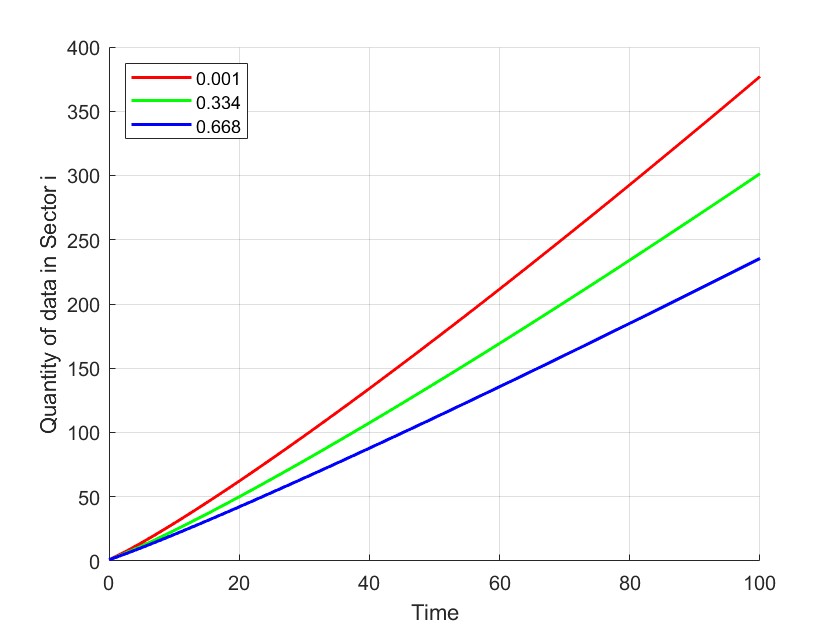}%
    }
        \caption{Balanced data growth}
                \footnotesize
        Note: We set $\psi^L=L=K=1$, $\boldsymbol{D}_0 := \bm{1}$, $f(i) = 1 - i$, $\eta = 0.2$, $g=\delta_{\{0\}}$, $\sigma = 0.5$. \\ Details of procedure in \cref{appendix:simulation_details}.     \label{fig:complements}
\end{figure}

%\begin{figure}[t]
%\begin{minipage}[t]{0.45\linewidth}
%Indeed, this argument also implies that automation must come to dominate the standard logic of cost-disease in the long run: automation is always taking place fast enough to offset Baumol effects. This is illustrated numerically in \cref{fig:shares_complements} where we find an increasing and concave pattern for the growth of capital share. 
%\end{minipage}
%\hfill%
%\begin{minipage}[t]{0.55\textwidth}\vspace{0pt}
%\vspace{-0.5em}
%\centering
%{\includegraphics[width=0.95\textwidth]{draft_figures_new/sigma0p5_shares.jpg}}
%\vspace{-1em}\caption{}\label{fig:shares_complements}
%\end{figure}
%\end{minipage}
%\end{figure}

%In simulations, we find the relative data ratio of tasks to be roughly constant over time, suggesting a constant rate of accumulation for all tasks. Such a {\emph{balanced data equilibrium path}} is a direct consequence of the complementary nature of tasks: at any point in time, capital is allocated to the least productive tasks [[cite that panel]]. Were this not to be the case, data will become imbalanced as the economy concentrates its capital resources towards a smaller set of data-rich tasks where capital is most productive already. \\ \\ 

Next, consider the case in which tasks are gross substitutes, but smaller than the reciprocal of the degree of diminishing returns of data in capital productivity $(\sigma \in [1,1/\eta))$. This case is more nuanced since the direction of capital allocation flips towards more productive tasks, so data growth might no longer be balanced---in principle, data-rich tasks can attract more capital, accumulate more data, and thereby attracting even more capital such that automation is plausibly constrained to a small set of highly capital productive tasks. We show that this cannot happen as the data feedback loop is dominated by diminishing returns from data: suppose again that the automation boundary is interior in the limit. Then, the ratio of productivity between the most and least productive task is \emph{bounded} due to diminishing returns from data accumulation since $\eta < 1$, which immediately implies that the relative price of capital will remain finite as well. A similar argument therefore yields
\[
\underbrace{\psi^K_{\gamma_t}}_{\substack{\to +\infty}} = \underbrace{(r_t/w_t)}_{\substack{\text{$\to C < +\infty$}}} \cdot \psi^L
\]
yielding the same contradiction as before. The proof in \cref{appendix:proofs} makes this heuristic argument formal. Moreover, like the gross complements case, full limit automation implies production will become entirely capital driven. In this case, however, extensive automation and intensive capital improvements both push toward a larger capital share given that tasks are gross substitutes.

\cref{fig:psi_vs_ratio_A} illustrates the evolution of automation boundary in this case. Red lines represent the relative factor price of capital and blue curves reflect the cross-sectional distribution of task productivities. The important observation is that the automation boundary $\gamma_t$ monotonically increases, until the economy is fully automated in the limit. In this limit, capital share converges to unity. This implies that when tasks are gross complements or weak substitutes, the expansion of the automation boundary dominates the Baumol effect that stems from capital augmenting productivity and drives the growth in capital share.
 
% \question{I put all the way implications in one subsection}

%Balanced data accumulation, in turn, ensures that growth in task-specific capital productivities are not excessively biased toward low index tasks (easy to automate). This, in turn, ensures that even though the ratio of rental rate to wages $r_t/w_t$ is rising, [XYZ]

%\begin{figure}[t]
 %   \centering
 %       \includegraphics[width=0.6\textwidth]{Images/psi_vs_ratio_A.png}%
  %      \caption{Capital allocation biased toward data-poor}\label{fig:psi_vs_ratio_A}
%\end{figure}

% \cref{fig:psi_vs_ratio_A} illustrates the evolution of the cross-sectional distribution of productivities $(\psi_{it})_i$ over different periods (\textcolor{blue}{blue curves}). Although tasks accumulate data at different speeds (reflected in the curvature of $(\psi_{it})_i$), the condition $\sigma \leq 1/\eta$ guarantees that the automation boundary increases toward $1$ even as the relative price of capital $r_t/w_t$ is simultaneously rising. 

\subsection{High Substitutability Across Tasks} \label{subsec:subsH}

We now show that the automation boundary is tight, i.e., when $\sigma > 1/\eta$, automation can fizzle out. In this case, the economy is not fully automated in the limit and the growth path of data remains imbalanced. Of course, this can only happen if tasks and their initial data allocation are not ex-ante symmetric. Otherwise, since the CES aggregator is symmetric, if both initial data $\pmb{D}_{i0}$ and task-specific data productivities $f_i$ are constant across tasks, all tasks produce the same quantity along this `knife edge' equilibrium path.\footnote{No matter how we break ties in order to choose which tasks to automate and which tasks to produce with labor.} 

\cref{fform substitutes2} provides a sufficient condition for data growth to be imbalanced when the tasks are not ex-ante perfectly symmetric. This sufficient condition qualifies the degree of heterogeneity in task-specific data productivity, $f_i$. 

\begin{definition}[$\sigma$-regular] \label{fform substitutes2}
    $\pmb{f}: X \to \mathbb{R}_{> 0}$ is \emph{$\sigma$-regular} if there exists 
     $L_{\sigma, f} \in \left(0, 1\right)$ s.t.
    \[
    \frac{\sigma}{\sigma-1}f(i)
    < \frac{1}{i}\int_{0}^{i}f(j)dj \quad \text{for all $i \geq L_{\sigma, f}$.} \label{f inequality}
    \]
    \begin{comment}
    \begin{enumerate}
        \item $f(i) > 0 \text{ }\forall i \in [0, 1)$ and $f(1) \geq 0$. [[I DO NOT THINK WE NEED THIS; IF IT IS NOT FULFILLED THEN WE GET THAT AUTOMATION IS UPPER-BOUNDED FOR FREE]]
        \item There exists
    \end{enumerate}
                \end{comment}
    Let $\mathcal{F}(\sigma)$ be the set of $\pmb{f}$ that are $\sigma$-regular. 
\end{definition}

$\sigma$-regularity states that there must be sufficient heterogeneity in how efficiently different tasks use data. It essentially requires the reverse hazard rate of $f_i$ to fall faster than the reciprocal of the task index $i^{-1}$: 
\[
\underbrace{\cfrac{f(i)}{\int^i_0 f(j) dj}}_{\substack{\text{`reverse} \\ \text{hazard rate'}}} < \frac{1}{i}\cdot \frac{\sigma - 1}{\sigma}
\]
This condition guarantees enough heterogeneity in how tasks use data so that the positive data feedback loop mechanism can take hold. Note that for a fixed  vector $\pmb{f}$, increasing task substitutability makes the $\sigma$-regularity condition more likely to be satisfied. Indeed, in the limit where $\sigma \to \infty$, even a near-constant $\pmb{f}$ will result in interior automation. \cref{example:sigma-reg} illustrates an instance of a $\sigma$-regular vector $\pmb{f}$.

\begin{example}[Example of $\sigma$-regularity] \label{example:sigma-reg}
    An example parameterization of $\pmb{f}$ is $f(i) = 1 + \epsilon - i$ for any $\epsilon < (\sigma - 1)/2 $. In this case, the condition for $\sigma$-regularity becomes
    \[
        \frac{\sigma}{\sigma - 1} \left(1 + \epsilon - j \right) < 1 + \epsilon - \frac{1}{2}j
    \]
    so any $L_{\sigma, f} \in \left(\frac{2 + 2\epsilon}{1 + \sigma}, 1\right)$ suffices.
\end{example}
With the definition of $\sigma$-regularity at hand, we are ready to state our second main result, pertaining to bounded automation.

\begin{proposition}[Bounded automation and imbalanced data] \label{prop:substitutes} 
If $\sigma > 1/\eta$ and $\pmb{f}$ is $\sigma$-regular, then: 
\begin{itemize}
    \item[(i)] \textbf{Partial automation:} The economy is never fully automated: $\lim_{t \to \infty} \gamma_t < 1$; 
    
    \item[(ii)] \textbf{Imbalanced data:} For any task $j \in \text{Int}(X)$ that is either always produced using capital or always produced using labor after some finite time: 
    \[
    \lim_{t \to \infty} \cfrac{D_{0t}}{D_{jt}} = \infty.
    \]
    % \textcolor{blue}{BX: Bryant check if this can be relaxed just to any task $j\neq0$.}
\end{itemize}
\end{proposition}

\begin{comment}
    Old version: If $\sigma \geq 1/\eta$ and $\pmb{f}$ is regular, then for any initial stock of data $\bm{D}_0$ weakly decreasing, if there are no data spillovers i.e., $W$ is induced by $g = \delta_{\{0\}}$, \emph{automation is upper-bounded}:  
    \[
        \sup_{t \in \mathcal{T}} \gamma_t \leq \max \big\{L_f, \gamma_0\big\} < 1 \text{ }\forall t.
    \]
    Moreover, the condition on the elasticity is tight: if $\sigma < 1/\eta$ and $\pmb{f}$ is constant, then for any any initial stock of data $\bm{D}_0$ and any spillovers $W \in \mathcal{W}$, (i) \emph{automation monotonically increases} $\gamma_t$ increases; and (ii) there is \emph{full limit automation}: $\lim_{t \to \infty} \gamma_t = 1$.
\end{comment}

%\begin{proof}
 %   See \cref{sec:substitutes ub pf}.
%\end{proof}

\begin{figure}[t]
    \centering
    \subfloat[Capital biased to data-rich \label{fig:leftComp}]{%
        \includegraphics[width=0.5\textwidth]{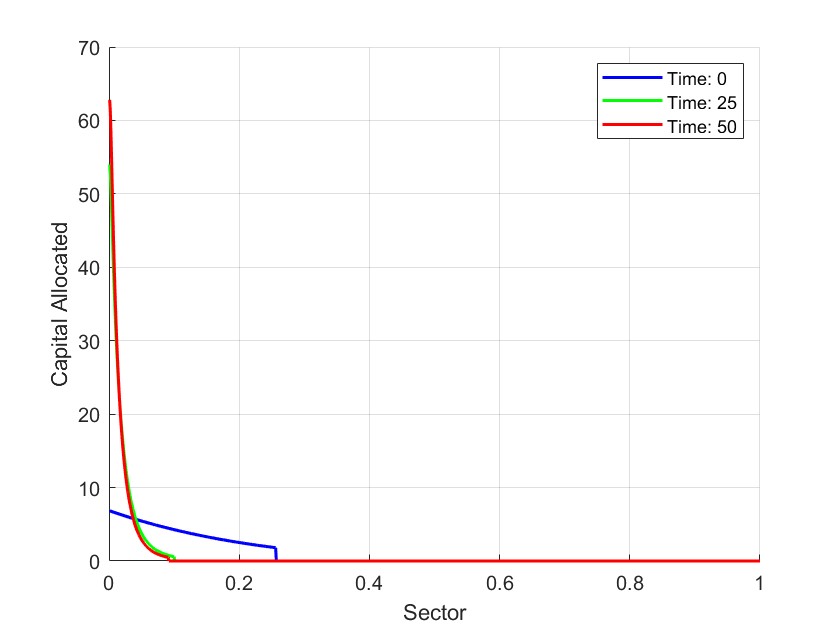}%
    }
    \hfill % Pushes the subfigures apart
    \subfloat[Data stock for different tasks \label{fig:rightComp}]{%
        \includegraphics[width=0.5\textwidth]{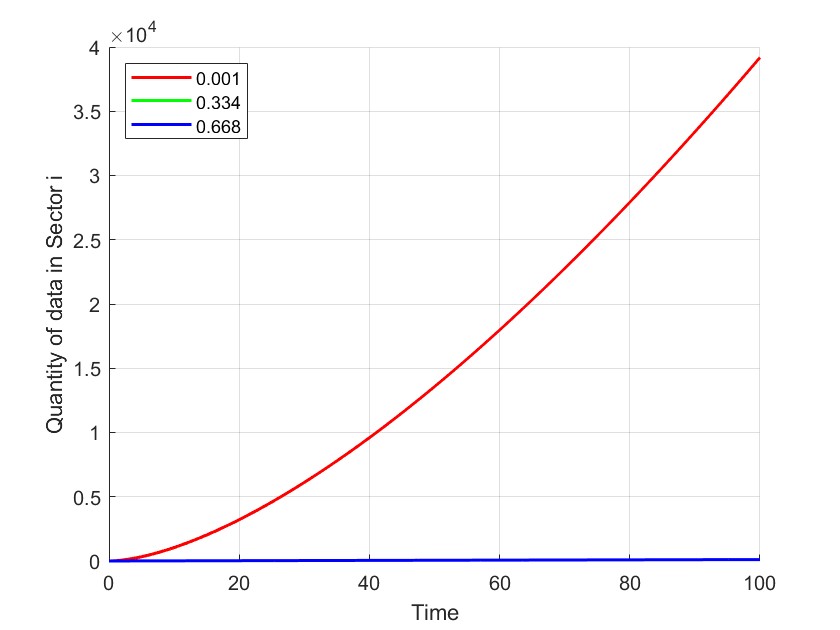}%
    }
        \caption{Imbalanced data growth}
 % Add some vertical space between caption and note
        \footnotesize
        Note: parameters same as \cref{fig:complements} but with $\sigma = 5.5$.
    \label{fig:substitutes}
\end{figure}

\cref{fig:substitutes} illustrates \cref{prop:substitutes}. Panel \ref{fig:leftComp}  shows the evolution of capital allocation over time. There is  a bias in capital allocation towards data-rich tasks, which in turn leads to a contraction of the automation frontier. Panel \ref{fig:rightComp} plots the evolution of effective data stock of various tasks and illustrates that data becomes highly concentrated in a few capital intensive tasks. 

The key intuition behind this result is that sufficient substitutability distorts the direction of capital allocation toward tasks that have high capital productivity. This, in turn, induces greater production of those tasks in general equilibrium, generating more task-specific data and amplifying the capital productivities of those tasks even further. This force pushes our economy toward \emph{imbalanced data accumulation} and \emph{partial automation}.

To see this, consider a similar thought experiment as in the previous section, where we run the economy forward in time,  holding the automation boundary fixed. The relative rental rate can be written as: 
\begin{equation} \label{eqn:r_over_w_alt}
\frac{r_t}{w_t} = \cfrac{L}{K} \cdot \cfrac{\psi^L}{1-\gamma_t} \cdot \int^{\gamma_t}_0 \left(\cfrac{\psi^K_j}{\psi^K_{\gamma_t}}\right)^{\overbrace{\sigma - 1}^{\textcolor{red}{>1}}} dj = \frac{\psi^K_{\gamma_t}}{\psi^L}.
\end{equation}
That is, the rental rate of capital relative to labor is a `weighted average' of capital productivities among the inframarginally automated tasks. Holding the level of automation fixed, in this case we show the relative productivity growth of inframarginal tasks outpace that of the marginal task, a direct consequence of capital allocation being sufficiently biased towards data-rich tasks, so much so that it overcomes diminishing returns from data accumulation.\footnote{Technically this only occurs in the regime where $\gamma_t \geq L_{\sigma, f}$. See \cref{appendix:proofs} for additional details.} As such, $\gamma_t$ would have to decrease to match the increased productivity of inframarginal tasks.

\begin{figure}[t]
    \centering
        \subfloat[Capital and labor \label{fig:shares_sub}]{
        \includegraphics[width=0.48\textwidth]{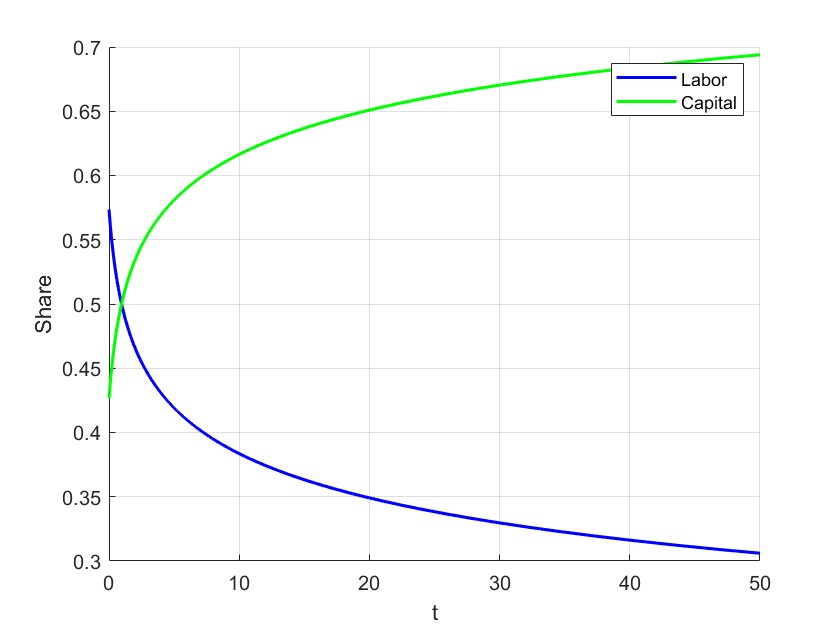}}
\hfill
    \subfloat[Automation boundary\label{fig:psi_vs_ratio_B}]{%
        \includegraphics[width=0.4\textwidth]{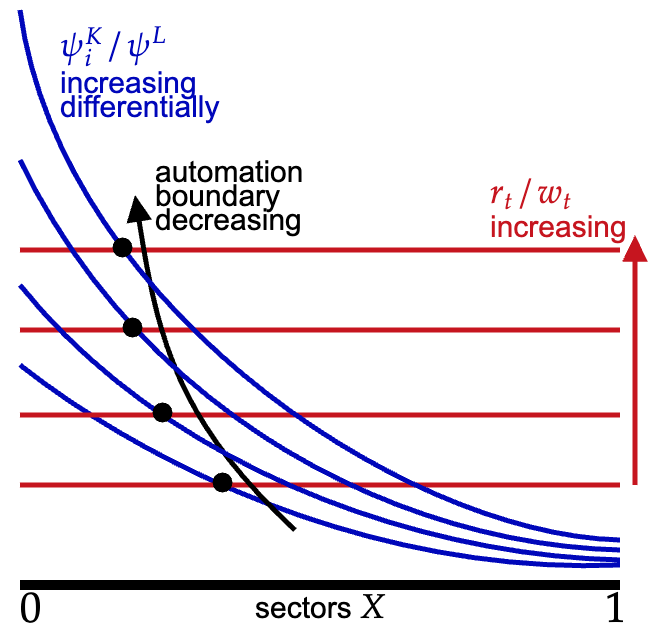}}

        \caption{Dynamics of equilibrium outcomes when tasks are highly substitutable ($\sigma >\frac{1}{\eta}$). In this case, capital allocation across tasks is biased toward data-rich tasks.}
   \label{fig:automationBoundary}
\end{figure}

\cref{fig:psi_vs_ratio_B} depicts the evolution of automation boundary when the degree of substitutability among tasks is sufficiently high, $\sigma>\frac{1}{\eta}$. Similar to the case of $\sigma< \frac{1}{\eta}$, red lines (blue curves) represent the relative factor price of capital (cross-sectional distribution of task productivities). In this case, the automation boundary $\gamma_t$ must decline, for factor markets to clear, and the upper bound of automation boundary represent the automation level at which this occurs. As in the other case (gross complements and weak gross substitutes), capital share converges to one in the limit since $\lim_{t\to \infty} \inf_{\gamma \in [0,1]} \psi^K_{\gamma,t} = \infty$. We further verify numerically that the capital share increases along the transition path in \cref{fig:shares_sub}. The driving force of the growth of capital share, however, is distinct from the low substitutability case. Since tasks are now gross substitutes, an increase in the capital index increases the capital share, whereas any resulting contraction in the automation boundary $\gamma_t$ tames the initial increase. While in the gross complements case, the expansion of the automation boundary dominates the Baumol effect, here the increase in capital share suggests the opposite must be true. 

\begin{figure}[t] 
    \centering
    \subfloat[Evolution of automation boundary $\gamma$ \label{fig:leftCompDyn}]{%
        \includegraphics[width=0.5\textwidth]{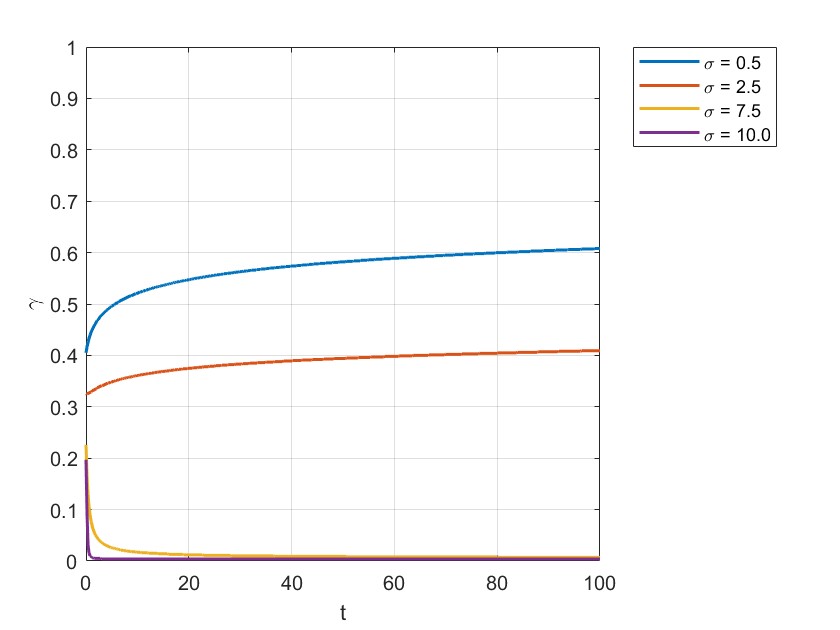}%
    }
    \hfill % Pushes the subfigures apart
    \subfloat[Data of task $0.5$ relative to task $0$\label{fig:rightCompDyn}]{%
        \includegraphics[width=0.5\textwidth]{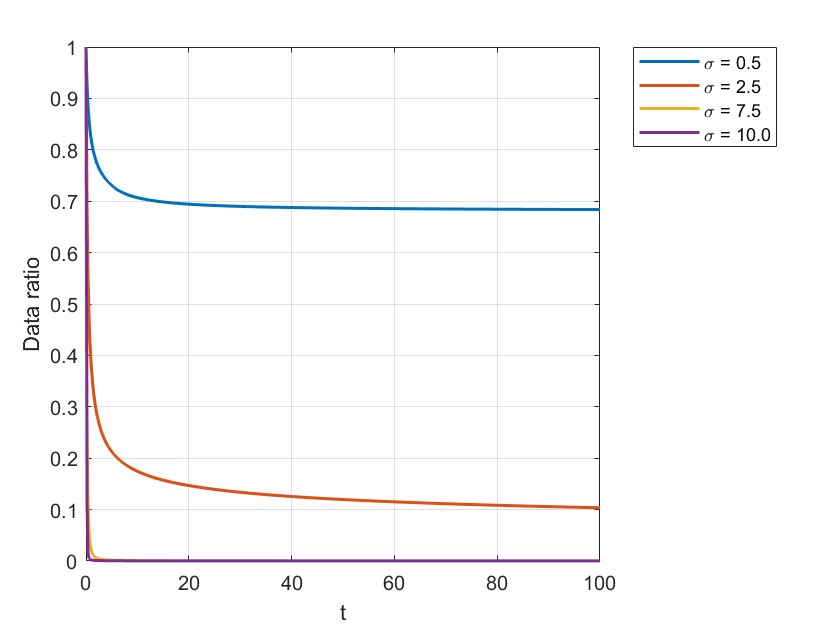}%
    }
    \caption{Comparative dynamics over time for different values of $\sigma$, when $\eta=0.2$.}
    \label{fig:comparativeDynamics}
\end{figure}

\cref{fig:comparativeDynamics} summarizes \cref{prop:complements,prop:substitutes}. Panel \ref{fig:leftCompDyn} plots the evolution of the automation boundary, while Panel \ref{fig:rightCompDyn} plots the evolution of data used in task $0.5$ (the task with median productivity) relative to task $0$ (most productive task) over time, each for four  different degrees of cross-task substitutability $\sigma$. We assume $\eta = 0.2$ such that $\sigma = 1/\eta = 5$ is the threshold that deternines whether teh economy is fully automated or not.

\subsection{Wages} \label{subsec:wages}

Taken together, our results on how task substitutability shapes the direction and limit of equilibrium automation has sharp implications for wages. When tasks are  not strong enough substitutes (equivalently, sufficiently complementary) the expression of equilibrium wage in \cref{eqn:wage} illustrates that there are two distinct effects of data-driven automation on the level of wages as in \cite{acemoglu2018race}. On the one hand, data-driven increases in capital productivity drives up the productivity of labor, thereby increasing wages. On the other hand, the contraction of labor onto a smaller set of tasks increases the effective labor supply, thereby reducing wages. The latter displacement effect is identical to \cite{acemoglu2018race} but the former productivity effect stems from the role of data in augmenting capital productivity on the intensive margin rather than via the extensive margin by replacing inefficient labor. In the long-run they cancel out \emph{exactly} such that wages eventually stagnate. %---we will formalize this momentarily in \cref{prop:wages}. 

Alternatively, when tasks are sufficiently substitutable, \cref{prop:substitutes} implies that the automation boundary can be bounded above, thereby neutering the displacement effect that had offset wage growth in the complements case. As such, long-run wage growth will now be dominated by the productivity effect from data-driven automation.  
The following proposition formalizes this discussion.

\begin{proposition}[Long-run wages] \label{prop:wages} For any initial stock of data $\pmb{D}_0$ consistent with the hypotheses of \cref{lemma:static eq lemma}: 
    \begin{itemize}
        \item[(i)] If $\sigma \leq 1/ \eta$, then wages stagnate: 
        \[
        \lim_{t \to \infty} w_t = C < +\infty. 
        \]
        \item[(ii)] If $\sigma > 1/\eta$ and $\pmb{f}$ is $\sigma$-regular, then wages grow unboundedly: 
        \[
        \lim_{t \to \infty} w_t = +\infty. 
        \]
    \end{itemize}
\end{proposition}

\cref{prop:wages} is tightly connected to the literature on how automation technologies impact wages in opposing direction. On the one hand, they can increase wages--by increasing the marginal productivity of labor. On the other hand, they can decrease wages by displacing the equilibrium set of tasks performed by labor (see, e.g., Propositions 1 and 6 of \cite{acemoglu2024capital}). \cref{prop:wages} makes precise how the endogenous composition of data accumulation and the ensuing direction of data-driven automation determine which of these forces dominate and when they offset each other exactly.\footnote{This is another way in which data accumulation is distinct from capital accumulation: supposing our economy is fully automated in the limit $(\lim_{t} \gamma_t = 1)$, long-run wages are, perhaps unsurprisingly, proportional to $K^{\frac{\sigma - 1}{\sigma}}$ which is strictly decreasing in $K$ whenever tasks are strictly gross substitutes $(\sigma < 1)$, and strictly increasing in $K$.} Indeed, we emphasize that \cref{prop:heterogeneous-speed} is driven distinctly by endogenous data accumulation---and hence capital productivity: it is precisely when tasks are \emph{complementary} that long-run wages stagnate because it drives the economy to produce more 'data-poor' tasks (e.g., that produced by labor). While this might aid wages in the short-run, this eventually augments the productivity of AI to perform such tasks. It is this sense that labor is \emph{training their own replacement}. 

Finally, it is worth noting that the wage stagnation result will hold when there is capital accumulation as well, as highlighted in the proof of \cref{prop:wages}. Contemporaneous work by \cite{restrepo2025we} similarly shows that wage will stagnate in an economy featuring capital accumulation but with fixed capital productivities across tasks. Allowing capital productivities to expand across tasks should increase the relative price of labor under gross complementarity. However, as argued earlier, the resulting displacement is sufficiently strong in our setting so as to keep wages constant. In this sense our result can be viewed as a strengthening of the wage stagnation proposition in \cite{restrepo2025we}. 

%Finally, our numerical simulations reveal an interesting observation. We find that not only is the automation boundary bounded above, as \cref{prop:substitutes} shows analytically, but it also declines monotonically over time. Hence, equilibrium contraction in automation boundary boosts wage growth further; see \cref{appendix:additional_numerical_results}.

\subsection{Speed of Automation} 

We have thus far analyzed the limiting behavior of data-automation feedback loops. This raises a natural question of how \emph{quickly} data-driven automation might take place. We offer a tight characterization: 

\begin{proposition}[Speed of data-driven automation]
\label{prop:heterogeneous-speed}
Suppose that $\sigma < 1/\eta$ and normalize \(\psi^L=1\).
Then there exists some finite time $T < +\infty$ such that for all times $t \geq T$: 
\begin{align}
\label{eqn:speed_bounds}
\frac{L\underline M}{K\overline f}
    \left(
        \overline B^{1-\eta}
        +(1-\eta)\frac{K\overline f}{\underline M}(t-T)
    \right)^{-\frac{\eta}{1-\eta}}
    \leq 
    1 - \gamma_t
    \leq
    \frac{L\overline M}{K\underline f}
    \left(
        \underline B^{1-\eta}
        +(1-\eta)\frac{K\underline f}{\overline M}(t-T)
    \right)^{-\frac{\eta}{1-\eta}} \tag{S} 
\end{align}
where $\underline M, \overline M, \underline B, \overline B > 0$ are constants, $\underline f := \min \pmb{f}$, and $\overline f := \max \pmb{f}$. 
Importantly,  the fraction of tasks produced by labor decays as a power law, $1 - \gamma_t=\Theta\left(t^{-\frac{\eta}{1-\eta}}\right)$. Moreover, if $\pmb{D}_0$ is less imbalanced than the the balanced-data limit, i.e.,
\[
    \frac{D_{i0}}{D_{j0}}
    \leq
    \left(\frac{f_i}{f_j}\right)^{\frac{\sigma}{1-\sigma\eta}}
    \qquad \forall i<j,
\]
then the above bounds on speed \eqref{eqn:speed_bounds} holds for all times $t \geq 0$ (i.e., setting $T = 0$).
\end{proposition}

\cref{prop:heterogeneous-speed} states that even when data-driven automation leads to full limit automation, equilibrium automation can be surprisingly slow---the share of tasks produced by labor decays as a power law in time. 

Note that the first part of \cref{prop:heterogeneous-speed} is stated as an asymptotic result, so it holds for any initial data stock. Alternatively, the second part imposes an additional condition on the initial data stock, that it is not more imbalanced than the long-run composition (as given by \cref{prop:complements}), which then delivers the same bounds on equilibrium non-asymptotically i.e., for all times $t \geq 0$.

This allows us to sharpen our analysis of the speed of equilibrium automation
along the transition path. The bound on the process in \eqref{eqn:speed_bounds} has the generic form
\[
    1-\gamma_t
    \simeq
    A\left(B^{1-\eta}+Ct\right)^{-\frac{\eta}{1-\eta}}
    =
    AB^{-\eta}
    \left(1+\frac{t}{\tau}\right)^{-\frac{\eta}{1-\eta}} 
    \quad 
    \text{where 
    $\tau:=\frac{B^{1-\eta}}{C}$
    }
\]
and \(A,B,C>0\) are constants independent of time. Here \(B\) can be
interpreted as the relevant local data stock at date \(0\), while \(C\)
summarizes the effective rate at which accumulated production expands the data
stock near the automation boundary.\footnote{For the lower bound on \(1-\gamma_t\),
\(B=\overline B\) and \(C=(1-\eta)K\overline f/\underline M\). For the upper
bound on \(1-\gamma_t\), \(B=\underline B\) and
\(C=(1-\eta)K\underline f/\overline M\).}
Differentiating the generic process gives
\[
\begin{aligned}
    \dot\gamma_t
    %&\simeq
    % \frac{\eta}{1-\eta}AC
    %\left(B^{1-\eta}+C \times t\right)^{-\frac{1}{1-\eta}}
    %\\
    &=
    \frac{\eta}{1-\eta}\frac{AC}{B}
    \left(1+\frac{t}{\tau}\right)^{-\frac{1}{1-\eta}} .
\end{aligned}
\]
This expression showcases two speed regimes: When \(t\ll\tau\),
the inherited data term \(B^{1-\eta}\) dominates the accumulated-flow term
\(C \times t\):   
\[
    \dot\gamma_t
    \simeq
    \frac{\eta}{1-\eta}\frac{AC}{B}.
\]
Hence the automation boundary grows linearly in time early in the process since initial phase is governed by the inherited data stock near the automation boundary. 
\begin{comment}
In the two bounding processes, the short-run slopes reduce to
\[
\begin{aligned}
    \dot\gamma_t
    &\simeq
    \frac{\eta L}{\overline B}
    &&\text{using the lower bound on \(1-\gamma_t\),}
    \\
    \dot\gamma_t
    &\simeq
    \frac{\eta L}{\underline B}
    &&\text{using the upper bound on \(1-\gamma_t\).}
\end{aligned}
\]
\end{comment}
By contrast,
when \(t\gg\tau\), the accumulated-flow term \(C \times t\) dominates and
\[
\dot \gamma_t
    \simeq
    \frac{\eta}{1-\eta}
    A C^{-\frac{\eta}{1-\eta}}
    t^{-\frac{1}{1-\eta}}.
\]
which slows down over time because of diminishing returns to data accumulation. 

\begin{figure}[t]
    \centering
    \label{fig:speed-bounds}
    \begin{minipage}{0.48\textwidth}
        \centering
        \includegraphics[width=\linewidth]{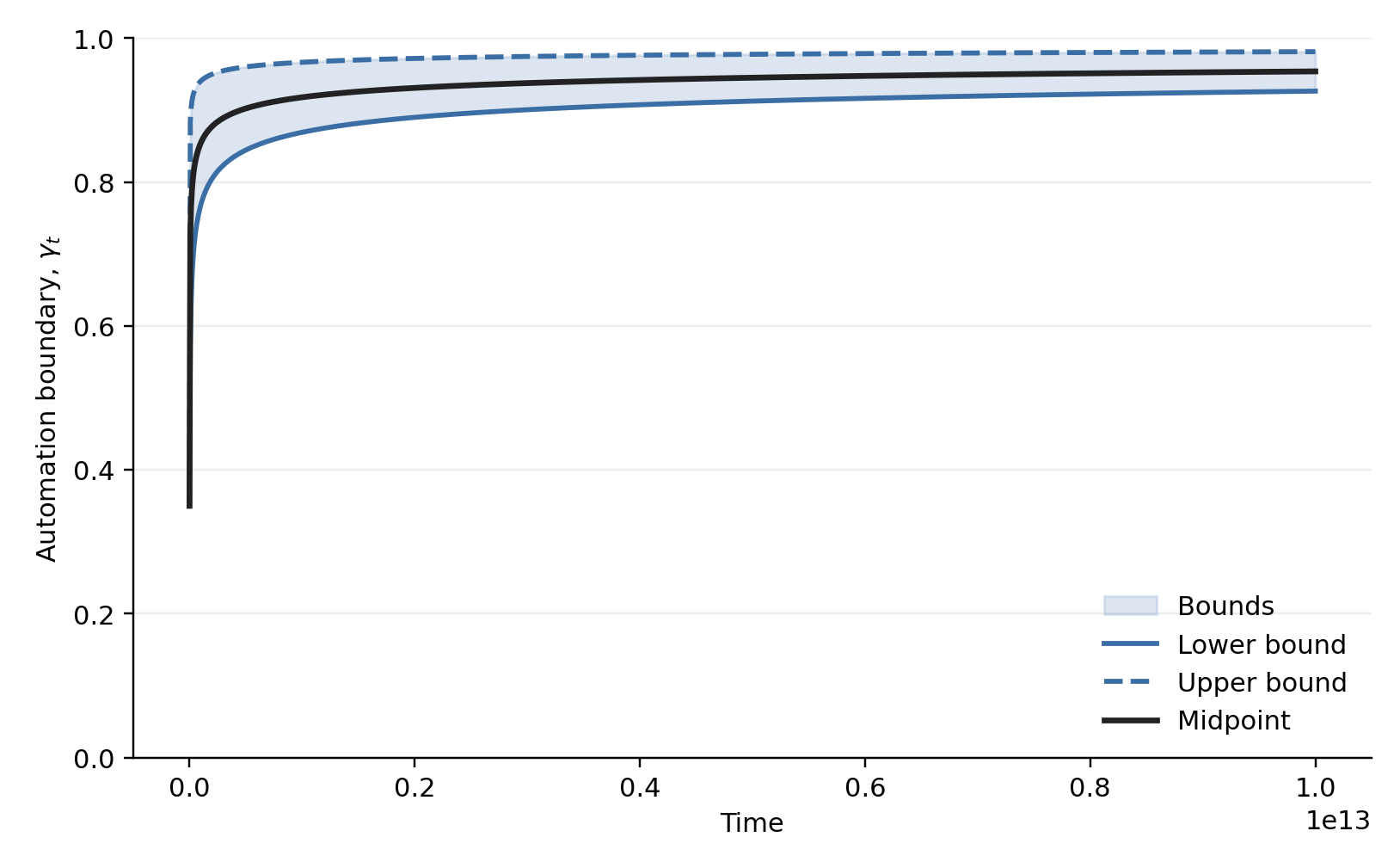}
        \vspace{-0.5em}
        \centerline{\small (a) Long-horizon path}
    \end{minipage}
    \hfill
    \begin{minipage}{0.48\textwidth}
        \centering
        \includegraphics[width=\linewidth]{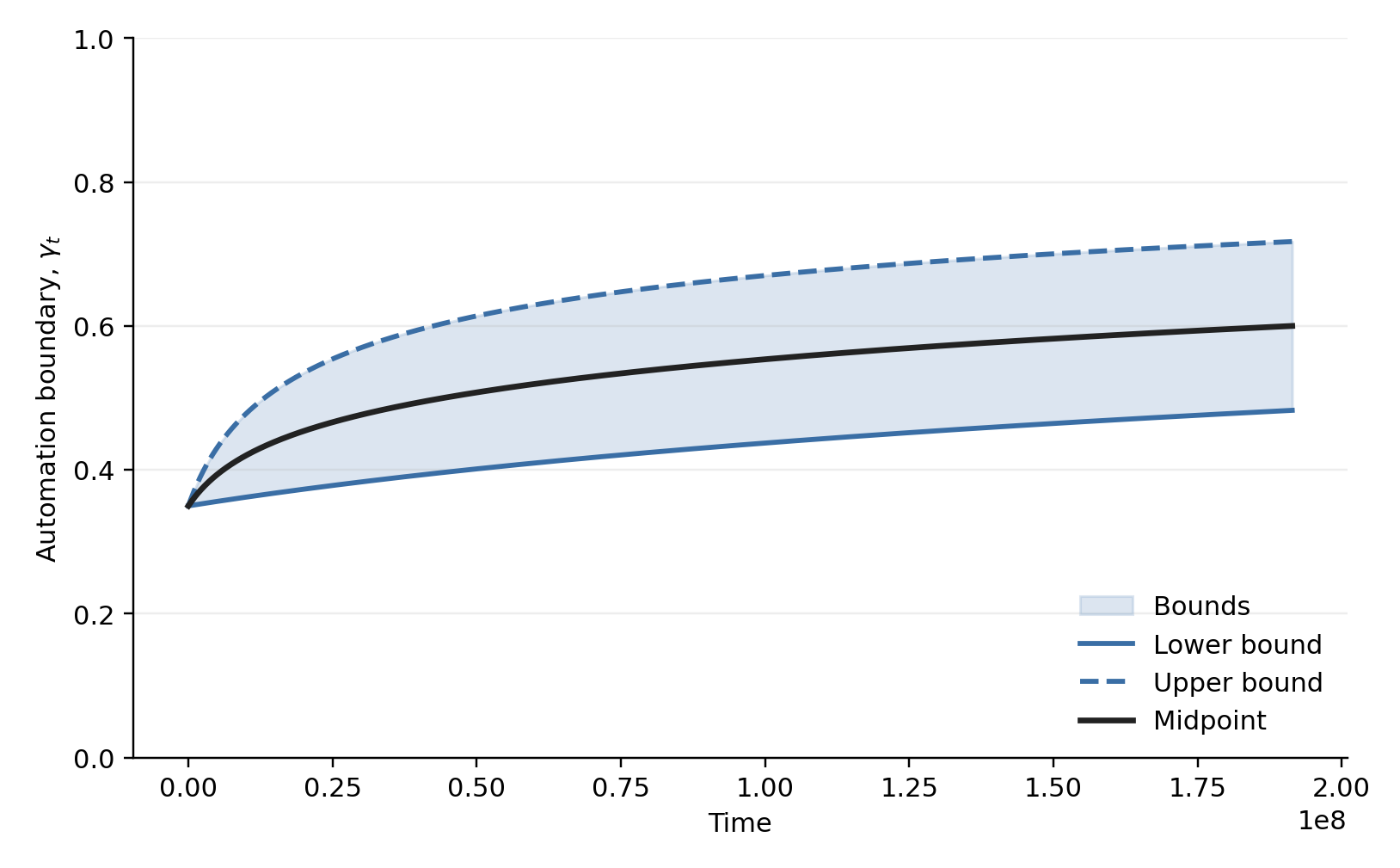}
        \vspace{-0.5em}
        \centerline{\small (b) Linear-regime}
    \end{minipage}
        \caption{Bounds on the speed of data-driven automation}
\end{figure}

\cref{fig:speed-bounds} illustrates this by plotting bounds for the automation boundary \(\gamma_t\).\footnote{We set
    \(\eta=0.2\), \(\sigma=0.5\), \(K=1\), \(L=120\),
    \(\underline f=0.95\), \(\overline f=1\), and
    \(D_{i0}=10^{10}\) for all tasks.} 
    Panel (a) plots the path of automation over a long horizon. Note here that we see a `sudden' increase in automation (as measured on that time scale) followed by a tapering off: when $\approx 95\%$ of tasks are already automated, the speed of automation progressively slows. Panel (b) zooms in on the early transition by a factor of $10^5$: here the path is approximately linear in calendar time.

\section{Contagious Automation} \label{section:contagious}

% \textcolor{blue}{BX: Bryant fix the kernel renormalization issue in the connectedness $\to$ automation proof.}

In this section, we introduce cross-task spillovers into our economy as captured by the graphon $W: [0,1]^2 \to [0,1]$ and analyze how it shapes (i) whether the economy is fully automated in the limit; (ii) the long-run composition of data and production; and (iii) the speed of automation.
%[[I DELETE THIS AND INCORPORATE SOME OF IT WITH THE WRITING I THINK IT IS VERY REPEITATIVE]]

\subsection{Limit Automation with Spillovers} 

\paragraph{Connectedness implies full automation.} With cross-task spillovers, automation will not, in general, take a threshold form since a high index task might, through spillovers, accumulate a higher effective data stock than a lower index task. To this end, we will define:  
\[
    \mathcal{K}_t
    :=
    \left\{
    i\in X:
    \frac{r_t}{\psi^K_{i,t}}\le \frac{w_t}{\psi^L}
    \right\}
    \quad \text{and} 
    \quad 
    \mathcal{L}_t:=X\setminus \mathcal{K}_t.
\]
as the set automated tasks and tasks produced by labor at time $t$.
We let $\mu$ denote the Lebesgue measure so full automation means $ \mu(\mathcal{L}_t)\to0$.

\begin{definition}[Uniform strong connectedness]\label{def:uniform_strong_connectedness}
A directed graphon \(W\) is uniformly strongly connected if there exist
\(n\in\mathbb N\) and \(\varepsilon>0\) such that
\[
    W^{(n)}(i,j)\ge \varepsilon
    \qquad\text{for a.e. }(i,j)\in X^2.
\]
where $W^{(n)}$ is defined as follows: $W^{(1)}=W$ and 
\[
W^{(m+1)}(i,j)
    :=
    \int_X W(i,k)W^{(m)}(k,j)\de k \quad \text{for $m \geq 1$.}
\]
Thus \(W^{(m)}(i,j)\) is the total intensity of length-\(m\) directed spillover paths from source \(j\) to beneficiary \(i\).
\end{definition}

We call this uniform strong connectedness because every task can receive data,
directly or indirectly, from every other task within a uniformly bounded number
of spillover steps and with uniformly positive intensity.

\begin{proposition}[Contagious automation]
\label{prop:generalW_full_automation}
Suppose \(W\) satisfies uniform strong connectedness. For any \(\sigma>0\),
any \(\eta\in(0,1)\), and any bounded initial data stock \(\bm{D_0}\in L^1_+(X)\),
every equilibrium path satisfies full automation: 
\[
    \mu(\mathcal{L}_t)\to0.
\]
\end{proposition}

\cref{prop:generalW_full_automation} states that connectedness \emph{per se} is sufficient for full limit automation, even if cross-task production is very substitutable $(\sigma >> 1)$ and there are enormous differences in returns to data across tasks ($f_i$ is steeply decreasing in $i$). The basic intuition is that when tasks are highly substitutable ($\sigma > 1/\eta$) \emph{local} spillovers that generate a connected graphon can stand-in for \emph{global equilibrium} forces from the gross-complements case ($\sigma \leq 1$) that guides capital allocation and thereby data accumulation toward data-poor tasks.\footnote{This overturns \cref{prop:substitutes}, which establishes partial automation in a data autarky when tasks are sufficiently substitutable.} 

Put differently, the increasing return to scale that follows any degree of data spillover across tasks is so strong that it overcomes any force towards capital concentration, caused by high degree of task substitutability. As such, the degree of task complementarity that is required to distribute capital across tasks when the economy is in data autarky is not needed when there is data spillover across tasks.

That local spillovers \emph{per se} are sufficient is perhaps surprising, and goes against the broad intuition that connectedness is often insufficient for well-behaved limit behavior in networks with nonlinear dynamics. For example, in models of opinion dynamics, it is well-known that connectedness \emph{per se} is not enough for convergence under non-linear aggregation of opinions \citep*{cerreia2024dynamic} or in binary-action coordination games \citep{morris2000contagion}. Our setting is distinct from these models because (i) equilibrium factor prices (a `global' variable) shape the path of data accumulation $(\pmb{D}_t)_t$; (ii) our state space is unbounded; and (iii) in the long-run the whole vector of data and productivities diverge to infinity ($\pmb{D}_t \to +\infty$ and $\pmb{\psi}^K_t \to +\infty$) whether or not the graphon is connected. 

A key feature governing whether automation is contagious or not is the \emph{limit ratio of effective task-specific data}, which---for any two tasks $i,j \in \text{Int}(X)$---is given by  
\[
\lim_{t \to \infty}\cfrac{\mathcal{A}_i(\pmb{D}_t)}{\mathcal{A}_j(\pmb{D}_t)} =: \mathcal{R}_{ij}. 
\]
A sufficient condition for contagious automation is that this ratio does not go to either $0$ or infinity i.e., $0 < \mathcal{R}_{ij} < +\infty$. If it did, then in the case where tasks are arbitrarily substitutable $(\sigma >> 1)$, one of those tasks need not be produced by capital in equilibrium. 

What governs the ratio of effective task-specific data $\mathcal{R}_{ij}$? A loose intuition is that connectedness ensures for some task $k \in X$ close to $i$, their neighbors overlap nontrivially, as illustrated by \cref{fig:overlapping} that illustrates a graphon $W$ in which tasks that are closer to each other (in index) have stronger spillovers. 
\begin{figure}[t] 
    \centering
        \includegraphics[width=0.5\textwidth]{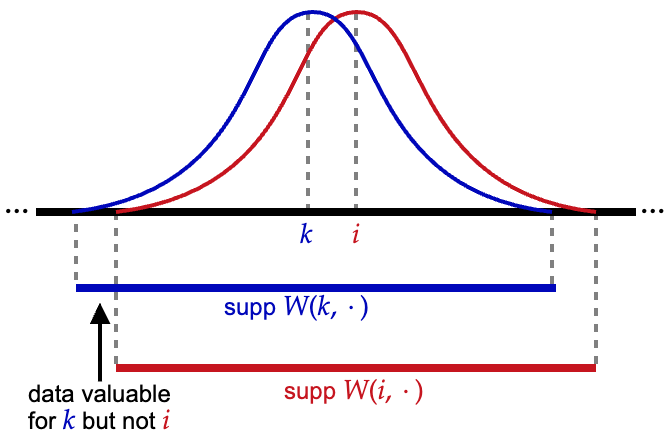}%
    \caption{Overlapping neighbors}\label{fig:overlapping}
\end{figure}
But this, \emph{per se}, does not imply that the limit ratio of effective data $\mathcal{R}_{ik}$ is bounded away from $0$ because their supports are non-overlapping---task $k$ but not $i$ could value the data of some other task $\ell \in X$. 
Then, if $D_{\ell t}$ grows quickly, the limit ratio of effective data $\mathcal{R}_{ik}$ could, in principle, be $0$. Nonetheless in \cref{app:contagion} we show via a chaining argument that, because the space of tasks $X$ is bounded, this cannot happen \emph{in equilibrium}---connectedness implies that along the equilibrium path full automation is inevitable.\footnote{In \cref{network_structure_long_run} we also characterize the long-run composition of task production and data as a function of the leading eigenfunction of the graphon.}

\subsection{How Do Spillovers Shape Speed?} We next numerically investigate the speed of automation in the presence of spillovers. \cref{fig:spillover_gamma} illustrates the impact of varying the breadth and depth of cross-task spillovers on the evolution of the automation boundary, separately for the gross complements case and the gross substitutes case.

In the case of complements, automation increases roughly linearly in log time with or without spillovers. Intensifying the \emph{breadth} of spillovers while keeping the intensity of spillovers constant has little impact on the pace of automation. This is because data across tasks are, in equilibrium, already balanced absent spillovers per \cref{prop:complements}, so the benefit from being able to draw on other tasks' data is muted. If, however, we intensify the \emph{depth} of spillovers i.e., increasing the intensity of spillovers, this generates a `one-time' boost to automation since doing so is equivalent to increasing the degree to which each task values data. This is illustrated in Panel \ref{fig:leftSpill}.

\begin{figure}[t] 
    \centering
    \subfloat[$\sigma = 0.5$ \label{fig:leftSpill}]{%
        \includegraphics[width=0.5\textwidth]{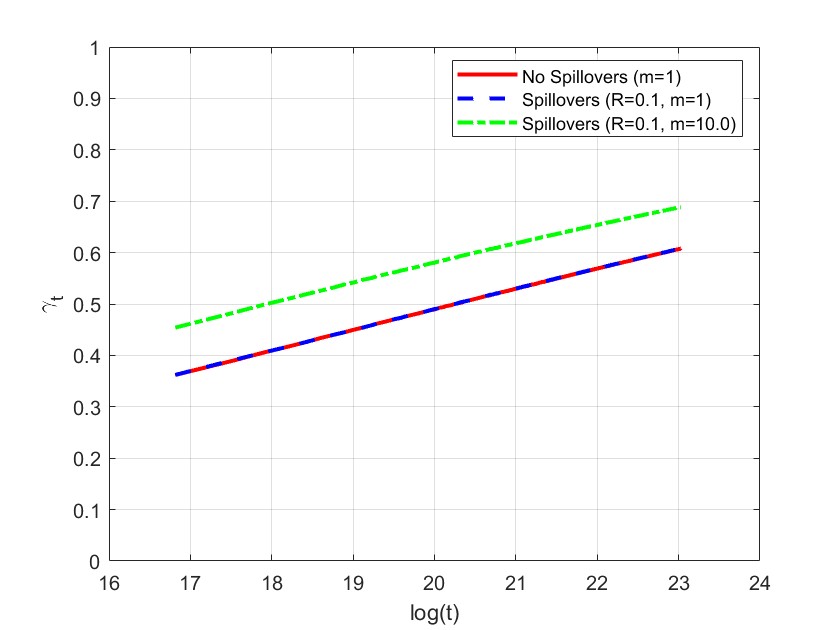}%
    }
    \hfill
    \subfloat[$\sigma=5.5$ \label{fig:rightSpill}]{%
        \includegraphics[width=0.5\textwidth]{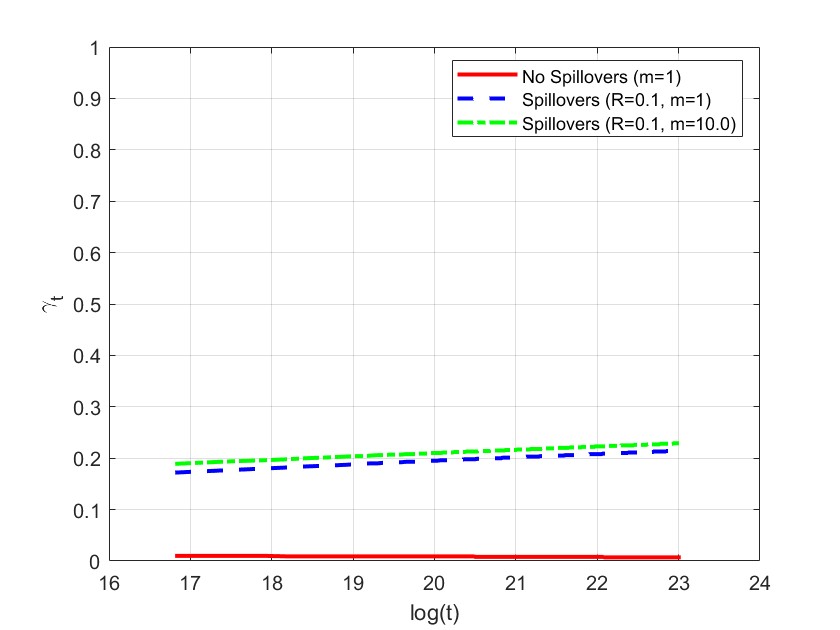}%
    }
    \caption{Evolution of $\gamma$ under alternative spillover geometries}
    \footnotesize
    Note: $m$ denotes the measure of the $g(\cdot)$ kernel that induces $W(\cdot,\cdot)$, which captures the depth of spillovers. The remaining parameters are as in  \cref{fig:complements} except with $L=100$.
    \label{fig:spillover_gamma}
\end{figure}

In the case of substitutes, whereas automation declines linearly in log time absent spillovers, the trajectory flips when spillovers are introduced, consistent with \cref{prop:generalW_full_automation}. Indeed, we know data will be highly skewed towards data-rich tasks in the case where $\sigma > 1/\eta$, so the ability to draw on their data has a dramatic impact on furthering automation. Scaling the depth of spillovers similarly increases automation, but the impact is less pronounced compared to scaling breadth in the case of substitutes in the presence of imbalanced data. This is illustrated in Panel \ref{fig:rightSpill}. 

The broad economic takeaway from this exercise is that equilibrium automation with and without spillovers is surprisingly slow---approximately linear growth in \emph{log-time}. This is consistent with \cref{prop:heterogeneous-speed} that shows that in the absence of spillovers, the automation boundary exhibits `power law' behavior in time with a `fat tail' of tasks performed by labor. Here, data spillovers lead to a \emph{level change} in the automation boundary, but does not change its long run speed. Of course, this analysis 

Throughout the analysis above, we have kept the capital stock fixed. This might approximate the shorter run in which supply chain constraints prevent rapid scaling up of computational resources. The takeaway is that while data-driven automation is powerful on its own, it is unlikely to speed up automation substantially in the long run. But in the short-run, we will see that the network structure plays an important role in shaping automation dynamics.

 %\cref{sec:efficiency} illustrates that there is a data externality in the composition of data accumulation which makes the automation  inefficiently slow.

\subsection{Core-Periphery Automation} 

\begin{figure}[t]
    \centering
    \includegraphics[width=0.4\linewidth]{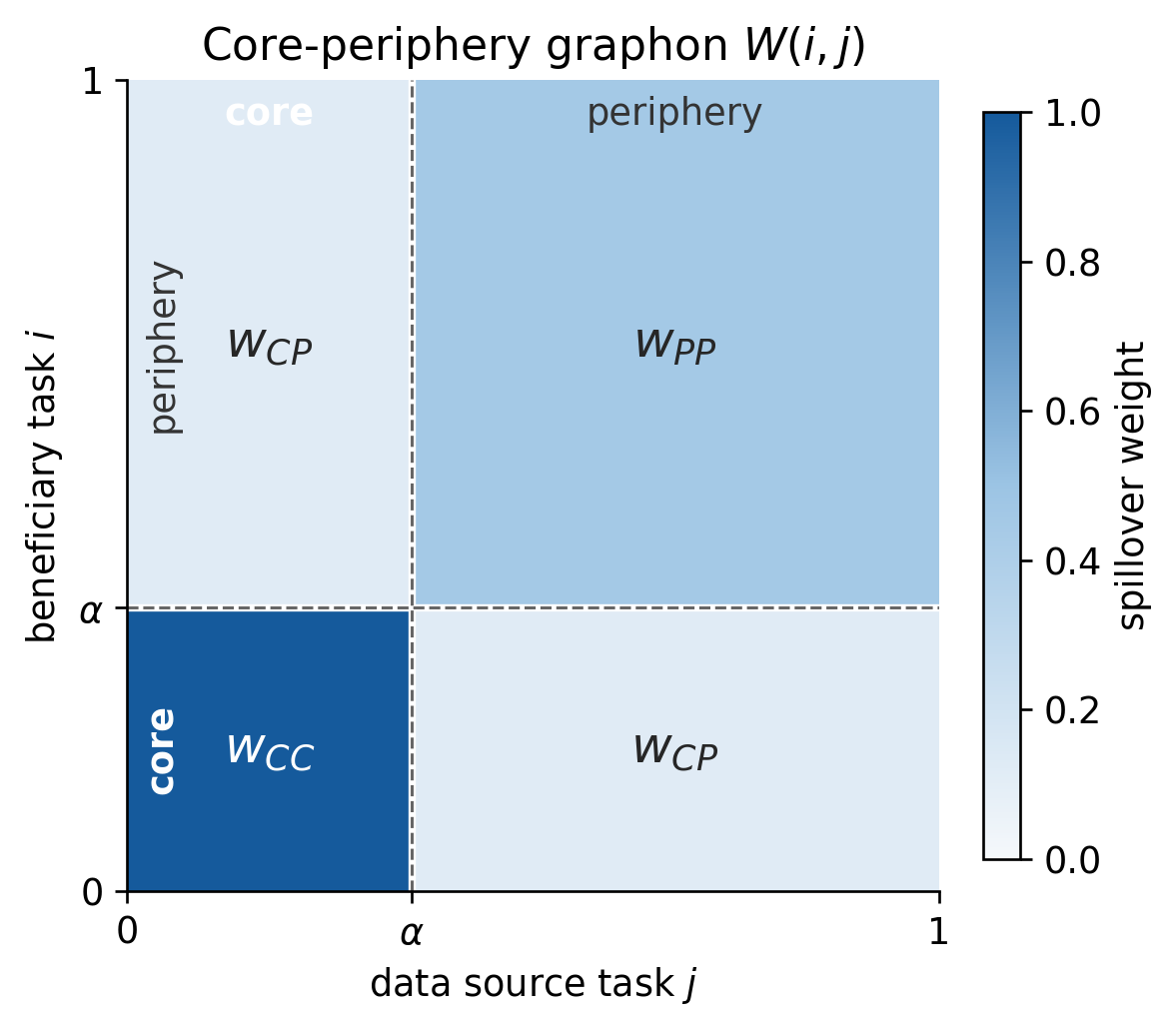}
    \caption{A two-block core-periphery graphon. Rows index the beneficiary
    task \(i\), columns index the data source task \(j\), and darker shading
    indicates stronger spillovers.} 
    \label{fig:core-periphery-graphon}
    %The graphon is directed: \(W_{CP}\) is the bridge from core data to peripheral effective data, while \(W_{PC}\) is the reverse bridge. The pure weak-cut case studied below sets \(w_{PP}=w_{PC}=0\), so the peripheral block becomes automatable only through \(w_{CP}\).
\end{figure}

We now illustrate the role of network sturcture in short run dynamics via the following stylized example. Partition tasks \(X=[0,1]\) into a core \(C\) of measure \(\alpha\) and a
periphery \(P\) of measure \(1-\alpha\). Tasks are homogeneous within
each block: \(f_C>f_P\) and \(D_{C0}>D_{P0}\). Effective data at the block level is
\[
    A_C= \underbrace{w_{CC}D_C}_{\text{core$\leftrightarrow$core}}+ \underbrace{w_{PC}D_P}_{\text{periphery$\to$core}}
    \quad \text{and}
    \quad 
    A_P=
    \underbrace{w_{CP}D_C}_{\text{core$\to$periphery}}+ 
    \underbrace{w_{PP}D_P}_{\text{periphery$\leftrightarrow$periphery}}.
\]
Capital productivities are \(\psi_C=f_CA_C^\eta\) for core tasks and \(\psi_P=f_PA_P^\eta\) for periphery tasks. Our economy evolves across two phases. In the first phase, only core sectors are automated. In the second phase, automation begins to take hold in the wider economy. There are two crucial forces at play: 
\begin{itemize}
    \item \emph{Core data-feedback loops:} the overproduction of core sectors (relative to peripheral sectors) generate data that exhibit strong within-block spillovers (high $w_{CC}$). This drives up the requisite productivity required for non-peripheral sectors to be automated: the opportunity cost of deploying capital in peripheral sectors goes up when core sectors are productive. 
    \item \emph{Bottlenecked spillovers:} If there are only weak spillovers from core sectors e.g., software and coding related tasks to peripheral sectors (small $w_{CP}$). This means that the capital productivity of peripheral sectors lag behind, and it takes longer for automation to reach the wider economy.
\end{itemize}

Both of these forces can be gleaned from the following equation that pins down the requisite degree of capital productivity in peripheral sectors for automation to begin in the wider economy: 
\[
    \left(\frac{L\alpha}{K (1-\alpha)}\right)^{1/\sigma}
    \psi_C^{\frac{\sigma-1}{\sigma}}.
\]
and the time $\tau$ at which the wider economy begins to be automated can be analytically characterized via a pair of coupled ODEs. We emphasize  here that the automation can have a non-monotonic relationship with link strength: by intensifying the degree of within-core spillovers $w_{CC}$, this can delay the time at which peripheral tasks begin to be automated. This is because $\psi_C$ now grows more quickly which, in turn, sucks up capital. The corresponding automation paths are illustrated in \cref{fig:core_periphery_shortrun_wcc_two_panel} where the dotted green lines show the total share of sectors automated in the economy, while the blue line illustrates the total share of the periphery automated. The left panel shows the case with low intensify spillovers within the core, and the right panel shows the case with high intensity spillovers.\footnote{A fuller set of simualtions for different parameter values is in \cref{appendix:peripheral automation sim}.}

\begin{figure}[t]
    \centering
    \includegraphics[width=0.75\linewidth]{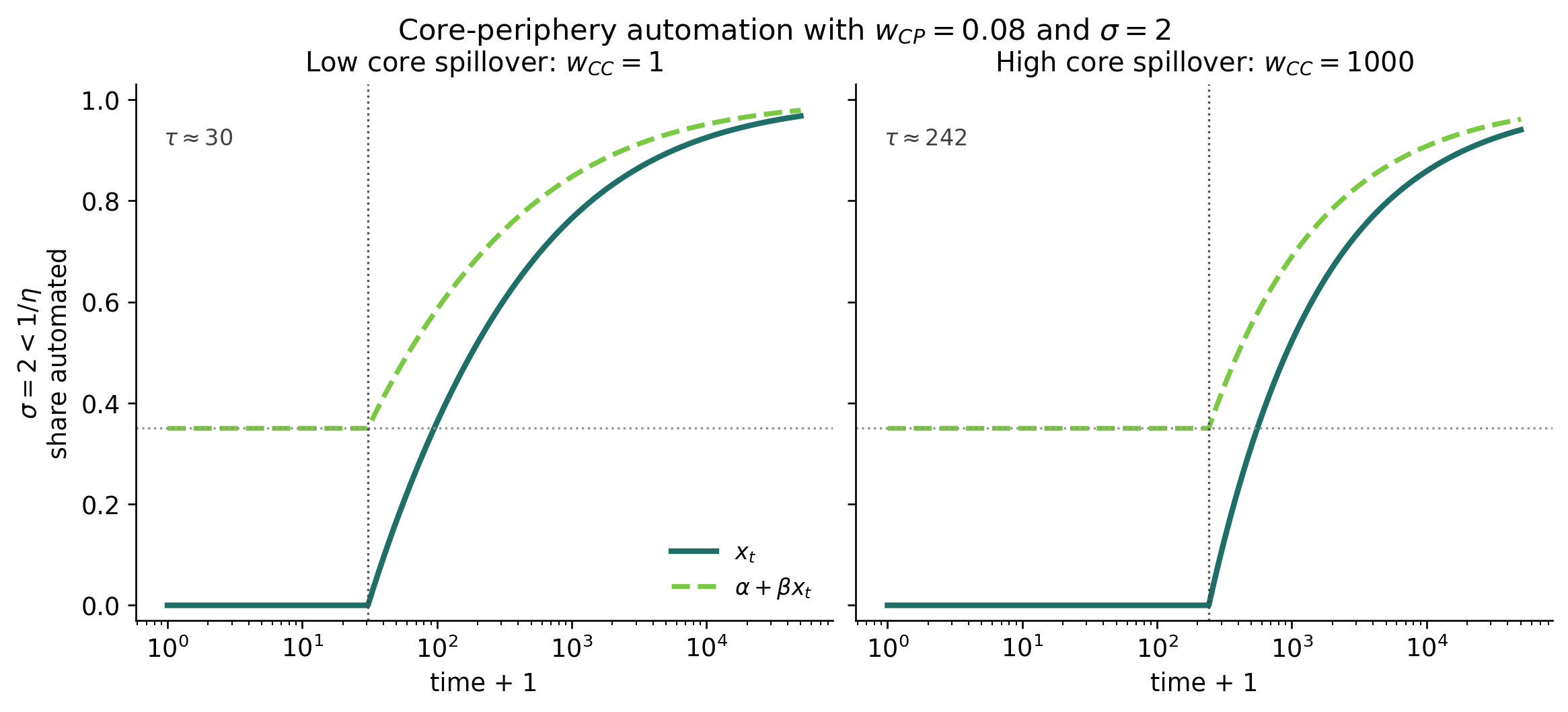}
    \caption{Automation paths} 
    \label{fig:core_periphery_shortrun_wcc_two_panel}
\end{figure}

Taken together, our results suggest that although the long-run speed of the economy exhibits power law behavior, the short-run dynamics can depend richly on the network structure of data spillovers. 

%\cref{fig:core-periphery-grid} highlights two points. First, the weak-cutdelay is present both below and above the threshold \(1/\eta\): smaller\(w_{CP}\) shifts the onset of peripheral automation to the right. Second, thedelay is much more severe in the high-substitutability row. This is the regimein which the core's productivity advantage strongly raises the automationcutoff faced by the periphery.

% \textcolor{blue}{TWO GRAPHS HERE. ON THE LHS COMPLEMENTS CASE, AUTOMATION BOUNDARY IN LOG-TIME, WITH 3 LINES. NO SPILLOVERS, ADD SPILLOVERS BUT KEEP $\int g dx = 1$, AND SPILLOVERS BUT MAKING $\int g dx$ larger. ON THE RHS, AUTOMATION BDD IN LONG-TIME FOR SUBSTITUES WITH THE SAME 3 GRAPHS BUT CAN YOU SIMULATE IT FOR LONGER?}

%%%WE CAN PROBABLY SAY SOMETHING ABOUT BALANCED PATH GIVEN g/W HERE! 
%\begin{definition}[More connected]
%    For two graphons $W, W' \in \mathcal{W}$, say that $W$ is \emph{more connected} than $W'$ (denoted $W \geq W'$) if $W(i,j) \geq W'(i,j)$ for each $i,j \in X$. 
%\end{definition} 

\section{Is Data-Driven Automation Efficient?} \label{sec:efficiency} 

We have thus far offered an analysis of when and why our economy exhibits full limit automation, with and without data spillovers. A natural question is whether this degree of automation is efficient. 

A basic observation is that because all producers in our economy are atomless and data generated from economic activity of each producer is shared for free among all intermediate good producers, they do not internalize the value of producing more data today on future production---both for their own task or for other tasks that might benefit from their data. Indeed, modern machine learning systems rely on vast amounts of data that no single firm produces on their own so they might not internalize this externality in equilibrium. 

In this section, we again restrict attention to the economy in data autarky and compare the equilibrium path of data accumulation with the first-best benchmark.  To preview our results---under two simplifications to be discussed---we show that the myopic equilibrium path is generically suboptimal. Specifically, the planner would amplify the market's bias: in the gross complements and weak substitutes case ($\sigma \leq 1/\eta$), the planner would allocate even more capital to tasks with less data to speed up the data convergence exhibited in the decentralized equilibrium; in the strong substitutes case ($\sigma > 1/\eta$), the planner would allocate even more capital to tasks with more data, thereby accelerating the data divergence observed in the decentralized equilibrium.

\paragraph{Simplifying assumptions: The coarsened economy.} To focus on planner's incentives to `tilt' the direction of capital allocation, we study a simpler, coarsened economy that is more amenable to analysis. We make two simplifications of our main model (in data autarky) to this end. First, we suppose production is entirely capital-driven i.e., $L = 0$ and we are in an `$AK$-style' economy, albeit with no capital accumulation. This simplification allows us to focus on the direction of capital tilting; we expect a similar mechanism would be at play for the allocation of labor. Second, we collapse our tasks into two blocks for tractability. We call the task-block $[0,1/2]$ `$h$' (high data productivity), and the task-block $(1/2,1]$ `$\ell$' (low). Qualitative insights from this simple two task-block case should generalize to our baseline economy featuring a continuum of tasks.

\paragraph{Planner's problem.} Consider a planner who chooses the allocation of capital to maximize output dynamically for some discount rate $\rho > 0$. The planner solves: 
\begin{align} \label{prob:planner}
    \max_{k_{1t},k_{2t}} &\int_{0}^{T}e^{-\rho t}\underbrace{\left(\sum_{i\in\{h,l\}}\frac{1}{2}\left(\psi^{K}\left(D_{it}\right)k_{it}\right)^{\frac{\sigma-1}{\sigma}}\right)^{\frac{\sigma}{\sigma-1}}}_{\equiv Y_{t}}dt
    \\
    \text{s.t.}
    &\quad \sum_{i\in\{h,\ell\}}\frac{1}{2} \cdot k_{it}=K \nonumber \quad \text{and} \quad \pmb{D}_t = \pmb{D}_0 + \int^t_0 \pmb{y}_s ds \nonumber
\end{align}
where the constraints are the capital market clearing condition and the usual data stock law of motion.\footnote{That is, $D_{it} = D_{i0} + \int^t_0 y_{is}ds$ for each $i \in X$.} 
We say that the equilibrium  is efficient if it solves the planner's problem \eqref{prob:planner}. \cref{prop:generic_inefficiency} summarizes our  main  efficiency result.

\begin{proposition}[Inefficiency and planner's problem] \label{prop:generic_inefficiency}
    If $f_h=f_l$ and $D_{h0} > D_{\ell0}$, then the equilibrium path is efficient iff $\sigma = 1/\eta$. Otherwise:
    \begin{itemize}
        \item[(i)] Compared to the myopic allocation given planner's current data stock, the planner allocates---for all $t \in [0,T]$---weakly more capital to $l$ if $\sigma < 1/\eta$ and weakly more capital to $h$ otherwise.
        
        \item[(ii)] Compared to the \emph{decentralized equilibrium}, the planner allocates
        \begin{enumerate}
            \item  weakly more capital to $l$ initially, but may reverse if $\sigma < 1$,
            \item weakly more capital to $l$ if $1\leq \sigma < 1/\eta$,  for all $t \in [0, T]$,
            \item  weakly more capital to $h$ if $\sigma > 1/\eta$,for all $t \in [0, T]$.
        \end{enumerate}
    \end{itemize}
\end{proposition}
    
%\begin{proof}
 %   See \cref{sec:pf_generic_inefficiency}.
%\end{proof}
% \paragraph{How does the planner wish to tilt capital?}

% \begin{proposition} \label{prop:coarsened_capitaltilt}
%     Suppose that tasks are gross complements ($\sigma \leq 1$). Then for each time $t \geq 0$, the planner allocates more capital toward the low productivity block $\ell$ vis-a-vis the equilibrium allocation i.e., $k^*_{ht} \leq k_{ht}$. 
% \end{proposition}

\begin{figure}[t]
    \centering
    % This command forces the subfigure captions to center, even with line breaks
    \captionsetup[subfigure]{justification=centering}

    \subfloat[{\small Capital allocated to task-block 1 \\ when $\sigma = 0.5$}]{%
        \includegraphics[width=0.45\textwidth]{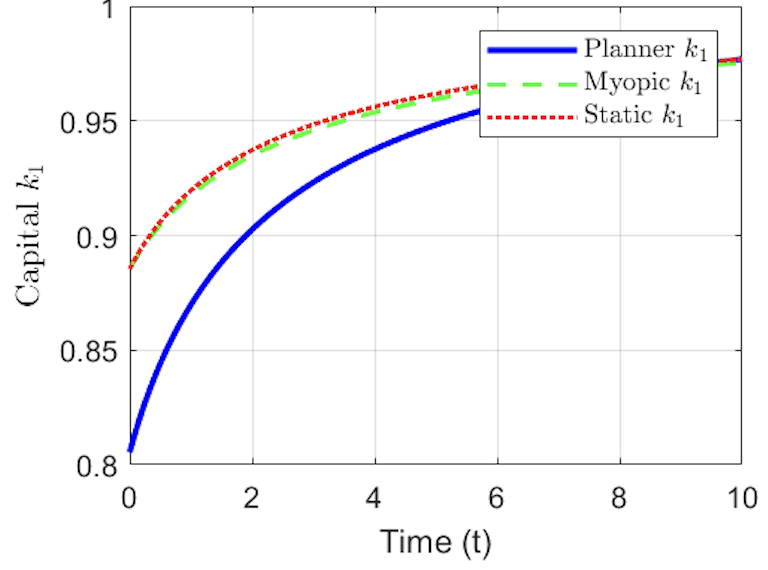}%
    }
    \hfill % Pushes the subfigures to the edges
    \subfloat[{\small Capital allocated to task-block 1 \\ when $\sigma = 5.5$}]{\includegraphics[width=0.45\textwidth]{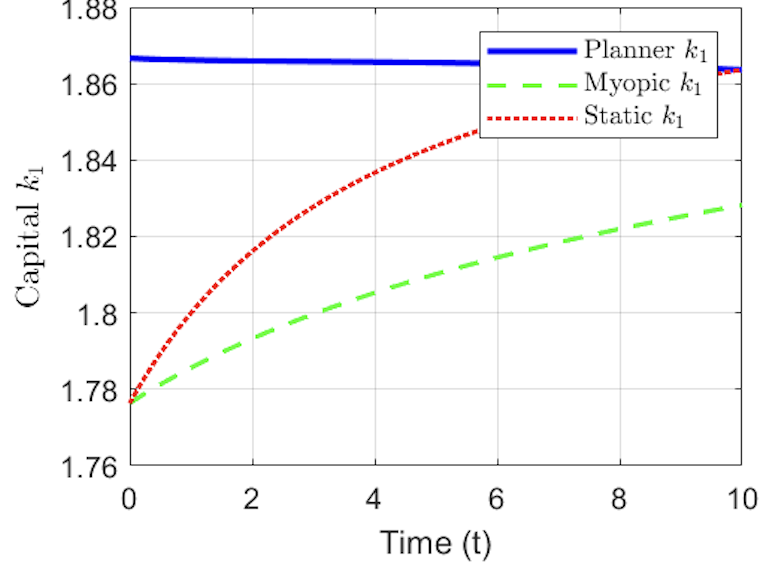}%
    }
    \caption{Numerical illustration of the planner's solution}
    \label{fig:planner}
\end{figure}

 The difference between planner's decision how to allocate capital (unequally) compared to the market's decision is governed by the following inequality:
\begin{equation} \label{eqn:planner_comparison}
    \frac{\lambda_{lt}}{\lambda_{ht}} \lessgtr \frac{\psi^K_{ht}}{\psi^K_{lt}},
\end{equation}
where $\lambda_{it}, \text { } i\in \{l,h\}$ is the co-state variable capturing the value of block $i$'s data. Away from the knife edge case where $\sigma = 1/\eta$ where the above relation holds with equality at all times, the equilibrium path of capital allocation is generically inefficient. 

The important observation is that the direction of inefficiency in equilibrium capital allocation depends  on the degree of cross-task substitutability. If $\sigma < 1/\eta$, we show the left hand side of Equation \eqref{eqn:planner_comparison} will be greater, or that the relative value of data for low data tasks is sufficiently high so as to justify distorting production presently. In this case, the market inefficiently allocates too much capital to tasks that are currently data-rich---even if the equilibrium already biases capital allocation toward the data-poor block $\ell$, the planner prefers to tilt this allocation more aggressively. In so doing, the planner speeds up the rate at which block $\ell$'s data is accumulated---at the cost of static misallocation---in order to loosen the bottleneck more quickly. 

Conversely, when tasks are highly substitutable $\sigma > 1/\eta$, the right hand side of \cref{eqn:planner_comparison} will be larger, so the planner is incentivized to lean into tasks that are \emph{already} data-rich---even though the equilibrium already biases capital allocation toward the data-rich block $h$, the planner prefers to tilt this allocation more aggressively. In doing so, the planner harnesses the data-feedback loop to her advantage. The proof of \cref{prop:generic_inefficiency} formalizes the above argument. 

\begin{figure}[t] 
    \centering
        \includegraphics[width=0.55\textwidth]{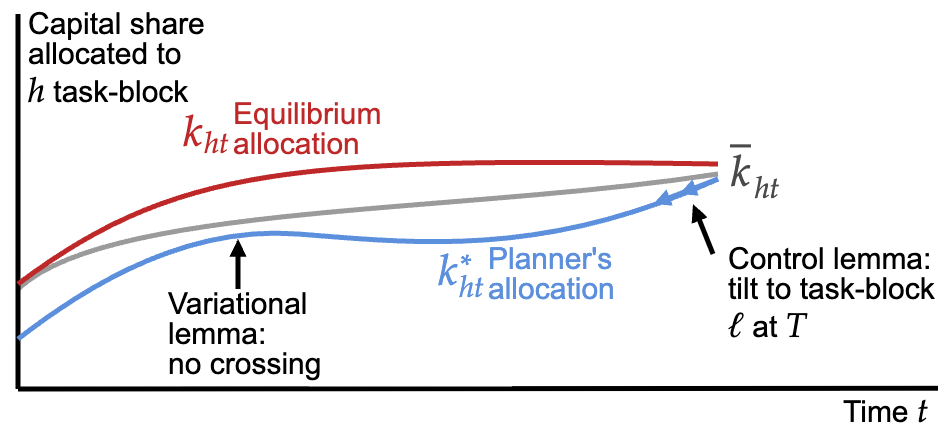}%
    \caption{Illustration of proof ideas for $1 < \sigma < 1/\eta$}\label{fig:tilt}
\end{figure}
%The proof of  \cref{prop:generic_inefficiency} is surprisingly involved and is deferred to \cref{appendix:proofs}.
 \cref{fig:tilt} illustrates  the intuition for the case when $\sigma < 1/\eta$. 
Write $\boldsymbol{k}^*_h$ to denote the planner's optimal path of capital allocation to block $h$, and $\boldsymbol{k}_h$ to denote the equilibrium path. For each task-block $h$, $\boldsymbol{k}^*_h$ and $\boldsymbol{k}_h$ are two infinite-dimensional vectors, where the dimension is time. To show that $k^*_{ht} \leq k_{ht}$ for all $t$, as claimed by the statement of \cref{prop:generic_inefficiency}, we combine a control argument to sign the direction of tilting in the long-run by running time backwards, and a variational argument to argue that the planner's optimal capital allocation cannot cross an auxillary path of capital $\boldsymbol{\bar k}$ that maximizes static time-$t$ production given the same data stock as the planner's.

%Write $(k^*_{ht})_t$ to denote the planner's optimal path of capital allocation to block $h$, and $(k_{ht})_t$ to denote the equilibrium path.
%$(k_{ht})_t$

\paragraph{Reintroducing labor to the coarsened economy.} \cref{prop:generic_inefficiency} applies to an economy with no labor, which does not speak to the \emph{speed} of automation under the planner's problem. In \cref{appendix:planner_numerics}, we reintroduce labor into the economy when $\sigma \leq 1/\eta$ and show numerically that the planner, as before, continues to tilt capital towards the data-poor task. This, in turn, \emph{accelerates} automation relative to the decentralized equilibrium. We expect the opposite to be true in the $\sigma > 1/\eta$ case, as the planner's amplified bias towards data-rich tasks concentrates the use of capital and thereby shrinks the frontier of automation.

\section{Capital Accumulation} \label{section:capitalaccumualation}

We have thus far analyzed an economy populated by hand-to-mouth households, without any capital accumulation. This was a deliberate modeling choice to focus on the role of heterogeneous data accumulation. Of course, capital accumulation---in the form of robotics, data centers, computational power, and so on---plays a central role in driving modern AI systems. We will accordingly augment our baseline model with capital accumulation but adopt the simplifying assumption that tasks are symmetric, under which the economy can be well approximated by an `AK' representation, and rule out cross-task spillovers. This is, in some sense, the `most pessimistic' case and will serve as a lowerbound on growth for an economy with higher substitutability and/or spillovers. We emphasize simplification is largely cosmetic and does not affect the core result that output and consumption become unbounded in finite time. That said, the symmetry assumption does have bite for the share of automated tasks; see the discussion of \cref{prop:substitutes} for details.  

We now introduce the standard household's consumption-savings problem. The representative household solves,
\[
\max_{c_{t},a_{t}}\int_{0}^{\infty}e^{-\rho t}\frac{c_{t}^{1-\theta}}{1-\theta}dt\quad s.t.\quad\dot{a}_{t}=w_t + r_{t}a_{t}-c_{t}
\]
where as usual $\theta$ is the CRRA parameter, $\rho>0$ is the discount rate, and $a_{t}$ is the wealth at time-t which earns a risk-free rate of $r_{t}$. $L$ is fixed and normalized to one. Our main result in this section is that data-driven automation combined with capital accumulation can deliver runaway growth (output becomes unbounded in finite time), which is reminiscent of the `supply-side singularity' of \cite{nordhaus2021we}. We shall return to discuss two major departures at the end of the section.

In light of the possibility of runaway growth, we have to be more careful with what it means for households to be trading off consumptions and savings decisions optimally. We introduce the following definition:

\begin{definition}[Feasibility and optimality of paths] \label{defn:hh_optimality}
    A path $(c_{t},a_{t})_t$ is feasible if it satisfies the household budget constraint. A feasible path is optimal if either: 
    \begin{itemize}[nosep]
        \item[(i)]  It delivers infinite present discounted utility.
        \item[(ii)] It satisfies both A and B:
        \begin{itemize}
            \item[A.] The Euler condition: ${\dot c_t}/{c_t} = \left(r_t - \rho \right)/{\theta} \quad \forall t \geq 0$.
            \item[B.] The transversality condition: 
            $\lim_{t \to \infty} e^{-\rho T}{a_t}/{c_t^\theta} = 0$.
        \end{itemize}
    \end{itemize}
\end{definition}

Absent capital accumulation, it is straightforward to verify that our economy features no growth in the limit akin to an `$AK$' economy with a fixed level of capital productivity, with growth being positive but finite. By combining (i) data accumulation that runs into diminishing returns but does not depreciate and (ii) endogenous capital accumulation that does depreciate, we show our economy exhibits a singularity: output becomes unbounded in finite time.

\begin{proposition}[The data-driven singularity]\label{prop:singularity}
    There exists a feasible and optimal path for the household's problem such that by some finite time $T$, consumption and output are unbounded and the economy is fully automated. That is,
    \[
    \lim_{t \to T}\gamma_t = 1, \quad
    \lim_{t\to T}\frac{d\ln Y_{t}}{dt}= +\infty,  \quad 
    \lim_{t\to T}\frac{d\ln c_{t}}{dt}= +\infty.
    \]
\end{proposition}

To develop intuition for \cref{prop:singularity}, consider an economy satisfying $Y_t = A_t K_t$ with an exogenous savings rate $s \in (0, 1)$ \`a la Solow. Then changes in output can be bounded below as follows:
\begin{align*}
    \dot{Y}_{t}	&=\dot{A}_{t}K_{t}+A_{t}\dot{K}_{t} \geq A_{t}\dot{K}_{t} =\left(s-\delta/A_{t}\right)A_{t}Y_{t}
\end{align*}
where the first inequality assumes $\dot A_t \geq 0$. If $A_t$ is of the form $Y_t^\omega$ for some $\omega > 0$, then output growth is explosive. We establish our economy indeed features this constant elasticity relation between productivity and output in the proof.

\cref{prop:singularity} echoes the `supply-side singularity' of \cite{nordhaus2021we} though there are two substantial differences. First, \cite{nordhaus2021we} assumes the capital productivity term grows exogenously at some exponential rate; by contrast, data is (i) accumulated endogenously in our economy; and (ii) exhibits potentially steeply diminishing returns. The key is that data-accumulation and capital-accumulation are \emph{strong complements}: either, on their own, can only generate finite growth, but they feed into each other to generate unbounded growth. Second, \cite{nordhaus2021we} assumes that labor and capital are gross substitutes; \cite{trammell2025economic} also note capital-driven growth typically requires gross substitutability across capital and labor. On the other hand, our result does not restrict $\sigma$ to be greater than one. In our economy, for a \emph{fixed} level of automation, we see from \eqref{eqn:GDP_KL} that capital and labor may well be complementary, but as data and/or capital accumulates, the automation boundary endogenoously expands so capital dominates production and (almost) obviates the role of labor. Put differently, whereas the complementarity between capital and labor at any moment in time is defined by task complementarity $\sigma$, data-driven productivity gains for capital will eventually push the task-level perfect substitutability to dominate.\footnote{Also related to this result is contemporaneous work by \cite{jones2025past} who find a singularity through the lens of semi-endogenous growth model provided that capital benefits sufficiently from the idea stock.}

\section{Conclusion} \label{sec:extensions}

In this paper, we develop a model of data-driven automation featuring (i) task-specific data creation; (ii)  endogenous data accumulation; and (iii) cross-task data spillovers, in general equilibrium. These features generate rich equilibrium dynamics in which data plays a dual role in \emph{augmenting} the productivity of already-automated tasks and \emph{expanding} the automation frontier. We provided tight conditions on the degree of cross-task substitutability for the economy to be fully automated in the limit (\cref{prop:complements,prop:substitutes}), and established how cross-task complementarity can lead long-run wages to stagnate (\cref{prop:wages}). We further developed tight non-asymptotic bounds on the speed of equilibrium data-driven automation (\cref{prop:heterogeneous-speed}), showing that it can be surprisingly slow with a `fat tail' of tasks performed by humans even in the long run. We also proved that \emph{any} degree of connectedness via cross-task data spillovers is sufficient to generate full limit automation (\cref{prop:generalW_full_automation}). Even so, equilibrium automation can be inefficiently slow with data externalities distorting the equilibrium composition of data accumulation (\cref{prop:generic_inefficiency}). Finally, we showed that the ability for data to simultaneous augment capital and expand the set of automated tasks can---even if capital runs into steeply diminishing returns from data---amplify capital accumulation to generate unbounded long-run growth (\cref{prop:singularity}). 

Our baseline model can be extended to include multiple sectors, feature non-separable data aggregation, and accommodate a time-varying set of tasks. We think these represent exciting avenues for future work.\footnote{A more detailed discussion can be found in \cref{appendix:discussion_extensions}.} Let us briefly elaborate on the latter. We have repeatedly emphasized how the sample efficiency of humans will, at least for time being, afford them a comparative advantage in the `low data' regime. This offers the possibility that when the set of tasks changes through time, humans retain a temporary advantage on new tasks. But this advantage does not persist: over time, the economy will accumulate more data on these new tasks until they become automated in equilibrium. In this economy, \emph{human labor facilitates automation} by producing frontier tasks---and attendant training data---precisely when data is scarce. 

Another possibility is that even though such tasks are new and the economy has no data on it specifically, there is sufficient overlap with existing tasks (as given by our graphon) that AI are already competent at such tasks by training on the existing stock of data for similar tasks. In this world, strong data spillovers allow automation to \emph{bootstrap itself} by continually using the endogenously generated data on existing, similar tasks to automate new tasks soon after they emerge. 

Which of these economies emerge will ultimately depend on the \emph{speed} at which new tasks emerge relative to the \emph{strength} of spillovers from present tasks to new tasks. When new tasks emerge quickly, tasks near the frontier have, themselves, not been around for long---and so data on them are also poor. Then labor could play an important dual role in both producing new tasks as well as generating data on it. But if new tasks emerge slowly, similar tasks have been around for long enough that spillovers suffice for such new tasks to be rapidly automated even in the absence of human labor.

\setlength{\bibsep}{0pt}
\bibliography{references}

\newpage
\appendix 

\counterwithin{lemma}{section} 
\counterwithin{proposition}{section} 

\begin{center}
    \LARGE{\textbf{APPENDIX TO DATA-DRIVEN AUTOMATION}}
\end{center}
\paragraph{Organization.} The appendix is organized as follows. \cref{appendix:proofs} collects proofs of results in the main text. \cref{appendix:discussion_extensions} provides a more detailed discussion of theoretical extensions alluded in the conclusion. \cref{appendix:simulation_details} presents the details of the solution method behind numerical simulations. \cref{appendix:additional_numerical_results} collects additional numerical results.

\section{Proofs} \label{appendix:proofs}

\subsection{\cref{sec:direction_limit}: Data Autarky}\label{appendix:dataAutarky}

\begin{proof}[ {\bf Proof of \cref{lemma:static eq lemma}}]\label{app:static eq lemma}

The final producer's first-order condition is $Y^{1/\sigma}\cdot y_i^{-1/\sigma} = p_i$, as usual. 

To derive the equilibrium objects in each period $t$, notice that from the final producer's FOC, labor must be equally distributed across the non-automated tasks $(\gamma_t, 1]$, so labor allocation is determined by  Equations \eqref{eqn:labor allocation} in the lemma.

The market clearing condition for capital can be written as
\[K_{i}\int_{0}^{\gamma_{t}}\frac{K_{j}}{K_{i}} dj =K\]

For any two automated $i$ and $j$ the final good producer's FOC implies
\[\frac{\psi_{i}^{K}K_{i}}{\psi_{j}^{K}K_{j}} = \left(\frac{\frac{r}{\psi_{i}^{K}}}{\frac{r}{\psi_{j}^{K}}}\right)^{-\sigma}\implies \frac{K_{i}}{K_{j}}=\left(\frac{\psi_{i}^{K}}{\psi_{j}^{K}}\right)^{\sigma-1}
\]
so capital allocation is therefore given by  Equations \eqref{eqn:capital allocation} in the lemma.

Next, plugging labor and capital allocations into the final producer's FOC and using the fact that intermediate tasks are perfectly competitive, we have that factor prices are thus given by Equations \eqref{eqn:rental rate} and \eqref{eqn:wage} in the lemma.

The automation boundary $\gamma_t$ is then pinned down by a task that is indifferent between using labor and capital to produce i.e., $\psi^K_{\gamma_t} / r = \psi^L / L$ which can be implicitly rewritten as \cref{gamma def} in the lemma.
We will later verify the above expression admits a unique solution. 

It remains to show regularity. From our data LOM,
\begin{equation} \label{eqn:dataLOM_noDep}
    D_{it} = D_{i0} + \int_0^t Y_{is} ds \text{ }\forall i
\end{equation}
Since data is initially continuous and strictly decreasing in task index, by \cref{lemma:static eq lemma}, we have $\dot D_{it} = Y_{i0} \geq Y_{j0} = \dot D_{jt} \quad \forall i \leq j$. This then implies that in a small neighborhood around $t=0$ the initial hypothesis that task data is weakly decreasing in task index survives, allowing us to conclude via the "Abstract Bootstrap Principle" \citep{tao2006nonlinear}.
\end{proof}

\begin{proof}[{\bf Proof of \cref{lemma:direction lemma}}] \label{sec:direction lemma pf}

    We will derive a generalized expression capturing arbitrary spillovers and show it reduces to the expression in the main text when data is task-specific. Rearranging (\ref{gamma def}), we have the following expression for capital productivity of the threshold task on the automation boundary: 
    \begin{equation} \label{gamma def alt}
    \psi_{\gamma_{t}}^{K}=\frac{\psi^{L}L}{K}\frac{1}{1-\gamma_{t}}\int_{0}^{\gamma_{t}}\left(\frac{\psi_{j}^{K}}{\psi_{\gamma_{t}}^{K}}\right)^{\sigma-1}dj
    \end{equation}

    Suppose we run this system forward with $\gamma_{t}=\gamma$ held fixed. If the change in RHS weakly exceeds the change LHS, $\gamma_{t}$ must be weakly decreasing in time given \cref{lemma:static eq lemma} since a larger $\gamma_{t}$ would make the RHS even larger and the LHS even smaller. Similarly, if the change in LHS strictly exceeds the change in RHS, $\gamma_t$ must be strictly increasing in time to restore the equality.
    The partial of LHS w.r.t. time is
    \begin{equation}
    \partial_{\mathcal{A}_{\gamma}}\psi_{\gamma}^{K}\cdot\partial_{t}\mathcal{A}_{\gamma}\left(\boldsymbol{D}_{t}\right)    
    \end{equation}
    Where we abuse notation by using $\partial_{\mathcal{A}_\gamma}$ as the partial derivative w.r.t. effective data. Note that we write $\gamma$ instead of $\gamma_t$ as we are holding it fixed.
    The partial of RHS w.r.t. time is
    \begin{equation}
    \frac{\psi^{L}L}{K}\frac{\sigma-1}{1-\gamma}\int_{0}^{\gamma}\left(\frac{\psi_{j}^{K}}{\psi_{\gamma}^{K}}\right)^{\sigma-2}\frac{\psi_{\gamma}^{K}\left(\partial_{\mathcal{A}_{j}}\psi_{j}^{K}\right)\left(\partial_{t}\mathcal{A}_{j}\left(\boldsymbol{D}_{t}\right)\right)-\psi_{j}^{K}\left(\partial_{\mathcal{A}_{\gamma}}\psi_{\gamma}^{K}\right)\left(\partial_{t}\mathcal{A}_{\gamma}\left(\boldsymbol{D}_{t}\right)\right)}{\left(\psi_{\gamma}^{K}\right)^{2}}dj    \end{equation}
    Rearranging and substituting in (\ref{gamma def alt}), $\gamma_{t}$ is weakly decreasing in time if and only if: 
    \begin{equation}
        \left(\sigma-1\right)\int_{0}^{\gamma}\left(\frac{\psi_{j}^{K}}{\psi_{\gamma}^{K}}\right)^{\sigma-1}\frac{\partial_{t}\log\mathcal{A}_{j}\left(\boldsymbol{D}_{t}\right)}{\partial_{t}\log\mathcal{A}_{\gamma}\left(\boldsymbol{D}_{t}\right)}dj \geq \sigma\int_{0}^{\gamma_{t}}\left(\frac{\psi_{j}^{K}}{\psi_{\gamma}^{K}}\right)^{\sigma-1}dj    
    \end{equation}
    Observe that in the case of no data spillovers, the expression reduces to
    \begin{equation}
    \left(\sigma-1\right)\int_{0}^{\gamma_{t}}\left(\frac{f\left(j\right)}{f\left(\gamma_{t}\right)}\right)^{1/\eta}\left(\frac{\psi_{j}^{K}}{\psi_{\gamma_{t}}^{K}}\right)^{2\sigma-1-1/\eta}dj\geq\sigma\int_{0}^{\gamma_{t}}\left(\frac{\psi_{j}^{K}}{\psi_{\gamma_{t}}^{K}}\right)^{\sigma-1}dj    \end{equation}
    % \question{\begin{itemize}
    %     \item In the big ratio on the lhs, there is $\varepsilon_{j}^{K}$ in both numerator and denominator: 1) this seems a typo or they actually cancel out? 2) $\varepsilon_{j}^{K}$ should be replaced by $\eta$
    %     \item makes ure you always use $\gamma_t$ and NOT $\gamma$ alone. bc $\gamma_t$ is an INDEX, and it makes sense for it to be a superscript.
    % \end{itemize}}
    
    % When $\sigma \leq 1$, this cannot hold, so $\gamma_t$ is strictly increasing in time. We remind the reader that $\varepsilon_{j}^{K}:=\partial\log\psi_{j}^{K}/\partial\log D_{j}$ is the elasticity of productivity w.r.t. data.
\end{proof}

\begin{proof}[{\bf Proof of \cref{prop:complements}}] \label{sec:gross comp pf}

    First, all conditions for  \cref{lemma:static eq lemma}, \cref{lemma:static eq lemma}, and \cref{lemma:direction lemma} are met, so equilibrium characterizations from those results apply. By \cref{lemma:static eq lemma}, we know $\gamma_{t}$ is monotonically increasing since $\sigma\leq1$. We will first show show that $\lim_{t\to\infty}\gamma_{t}=1$ when $\sigma<1/\eta$ for any $W\in\mathcal{\mathcal{W}}$ before establishing that the ratio of data is a positive constant for any two tasks in the limit. 
    \paragraph{Full limit automation (gross complements).}
    We can rewrite (\ref{gamma def}) as
    \begin{align}
    \psi_{\gamma_{t}}^{K} & =\frac{\psi^{L}L}{K}\frac{1}{1-\gamma_{t}}\int_{0}^{\gamma_{t}}\left(\frac{\psi_{j}^{K}}{\psi_{\gamma_{t}}^{K}}\right)^{\sigma-1}dj \leq \frac{\psi^{L}L}{K}\frac{\gamma_{t}}{1-\gamma_{t}}
    \end{align}
    Here the inequality follows from the fact that $\psi_{j}^{K}\geq\psi_{\gamma_{t}}^{K}$ and $\sigma\leq1$. Since $\psi_{\gamma_{t}}^{K}\geq\psi_{1}^{K}$ by \cref{lemma:static eq lemma}, it suffices to show that
    \begin{equation}
    \lim_{t\to\infty}\psi_{1}^{K}\left(\mathcal{A}_{1}\left(\boldsymbol{D}_{t}\right)\right)=\infty
    \end{equation}
    This is indeed the case since all tasks have the outside option of producing with labor, which yields a task quantity of 
    \begin{equation}
    \psi_{L}\frac{L}{1-\gamma_{t}}\geq\psi^{L}\frac{L}{1-\gamma_{0}}>0
    \end{equation}
    From our data LOM it is immediate that effective data for task 1--and therefore all tasks--go to infinity, so we must approach full automation in the limit.

    \paragraph*{Full limit automation (gross substitutes with $\sigma\protect\leq1/\eta$).}
    Again, we know $\psi_{\gamma_{t}}^{K}$ can be written equivalently as
    \begin{equation}
    \psi_{\gamma_{t}}^{K}=\frac{\psi^{L}L}{K}\frac{1}{1-\gamma_{t}}\int_{0}^{\gamma_{t}}\left(\frac{\psi_{j}^{K}}{\psi_{\gamma_{t}}^{K}}\right)^{\sigma-1}dj
    \end{equation}
    The same argument from the gross complements case shows that the LHS goes to infinity as $t\to\infty$. Towards a contradiction, suppose $\gamma_{t}$ is bounded above by some $U<1$,
    \begin{align}
    \frac{\psi^{L}L}{K}\frac{1}{1-\gamma_{t}}\int_{0}^{\gamma_{t}}\left(\frac{\psi_{j}^{K}}{\psi_{\gamma_{t}}^{K}}\right)^{\sigma-1}dj 
    % \leq\frac{\psi^{L}L}{K}\frac{1}{1-U}\int_{0}^{\gamma_{t}}\left(\frac{\psi_{j}^{K}}{\psi_{\gamma_{t}}^{K}}\right)^{\sigma-1}dj 
    \leq\frac{\psi^{L}L}{K}\frac{1}{1-U}\left(\frac{\psi_{0}^{K}}{\psi_{1}^{K}}\right)^{\sigma-1}
    \end{align}
    Here the inequality uses $\sigma\geq1$ and the monotonicity lemma. It suffices to now show $\psi_{0}^{K}/\psi_{1}^{K}$ is bounded for a contradiction. Differentiating the ratio of productivities w.r.t. time,
    {\allowdisplaybreaks \begin{align}
    \frac{d\ln\left(\psi_{0}^{K}/\psi_{1}^{K}\right)}{dt} & =\frac{\left(\psi_{0}^{K}\right)^{\prime}}{\psi_{0}^{K}}\int_{0}^{1}W\left(0,j\right)Y_{j}dj-\frac{\left(\psi_{1}^{K}\right)^{\prime}}{\psi_{1}^{K}}\int_{0}^{1}W\left(1,j\right)Y_{j}dj \nonumber \\
     & =\eta\left(\frac{f\left(0\right)}{\psi_{0}^{K}}\right)^{\frac{1}{\eta}}\int_{0}^{1}W\left(0,j\right)Y_{j}dj-\eta\left(\frac{f\left(1\right)}{\psi_{1}^{K}}\right)^{\frac{1}{\eta}}\int_{0}^{1}W\left(1,j\right)Y_{j}dj \nonumber \\
     & =\eta\left(\frac{f\left(1\right)}{\psi_{1}^{K}}\right)^{\frac{1}{\eta}}\left(\int_{0}^{1}W\left(1,j\right)Y_{j}dj\right)\left(\left(\frac{\psi_{1}^{K}}{\psi_{0}^{K}}\frac{f\left(0\right)}{f\left(1\right)}\right)^{\frac{1}{\eta}}\frac{\int_{0}^{1}W\left(0,j\right)Y_{j}dj}{\int_{0}^{1}W\left(1,j\right)Y_{j}dj}-1\right)
    \end{align}}
    Here we write $\psi_{i}^{K}$ for brevity, suppressing the time subscript. Note the ratio decreases wherever
    \[
    \left(\frac{\psi_{1,t}^{K}}{\psi_{0,t}^{K}}\frac{f\left(0\right)}{f\left(1\right)}\right)^{\frac{1}{\eta}}\frac{\int_{0}^{1}W\left(0,k\right)Y_{k,t}dk}{\int_{0}^{1}W\left(1,k\right)Y_{k,t}dk}
    \leq 1
    \]
    Note that
    \begin{align}
        \frac{\int_{0}^{1}W\left(0,k\right)Y_{k,t}dk}{\int_{0}^{1}W\left(1,k\right)Y_{k,t}dk}	&\leq\frac{Y_{0,t}}{Y_{1,t}}\frac{\int_{0}^{1}W\left(0,k\right)dk} {\int_{0}^{1}W\left(1,k\right)dk} \\
	    &\leq\left(\frac{\psi_{0}^{K}}{\psi_{1}^{K}}\right)^{\sigma}\frac{\int_{0}^{1}W\left(0,k\right)dk}{\int_{0}^{1}W\left(1,k\right)dk}\quad\text{by \cref{lemma:psi_to_y}}
    \end{align}
    All supporting lemmas are collected at the end of proofs to propositions for which they are first introduced. Applying the upperbound from above, the ratio will also decrease when the following condition holds 
    \begin{equation}
        \left(\frac{f\left(0\right)}{f\left(1\right)}\right)^{\frac{1}{\eta}}\left(\frac{\psi_{0,t}^{K}}{\psi_{1,t}^{K}}\right)^{\frac{\sigma\eta-1}{\eta}}\frac{\int_{0}^{1}W\left(0,k\right)dk}{\int_{0}^{1}W\left(1,k\right)dk}\leq1\iff\frac{\psi_{0,t}^{K}}{\psi_{1,t}^{K}}\geq\left(\frac{\int_{0}^{1}W\left(0,k\right)dk}{\int_{0}^{1}W\left(1,k\right)dk}\left(\frac{f\left(0\right)}{f\left(1\right)}\right)^{\frac{1}{\eta}}\right)^{\frac{\eta}{1-\sigma\eta}}
    \end{equation}
    The RHS is bounded by assumption, so we conclude that
    \begin{equation}
        \frac{\psi_{0}^{K}\left(\boldsymbol{D}_{0,t}\right)}{\psi_{1}^{K}\left(\boldsymbol{D}_{1,t}\right)}\leq\max\left\{ \frac{\psi_{0}^{K}\left(\boldsymbol{D}_{0,0}\right)}{\psi_{1}^{K}\left(\boldsymbol{D}_{1,0}\right)},\left(\frac{\int_{0}^{1}W\left(0,k\right)dk}{\int_{0}^{1}W\left(1,k\right)dk}\left(\frac{f\left(0\right)}{f\left(1\right)}\right)^{\frac{1}{\eta}}\right)^{\frac{\eta}{1-\sigma\eta}}\right\} <\infty\text{ }\forall t\geq 0
    \end{equation}
    This concludes the proof for the full limit automation result. We now turn to establishing that data will be balanced in the limit. 

    \paragraph*{Balanced data.} We now establish that the relative data ratio between any two tasks $i \neq j$ will be balanced. 
    {\allowdisplaybreaks \begin{align*}
        \lim_{t \to \infty}\cfrac{D_{it}}{D_{jt}}  
        &= \lim_{t \to \infty}\cfrac{y_{it}}{y_{jt}}  \tag{L'Hopital's} \\
        &= \lim_{t \to \infty}\cfrac{\psi^K_{it}K_{it}}{\psi^K_{jt}K_{jt}} \tag{tasks will be produced by capital asymptotically} \\
        &= \lim_{t \to \infty} \left(\cfrac{\psi^K_{it}}{\psi^K_{jt}}\right)^{\sigma}
        \tag{$K_{it}=\frac{\left(\psi_{it}^{K}\right)^{\sigma-1}}{\int_{0}^{\gamma_{t}}\left(\psi_{jt}^{K}\right)^{\sigma-1}dj}K$}
        \\
        &= \lim_{t \to \infty} \left(\cfrac{f_i}{f_j}\right)^{\sigma} \left(\cfrac{D_{it}}{D_{jt}}\right)^{\eta \sigma}
    \end{align*}}
    Rearranging gives the expression for $\lim_{t \to \infty} D_{it} / D_{jt}$. Observe that with or without capital accumulation, we will have $r/w \to \infty$. Note that this argument goes through with capital accumulation, since we have full automation so long as capital is bounded away from 0 (See \cref{gamma def}). 

    \begin{lemma}[An upperbound on the ratio of task production] \label{lemma:psi_to_y}
        \begin{equation}
        \frac{Y_{0}}{Y_{j}}\leq\left(\frac{\psi_{0}^{K}}{\psi_{j}^{K}}\right)^{\sigma}
        \end{equation}
    \end{lemma}

    \begin{proof}[{\bf Proof of \cref{lemma:psi_to_y}}]
        We know the most productive task will always be produced with capital, so
        \begin{equation}
        Y_{0}=\frac{\left(\psi_{0}^{K}\right)^{\sigma}}{\int_{0}^{\gamma}\left(\psi_{j}^{K}\right)^{\sigma-1}dj}K
        \end{equation}
        In addition, if task $j$ is produced with capital, the above expression holds with $0$ replaced by $j$. Otherwise, we know
        \begin{equation}
        Y_{j}=Y_{\gamma}=\frac{\left(\psi_{\gamma}^{K}\right)^{\sigma}}{\int_{0}^{\gamma}\left(\psi_{j}^{K}\right)^{\sigma-1}dj}K\geq\frac{\left(\psi_{j}^{K}\right)^{\sigma}}{\int_{0}^{\gamma}\left(\psi_{j}^{K}\right)^{\sigma-1}dj}K
        \end{equation}
        Dividing the two expressions, we have 
        \begin{equation}
        \frac{Y_{0}}{Y_{j}}\leq\left(\frac{\psi_{0}^{K}}{\psi_{j}^{K}}\right)^{\sigma}
        \end{equation}
        as desired.
    \end{proof}
\end{proof}

\begin{proof}[{\bf Proof of \cref{prop:substitutes}}] \label{sec:substitutes ub pf}

   First, all conditions for  \cref{lemma:static eq lemma}, \cref{lemma:static eq lemma}, and \cref{lemma:direction lemma} are met, so equilibrium characterizations from those results apply. We will first show the upperbound result before establishing that data will be imbalanced in the limit.
   
   \paragraph{Bounded automation.} We will prove
   \begin{equation} \gamma_{t}\leq\max\left\{ \gamma_{0},L_{f,\sigma}\right\} <1 \text{ }\forall t \geq 0 
   \end{equation} 
   We do so by establishing that $\gamma_{t}\geq L_{f,\sigma}$ implies $\gamma_{t}$ is weakly decreasing in time. We remind the reader $L_{f,\sigma}$ comes from the $\sigma$-regularity condition, which stipulates that there exists some constant $L_{\sigma, f} \in (0, 1)$ such that
    \begin{equation} 
        \frac{\sigma}{\sigma-1}f(i)
        < \frac{1}{i}\int_{0}^{i}f(j)dj \quad \text{for all $i \geq L_{\sigma, f}$.} 
    \end{equation}
    From \cref{lemma:direction lemma}, we know $\gamma_{t}$ is weakly decreasing in time if
    {\allowdisplaybreaks \begin{equation}
     \left(\sigma-1\right)\int_{0}^{\gamma_{t}}\left(\frac{f\left(j\right)}{f\left(\gamma_{t}\right)}\right)^{1/\eta}\left(\frac{\psi_{j}^{K}}{\psi_{\gamma_{t}}^{K}}\right)^{2\sigma-1-1/\eta}dj\geq\sigma\int_{0}^{\gamma_{t}}\left(\frac{\psi_{j}^{K}}{\psi_{\gamma_{t}}^{K}}\right)^{\sigma-1}dj 
    \end{equation}}
    % Applying our functional form assumption and specializing to the case with no data spillovers (i.e., $g=\delta_{\{0\}}$), we have
    % \begin{equation}
    % \int_{0}^{\gamma_{t}}\left(\frac{\psi_{j}^{K}}{\psi_{\gamma_{t}}^{K}}\right)^{\sigma-1}\frac{f\left(j\right)}{f\left(\gamma_{t}\right)}\underbrace{\left(\frac{\psi_{j}^{K}}{\psi_{\gamma_{t}}^{K}}\right)^{\sigma-\frac{1}{\eta}}}_{\geq1\text{ by \cref{lemma:static eq lemma}}}dj\geq\frac{\sigma}{\sigma-1}\int_{0}^{\gamma_{t}}\left(\frac{\psi_{j}^{K}}{\psi_{\gamma_{t}}^{K}}\right)^{\sigma-1}dj
    % \end{equation}
    Observe that 
    \begin{align}
    \int_{0}^{\gamma_{t}}f\left(j\right)\left(\frac{\psi_{j}^{K}}{\psi_{\gamma_{t}}^{K}}\right)^{\sigma-1}dj & \geq\cfrac{1}{\gamma_{t}}\left(\int_{0}^{\gamma_{t}}f(j)dj\right)\left(\int_{0}^{\gamma_{t}}\left(\frac{\psi_{j}^{K}}{\psi_{\gamma_{t}}^{K}}\right)^{\sigma-1}dj\right)\\
     & \geq\frac{\sigma}{\sigma-1}f\left(\gamma_{t}\right)\int_{0}^{\gamma_{t}}\left(\frac{\psi_{j}^{K}}{\psi_{\gamma_{t}}^{K}}\right)^{\sigma-1}dj
    \end{align}
    where the first inequality follows from Chebyshev/FKG (see, e.g., Theorem 1 in \cite{armstrong1993chebyshev}); the second is due to $\pmb{f}$ being $\sigma$-regular and $\gamma_{t}\geq L_{f,\sigma}$ by assumption. Dividing both sides of the expression above by $f(\gamma_t)$ concludes the proof since $f(j) \geq f(\gamma_t) > 0 \text{ } \forall j \leq \gamma_t$ and $1/\eta > 1$.
    
    \paragraph{Imbalanced data.}
    It is easy to verify $D_{it}/D_{jt}$ is monotonically increasing over time, so if $D_{it}/D_{jt}$ does not diverge, the limit exists by monotone convergence. First, if $j$ is always produced using capital after some time $T$, then following the same steps as before
    \[
    \lim_{t\to\infty}\frac{D_{0t}}{D_{jt}}=\left(\frac{f_{0}}{f_{j}}\right)^{\sigma}\lim_{t\to\infty}\left(\frac{D_{0t}}{D_{jt}}\right)^{\eta\sigma}
    \]
    Suppose the ratio is bounded, then the limit exists, but the RHS is strictly larger than the LHS. Moreover, we know $D_{0t}/D_{jt}\geq1$, so the inequality yields a contradiction. Therefore the limit must diverge to infinity in this case. On the other hand, if $j$ is produced only using labor only after some time $T$, then
    \begin{align*}
    \lim_{t\to\infty}\frac{D_{0,t}}{D_{jt}} & =\lim_{t\to\infty}\frac{Y_{0,t}}{Y_{jt}}\\
     & \geq\lim_{t\to\infty}\frac{Y_{0,t}}{\psi^{L}L/\left(1-j\right)}\tag{\text{all tasks above $j$ must be produced using labor}}\\
     & \geq\lim_{t\to\infty}\frac{1}{\psi^{L}L/\left(1-j\right)}\underbrace{\left(\frac{\left(\psi_{0,t}^{K}\right)^{\sigma-1}}{\int_{0}^{\gamma_{t}}\left(\psi_{k,t}^{K}\right)^{\sigma-1}dk}\right)}_{\geq1}\psi_{0,t}^{K}\\
     & \geq\lim_{t\to\infty}\frac{1}{\psi^{L}L/\left(1-j\right)}\psi_{0,t}^{K}=\infty
    \end{align*}
    This concludes the proof.
\end{proof}

\begin{proof}[{\bf Proof of \cref{prop:wages}}] \label{proof:wages}
Recall wage is given by
\[
w_{t}=Y_{t}^{\frac{1}{\sigma}}\left(\frac{\left(1-\gamma_{t}\right)\left(\psi^{L}\right)^{\sigma-1}}{L}\right)^{\frac{1}{\sigma}}
\]
so the long-run wage is governed by $Y_{t}\left(1-\gamma_{t}\right)$ as the rest are constants. Using our expression for output, we have
\[
Y_{t}\left(1-\gamma_{t}\right)=\left(\gamma_{t}^{\frac{1}{\sigma}}\left(\left(1-\gamma_{t}\right)\Psi_{K}\left(\gamma_{t},\boldsymbol{\psi}_{t}\right)K_{t}\right)^{\frac{\sigma-1}{\sigma}}+\left(1-\gamma_{t}\right)\left(\psi_{L}L\right)^{\frac{\sigma-1}{\sigma}}\right)^{\frac{\sigma}{\sigma-1}}.
\]
First consider the case where $\sigma \leq 1/\eta$. Note that the second term is going to zero in the limit. As for the first term, we need to evaluate $\left(1-\gamma_{t}\right)\Psi_{K}\left(\gamma_{t},\boldsymbol{\psi}_{t}\right)K_{t}$. Using our implicit definition of $\gamma_{t}$, we can write
\[
1-\gamma_{t}=\frac{1}{\psi_{\gamma_{t}}^{K}}\frac{\psi^{L}L}{K_{t}}\int_{0}^{\gamma_{t}}\left(\frac{\psi_{j}^{K}}{\psi_{\gamma_{t}}^{K}}\right)^{\sigma-1}dj
\]
So
\[
\left(1-\gamma_{t}\right)\Psi_{K}\left(\gamma_{t},\boldsymbol{\psi}_{t}\right)K_{t}=\psi^{L}L\frac{\Psi_{K}\left(\gamma_{t},\boldsymbol{\psi}_{t}\right)}{\psi_{\gamma_{t}}^{K}}\int_{0}^{\gamma_{t}}\left(\frac{\psi_{j}^{K}}{\psi_{\gamma_{t}}^{K}}\right)^{\sigma-1}dj
\]
When $\sigma \leq 1/\eta$, we have balanced data, so
\[
\lim_{t\to\infty}\frac{\Psi_{K}\left(\gamma_{t},\boldsymbol{\psi}_{t}\right)}{\psi_{\gamma_{t}}^{K}}\int_{0}^{\gamma_{t}}\left(\frac{\psi_{j}^{K}}{\psi_{\gamma_{t}}^{K}}\right)^{\sigma-1}dj=C<\infty
\]
On the other hand, when $\sigma > 1/\eta$, we know $\lim_{t\to\infty}Y_t = \infty$ and $\gamma_t$ will be upperbounded per \cref{prop:substitutes}, so the long-run wage therefore diverges to infinity. Note that this argument only requires the balanced data result from \cref{prop:complements}, which holds with or without capital accumulation as noted in the proof.

\end{proof}

%\subsubsection{Proof of \cref{prop:heterogeneous-speed}}

\begin{proof}[{\bf Proof of \cref{prop:heterogeneous-speed}}]
Under our maintained assumpitons on the boundedness of $\pmb{f}$ and $\pmb{D}_0$, there are
constants \(\underline f,\overline f,\underline D_0,\overline D_0>0\) such that
\[
    \underline f\leq f_i\leq \overline f,
    \qquad
    \underline D_0\leq D_{i0}\leq \overline D_0
    \qquad \forall i\in X.
\]

We first argue that this delivers a
whole-path bound. Fix \(i<j\) and write \(x_{ij,t}:=D_{it}/D_{jt}\). The
balanced-data ratio is
\[
    x^*_{ij}:=
    \left(\frac{f_i}{f_j}\right)^{\frac{\sigma}{1-\sigma\eta}} .
\]
Suppose \(1\leq x_{ij,0}\leq x^*_{ij}\). We claim that
\[
    1\leq x_{ij,t}\leq x^*_{ij}
    \qquad \forall t\geq0.
\]
To see this, observe that \(x_{ij,t}\) can cross the upper boundary only if
\(\dot x_{ij,t}>0\) when \(x_{ij,t}=x^*_{ij}\). If both tasks are automated,
then
\[
    \frac{y_{it}}{y_{jt}}
    =
    \left(\frac{\psi_i^K}{\psi_j^K}\right)^\sigma
    =
    \left(\frac{f_i}{f_j}\right)^\sigma x_{ij,t}^{\eta\sigma}.
\]
Hence
\[
    \frac{\dot x_{ij,t}}{x_{ij,t}}
    =
    \frac{y_{it}}{D_{it}}-\frac{y_{jt}}{D_{jt}}
    =
    \frac{y_{jt}}{D_{jt}}
    \left[
    \left(\frac{f_i}{f_j}\right)^\sigma
    x_{ij,t}^{\eta\sigma-1}
    -1
    \right],
\]
which is zero at \(x_{ij,t}=x^*_{ij}\) and negative above it because
\(\eta\sigma<1\). If both tasks are produced by labor, then \(y_{it}=y_{jt}\),
so the ratio moves toward one. Finally, if \(i\) is automated and \(j\) is
produced by labor, then \(y_{it}/y_{jt}\leq(\psi_i^K/\psi_j^K)^\sigma\), so the
same upper-bound argument applies. Thus \(x_{ij,t}\) cannot cross above
\(x^*_{ij}\). A symmetric argument at \(x_{ij,t}=1\), using \(f_i\geq f_j\) and
the monotonicity of capital productivity, shows it cannot cross below one.

It follows that, for every \(i<j\) and every \(t\geq0\),
\[
    \frac{\psi_i^K}{\psi_j^K}
    =
    \frac{f_i}{f_j}x_{ij,t}^{\eta}
    \leq
    \left(\frac{f_i}{f_j}\right)^{\frac{1}{1-\sigma\eta}}.
\]
Since \(f_i\) is bounded above and away from zero, the right-hand side is
uniformly bounded. In particular, there exists \(R<\infty\) such that
\[
    1\leq
    \frac{\psi^K_j}{\psi^K_{\gamma_t}}
    \leq R
    \qquad \forall j\leq\gamma_t,\; t\geq0.
\]
Under this stronger initial-condition restriction, we can take \(T_R=0\) in the
argument below. Since the automation boundary is weakly increasing along the
path, we can also take the lower bound on \(\gamma_t\) to be \(\gamma_0\), and
therefore take \(T=0\).
The rest of the proof is written for a generic date \(T\). In the
non-asymptotic case just described, the same argument applies with
\(T=0\), \(\gamma_\star=\gamma_0\),
\(\underline B=\underline D_0\), and \(\overline B=\overline D_0\).

Since \(\sigma<1/\eta\), \cref{prop:complements} implies full limit automation
and balanced data. Since \(f_i\) is bounded above and away from zero, the
bounded relative-productivity conclusion in \cref{prop:complements} gives an
\(R<\infty\) and a finite time \(T_R\) such that for every \(t\geq T_R\) and
every automated task \(j\leq\gamma_t\),
\[
    1\leq
    \frac{\psi^K_j}{\psi^K_{\gamma_t}}
    \leq R.
\]
The lower bound uses monotonicity of capital productivity in the task index.
Fix any \(\gamma_\star\in(0,1)\).
Full limit automation gives a finite time \(T_\gamma(\gamma_\star)\) such that
\(\gamma_t\geq \gamma_\star\) for all \(t\geq T_\gamma(\gamma_\star)\). Take
\(T=\max\{T_R,T_\gamma(\gamma_\star)\}\). The bounds below are non-asymptotic
for every \(t\geq T\), but the date \(T\) is determined by the eventual
convergence results in \cref{prop:complements}. Thus \(\gamma_\star\) is
arbitrary, but fixing a larger \(\gamma_\star\) may require a later
\(T_\gamma(\gamma_\star)\).

Define
\[
    M_t:=
    \int_0^{\gamma_t}
    \left(\frac{\psi^K_j}{\psi^K_{\gamma_t}}\right)^{\sigma-1}dj .
\]
The bounded-ratio condition implies the stated bounds on \(M_t\). If
\(\sigma\geq 1\), the integrand in \(M_t\) lies between \(1\) and
\(R^{\sigma-1}\), so
\[
    \gamma_\star
    \leq \gamma_t
    \leq M_t
    \leq \gamma_t R^{\sigma-1}
    \leq R^{\sigma-1}.
\]
If \(\sigma<1\), the integrand lies between \(R^{\sigma-1}\) and \(1\), so
\[
    \gamma_\star R^{\sigma-1}
    \leq \gamma_tR^{\sigma-1}
    \leq M_t
    \leq \gamma_t
    \leq 1.
\]
Thus \(0<\underline M\leq M_t\leq\overline M<\infty\).

By the static equilibrium characterization, the automation boundary satisfies
\[
    \psi_{\gamma_t}^{K}
    =
    \left(
        \frac{L}{K(1-\gamma_t)}
        \int_0^{\gamma_t}(\psi_j^K)^{\sigma-1}dj
    \right)^{1/\sigma}.
\]
Equivalently,
\[
    1-\gamma_t
    =
    \frac{L}{K}
    \frac{M_t}{\psi_{\gamma_t}^{K}}.
\]
Since the automation boundary is weakly increasing, the marginal task at time
\(t\) has accumulated data from labor production until it reaches the automation
boundary. Define cumulative labor-produced data for not-yet-automated tasks by
\[
    H_t:=\int_0^t\frac{L}{1-\gamma_s}ds.
\]
Then
\[
    D_{\gamma_t,t}=D_{\gamma_t,0}+H_t,
\]
and the boundedness of the initial data stock implies
\[
    \underline B_t
    :=
    \underline D_0+H_t
    \leq
    D_{\gamma_t,t}
    \leq
    \overline D_0+H_t
    =:
    \overline B_t.
\]
Let
\[
    \underline B:=\underline B_T,
    \qquad
    \overline B:=\overline B_T.
\]
Therefore
\[
    \psi_{\gamma_t}^{K}=f_{\gamma_t}D_{\gamma_t,t}^\eta.
\]
It follows that
\[
    1-\gamma_t
    =
    \frac{L}{K}\frac{M_t}{f_{\gamma_t}}D_{\gamma_t,t}^{-\eta}.
\]
Using \(f_{\gamma_t}\in[\underline f,\overline f]\), the bounds on \(M_t\), and
the envelope for \(D_{\gamma_t,t}\),
\[
    \frac{L\underline M}{K\overline f}\overline B_t^{-\eta}
    \leq
    1-\gamma_t
    \leq
    \frac{L\overline M}{K\underline f}\underline B_t^{-\eta}.
\]
Since \(\dot H_t=L/(1-\gamma_t)\),
\(\dot{\underline B}_t=\dot{\overline B}_t=\dot H_t\).
The previous inequalities imply
\[
    \frac{K\underline f}{\overline M}\underline B_t^\eta
    \leq
    \dot{\underline B}_t
    =
    \dot{\overline B}_t
    \leq
    \frac{K\overline f}{\underline M}\overline B_t^\eta .
\]
Compare \(\underline B_t\) and \(\overline B_t\) to the two ODEs
\[
    \dot B^-_t=\frac{K\underline f}{\overline M}(B^-_t)^\eta,
    \qquad
    \dot B^+_t=\frac{K\overline f}{\underline M}(B^+_t)^\eta,
    \qquad
    B^-_T=\underline B,\quad B^+_T=\overline B .
\]
Because \(0<\eta<1\), these comparison equations have closed-form solutions
\[
    B^-_t
    =
    \left(
        \underline B^{1-\eta}
        +(1-\eta)\frac{K\underline f}{\overline M}(t-T)
    \right)^{\frac{1}{1-\eta}},
\]
and
\[
    B^+_t
    =
    \left(
        \overline B^{1-\eta}
        +(1-\eta)\frac{K\overline f}{\underline M}(t-T)
    \right)^{\frac{1}{1-\eta}}.
\]
The comparison theorem gives \(B^-_t\leq \underline B_t\) and
\(\overline B_t\leq B^+_t\). Combining this with the previous bounds on
\(1-\gamma_t\) yields
\[
    1-\gamma_t
    \geq
    \frac{L\underline M}{K\overline f}(B^+_t)^{-\eta},
    \qquad
    1-\gamma_t
    \leq
    \frac{L\overline M}{K\underline f}(B^-_t)^{-\eta},
\]
which is the desired pair of power-law bounds.
Finally, consider the stronger initial-condition restriction in the statement.
The invariant-region argument at the beginning of the proof gives a constant
\(R<\infty\) such that
\[
    1\leq
    \frac{\psi^K_j}{\psi^K_{\gamma_t}}
    \leq R
    \qquad \forall j\leq\gamma_t,\; t\geq0.
\]
Thus no asymptotic appeal to \(\cref{prop:complements}\) is needed to obtain the
relative-productivity bound. Since \(\gamma_t\) is weakly increasing,
\(\gamma_t\geq\gamma_0\) for all \(t\geq0\), and therefore
\[
    0<\underline M\leq M_t\leq \overline M<\infty
    \qquad \forall t\geq0,
\]
where one may take
\[
    (\underline M,\overline M)
    =
    \begin{cases}
    (\gamma_0,R^{\sigma-1}), & \sigma\geq1,\\
    (\gamma_0R^{\sigma-1},1), & \sigma<1.
    \end{cases}
\]
Moreover \(H_0=0\), so the initial envelopes are simply
\(\underline B_0=\underline D_0\) and
\(\overline B_0=\overline D_0\). Repeating the same comparison argument above
with \(T=0\), \(\underline B=\underline D_0\), and
\(\overline B=\overline D_0\) gives the displayed bounds for every \(t\geq0\).
\end{proof}
%
%\begin{comment}
%\subsubsection{Proof of \cref{prop:speed_no_acc}}
%
%\textbf{Simplifying Assumptions}
%The follow simplifying assumptions are maintained through out. All sectors are symmetric, so $\boldsymbol{f}\equiv f$ and $\boldsymbol{D}_{0}=D_{0}$. $\psi^{L}$ is normalized to $1$ for simplicity. Given the symmetry, the amount of data will be identical across tasks, and we use $D_{t}$ to denote the prevailing data level at time $t$.
%
%\textbf{No capital accumulation}
%Set $K_{t}=K$ fixed for simplicity as well. Immediately from equilibrium clearing, we have a relation between data and the fraction of automated sectors
%\[
%\gamma_{t}=1-\frac{L}{K}D_{t}^{-\eta}
%\]
%Since all tasks will be produced in equal quantities, the speed of data production is equal to the amount of data produced by labor
%\[
%\dot{D}_{t}=Y_{t}=\frac{L}{1-\gamma_{t}}
%\]
%Let $\chi_{t}\equiv1-\gamma_{t}$ be the share of tasks performed by labor. We have
%\begin{align*}
%\chi_{t} & =\frac{L}{K}D_{t}^{-\eta}\\
%\dot{D}_{t} & =\frac{L}{\chi_{t}}
%\end{align*}
%The ODE admits a closed form solution
%\[
%D_{t}=\left(D_{0}^{1-\eta}+\left(1-\eta\right)t\right)^{\frac{1}{1-\eta}}
%\]
%Immediately
%\[
%\chi_{t}=\frac{L}{K}\left(D_{0}^{1-\eta}+\left(1-\eta\right)t\right)^{-\frac{\eta}{1-\eta}}
%\]
%We conclude that the fraction of automated sectors follow an approximately power-law relationship with respect to time. 
%\end{comment}

\subsection{\cref{section:contagious}: Contagious  Automation}\label{app:contagion}

%\question{\small this subsection goes}} \label{app:contagion}

%%%%%%%%%%%%%%%%%%%%%%%%%%%%%%

\begin{proof}[{\bf Proof of \cref{prop:generalW_full_automation}}]
The argument is to show that connected spillovers eventually prevent any task
from remaining permanently data-poor. We first show that aggregate raw data
grows without bound. We then use uniform strong connectedness to propagate a
lower bound on effective data through the graphon in finitely many doubling
intervals. This eventually makes every task's effective data comparable to
aggregate raw data. Hence all capital productivities diverge and remain
uniformly comparable. The final step rules out a positive-measure set of
labor-produced tasks, because capital eventually becomes strictly cheaper than
labor on almost every task.

For \(S\subseteq X\) and \(q>0\), define
\[
    \Gamma_q(S)
    :=
    \left\{
    i\in X:
    \int_S W(i,j)\,dj\ge q
    \right\}.
\]

\textbf{\underline{Step 1: Showing aggregate raw data diverges.}}
Let
\[
    M_t:=\int_XD_{i,t}\,di
\]
denote aggregate raw data. A labor-only allocation with \(L_i=L\) for all \(i\)
is feasible and delivers final output \(\psi^L L\). Thus equilibrium output is
bounded below by \(\psi^L L\). Since the CES aggregate is no larger than
\(\int_X y_{i,t}\,di\) on the unit-measure task space,
\[
    \dot M_t=\int_X y_{i,t}\,di\ge \psi^L L.
\]
Therefore \(M_t\to\infty\).

\textbf{\underline{Step 2: Propagating effective data through the network.}}
We first prove a simple implication: if task \(i\) has effective data at least a
fixed fraction of aggregate data, then task \(i\)'s output is at least a fixed
fraction of aggregate raw-data growth.

Fix \(b\in(0,1]\), and define
\[
    \lambda(b):=\frac{\underline f}{\overline f}b^\eta .
\]
Since \(b\le1\), we have \(\lambda(b)\le1\). Suppose
\[
    \mathcal A_i(\bm D_t)\ge bM_t.
\]
Because \(W\le1\),
\[
    \mathcal A_j(\bm D_t)\le M_t
    \qquad\text{for a.e. }j.
\]
Thus, using \(\psi^K_{j,t}=f_j\mathcal A_j(\bm D_t)^\eta\),
\[
    \psi^K_{i,t}
    \ge
    \underline f b^\eta M_t^\eta
    =
    \lambda(b)\overline f M_t^\eta
    \ge
    \lambda(b)\psi^K_{j,t}
    \qquad\text{for a.e. }j.
\]
This is a comparison of capital productivities as technologies; it does not
assume that either task is actually produced by capital. Let
\[
    \chi_{j,t}:=
    \min\left\{\frac{r_t}{\psi^K_{j,t}},\frac{w_t}{\psi^L}\right\}
\]
be the unit cost of task \(j\). The productivity comparison implies
\[
    \frac{r_t}{\psi^K_{i,t}}
    \le
    \lambda(b)^{-1}\frac{r_t}{\psi^K_{j,t}}.
\]
Since \(\lambda(b)\le1\),
\[
    \chi_{i,t}
    \le
    \lambda(b)^{-1}\chi_{j,t}
    \qquad\text{for a.e. }j.
\]
CES demand therefore gives
\begin{equation}\label{eq:contagion-output-lower}
    y_{i,t}
    \ge
    \lambda(b)^\sigma y_{j,t}
    \qquad\text{for a.e. }j.
\end{equation}
Integrating \cref{eq:contagion-output-lower} over \(j\), we obtain
\begin{equation}\label{eq:contagion-output-aggregate-lower}
    y_{i,t}
    \ge
    \lambda(b)^\sigma\dot M_t.
\end{equation}

We now repeatedly use this one-step implication. It will be convenient to do a time change by using aggregate-data doubling time: over each interval on which
\(M_t\) doubles, any task with effective data comparable to \(M_t\) produces a
fixed fraction of the new aggregate raw data. This allows the lower
bound to move from one spillover layer to the next. Fix \(T_0\) with \(M_{T_0}>0\). Define the following doubling times
\[
    \tau_m:=\inf\{t\ge T_0:M_t\ge 2^mM_{T_0}\}
    \quad \text{for $m=0,1,2,\ldots$.}
\]
Then
\[
    M_{\tau_{m+1}}=2M_{\tau_m}.
\]

For each \(m\), define
\[
    h_r^m(i):=
    \frac{1}{M_{\tau_m}}
    \int_X W^{(r)}(i,j)D_{j,\tau_m}\,dj,
    \qquad r=1,\ldots,n.
\]
Since \(0\le W\le1\) and \(\mu(X)=1\), induction gives
\(0\le W^{(r)}\le1\) for every \(r\). Hence
\[
    0\le h_r^m(i)\le1
    \qquad\text{for a.e. }i.
\]
By the recursive definition of \(W^{(r)}\) and Tonelli's theorem,
\begin{equation}\label{eq:contagion-h-recursion}
    h_{r+1}^m(i)
    =
    \int_X W(i,k)h_r^m(k)\,dk.
\end{equation}

Uniform strong connectedness gives \(h_n^m(i)\ge\varepsilon\) for a.e. \(i\).
Set \(a_n:=\min\{\varepsilon,1\}\), and recursively set
\[
    a_r:=\frac{a_{r+1}}{2},
    \qquad r=n-1,\ldots,1.
\]
Let
\[
    S_r^m:=\{i\in X:h_r^m(i)\ge a_r\}.
\]
Then \(S_n^m=X\) up to a null set, and
\[
    \mathcal A_i(\bm D_{\tau_m})
    \ge
    a_1M_{\tau_m}
    \qquad\text{for a.e. }i\in S_1^m.
\]

We next show that if a task has high \((r+1)\)-step exposure, then a nontrivial amount of its direct spillover weight must fall on tasks with high \(r\)-step
exposure:
\begin{equation}\label{eq:contagion-layer-inclusion}
    S_{r+1}^m\subseteq \Gamma_{a_r}(S_r^m)
    \qquad\text{up to null sets, for each } r<n.
\end{equation}
Indeed, if \(i\in S_{r+1}^m\), then by \cref{eq:contagion-h-recursion},
\[
    a_{r+1}
    \le
    h_{r+1}^m(i)
    =
    \int_X W(i,k)h_r^m(k)\,dk.
\]
On \(X\setminus S_r^m\), \(h_r^m<a_r\), while on \(S_r^m\), \(h_r^m\le1\).
Therefore
\[
    h_{r+1}^m(i)
    \le
    a_r+\int_{S_r^m}W(i,k)\,dk.
\]
Since \(a_{r+1}=2a_r\), it follows that
\[
    \int_{S_r^m}W(i,k)\,dk\ge a_r,
\]
which proves \cref{eq:contagion-layer-inclusion}. The sets \(S_r^m\) identify where the spillover exposure is large. We now choose constants \(b_r\) so that membership in \(S_r^m\) implies an actual lower bound on effective data after \(r-1\) doubling intervals. Define constants \(b_1,\ldots,b_n\) by
\[
    b_1:=a_1,
    \qquad
    b_{r+1}:=
    \frac{a_r}{2}
    \left[
    \lambda\left(\frac{b_r}{2}\right)
    \right]^\sigma,
    \qquad r=1,\ldots,n-1.
\]
By construction, \(0<b_r\le1\) for every \(r\). We claim that for each \(r=1,\ldots,n\),
\begin{equation}\label{eq:contagion-induction-claim}
    \mathcal A_i(\bm D_{\tau_{m+r-1}})
    \ge
    b_rM_{\tau_{m+r-1}}
    \qquad\text{for a.e. }i\in S_r^m.
\end{equation}
The case \(r=1\) follows from the definition of \(S_1^m\) and \(b_1=a_1\). Suppose \cref{eq:contagion-induction-claim} holds for some \(r<n\). For
\(u\in[\tau_{m+r-1},\tau_{m+r}]\), data is nondecreasing and
\(M_u\le2M_{\tau_{m+r-1}}\). Hence, for a.e. \(i\in S_r^m\),
\[
    \mathcal A_i(\bm D_u)
    \ge
    \frac{b_r}{2}M_u.
\]
Applying \cref{eq:contagion-output-aggregate-lower} with \(b=b_r/2\), we get
\[
    y_{i,u}
    \ge
    \left[
    \lambda\left(\frac{b_r}{2}\right)
    \right]^\sigma
    \dot M_u
    \qquad\text{for a.e. }i\in S_r^m.
\]
Integrating over \(u\in[\tau_{m+r-1},\tau_{m+r}]\) and dropping the
nonnegative initial data term,
\[
    D_{i,\tau_{m+r}}
    \ge
    \left[
    \lambda\left(\frac{b_r}{2}\right)
    \right]^\sigma
    \left(M_{\tau_{m+r}}-M_{\tau_{m+r-1}}\right).
\]
Since \(M_{\tau_{m+r}}=2M_{\tau_{m+r-1}}\),
\begin{equation}\label{eq:contagion-raw-data-lower}
    D_{i,\tau_{m+r}}
    \ge
    \frac12
    \left[
    \lambda\left(\frac{b_r}{2}\right)
    \right]^\sigma
    M_{\tau_{m+r}}
    \qquad\text{for a.e. }i\in S_r^m.
\end{equation}

Now take \(i\in S_{r+1}^m\). By \cref{eq:contagion-layer-inclusion},
\[
    \int_{S_r^m}W(i,k)\,dk\ge a_r.
\]
Using \cref{eq:contagion-raw-data-lower},
\[
\begin{aligned}
    \mathcal A_i(\bm D_{\tau_{m+r}})
    &=
    \int_X W(i,k)D_{k,\tau_{m+r}}\,dk \\
    &\ge
    \int_{S_r^m} W(i,k)D_{k,\tau_{m+r}}\,dk \\
    &\ge
    \frac12
    \left[
    \lambda\left(\frac{b_r}{2}\right)
    \right]^\sigma
    M_{\tau_{m+r}}
    \int_{S_r^m}W(i,k)\,dk \\
    &\ge
    \frac{a_r}{2}
    \left[
    \lambda\left(\frac{b_r}{2}\right)
    \right]^\sigma
    M_{\tau_{m+r}} \\
    &=
    b_{r+1}M_{\tau_{m+r}}.
\end{aligned}
\]
This gives \cref{eq:contagion-induction-claim}.

Since \(S_n^m=X\) up to a null set, \cref{eq:contagion-induction-claim} with
\(r=n\) gives, for every \(m\),
\begin{equation}\label{eq:contagion-doubling-bound}
    \mathcal A_i(\bm D_{\tau_{m+n-1}})
    \ge
    b_nM_{\tau_{m+n-1}}
    \qquad\text{for a.e. }i\in X.
\end{equation}

We now extend the bound from doubling times to all sufficiently large times.
Take any \(t\ge\tau_{n-1}\), and let \(s\ge n-1\) be such that
\(t\in[\tau_s,\tau_{s+1}]\). Set
\[
    m=s-n+1.
\]
Then \(m+n-1=s\), so \cref{eq:contagion-doubling-bound}, applied with this
starting \(m\), gives
\[
    \mathcal A_i(\bm D_{\tau_s})
    \ge
    b_nM_{\tau_s}
    \qquad\text{for a.e. }i\in X.
\]
Since data is nondecreasing,
\[
    \mathcal A_i(\bm D_t)
    \ge
    \mathcal A_i(\bm D_{\tau_s}).
\]
And since \(t\in[\tau_s,\tau_{s+1}]\),
\[
    M_t\le M_{\tau_{s+1}}=2M_{\tau_s}.
\]
Therefore
\[
    \mathcal A_i(\bm D_t)
    \ge
    \frac{b_n}{2}M_t
    \qquad\text{for a.e. }i\in X.
\]
Since \(W\le1\), also
\[
    \mathcal A_i(\bm D_t)
    \le
    M_t
    \qquad\text{for a.e. }i\in X.
\]
Thus, for all \(t\ge\tau_{n-1}\),
\begin{equation}\label{eq:contagion-effective-data-comparable}
    \frac{b_n}{2}M_t
    \le
    \mathcal A_i(\bm D_t)
    \le
    M_t
    \qquad\text{for a.e. }i\in X.
\end{equation}

\textbf{\underline{Step 3: Comparing productivities.}}
By \cref{eq:contagion-effective-data-comparable}, every task's effective data is
eventually bounded above and below by fixed multiples of \(M_t\). Since
\(M_t\to\infty\), every task's capital productivity diverges, and cross-task
capital productivities remain uniformly comparable.

Let
\[
    \underline\psi_t:=\operatorname*{ess\,inf}_{i\in X}\psi^K_{i,t}.
\]
Then
\[
    \underline\psi_t
    \ge
    \underline f
    \left(\frac{b_n}{2}M_t\right)^\eta
    \to\infty.
\]
Moreover,
\[
    \operatorname*{ess\,sup}_{i\in X}\psi^K_{i,t}
    \le
    \overline f M_t^\eta
    \le
    C_\psi\underline\psi_t
\]
for some finite constant \(C_\psi\).

\textbf{\underline{Step 4: Ruling out a positive labor share.}}
Suppose, toward a contradiction, that full automation fails. Then there exist
\(\delta>0\) and dates \(t_k\to\infty\) such that
\[
    \mu(\mathcal L_{t_k})\ge\delta.
\]
At such dates, CES demand gives \(y_{i,t_k}=C_{t_k}c_{i,t_k}^{-\sigma}\), where
\(c_{i,t_k}\) is the unit cost of task \(i\). On \(\mathcal K_{t_k}\),
\(c_{i,t_k}=r_{t_k}/\psi^K_{i,t_k}\), so capital market clearing gives
\[
    K
    =
    C_{t_k}r_{t_k}^{-\sigma}
    \int_{\mathcal K_{t_k}}(\psi^K_{i,t_k})^{\sigma-1}\,di.
\]
On \(\mathcal L_{t_k}\), \(c_{i,t_k}=w_{t_k}/\psi^L\), so labor market clearing
gives
\[
    L
    =
    C_{t_k}w_{t_k}^{-\sigma}
    \mu(\mathcal L_{t_k})(\psi^L)^{\sigma-1}.
\]
Dividing the two market-clearing equations,
\[
    \frac{r_{t_k}}{w_{t_k}}
    =
    \left[
    \frac{L}{K}
    \frac{
    \int_{\mathcal K_{t_k}}(\psi^K_{i,t_k})^{\sigma-1}\,di
    }{
    \mu(\mathcal L_{t_k})(\psi^L)^{\sigma-1}
    }
    \right]^{1/\sigma}.
\]
By Step 3, for some finite \(C_\sigma\),
\[
    \int_{\mathcal K_{t_k}}(\psi^K_{i,t_k})^{\sigma-1}\,di
    \le
    C_\sigma\underline\psi_{t_k}^{\sigma-1}.
\]
If \(\sigma\ge1\), this follows from
\(\psi^K_{i,t_k}\le C_\psi\underline\psi_{t_k}\). If \(\sigma<1\), it follows
from \(\psi^K_{i,t_k}\ge\underline\psi_{t_k}\). Since
\(\mu(\mathcal L_{t_k})\ge\delta\),
\[
    \frac{r_{t_k}}{w_{t_k}}
    \le
    C_r\underline\psi_{t_k}^{(\sigma-1)/\sigma}
\]
for some finite \(C_r\). Therefore, for a.e. \(i\),
\[
    \frac{r_{t_k}/\psi^K_{i,t_k}}{w_{t_k}/\psi^L}
    =
    \frac{r_{t_k}}{w_{t_k}}
    \frac{\psi^L}{\psi^K_{i,t_k}}
    \le
    C_r\psi^L\underline\psi_{t_k}^{-1/\sigma}
    \to0.
\]
For all sufficiently large \(k\), capital is strictly cheaper than labor for
almost every task. Hence \(\mu(\mathcal L_{t_k})=0\), contradicting
\(\mu(\mathcal L_{t_k})\ge\delta\). Thus \(\mu(\mathcal L_t)\to0\).
\end{proof}

%%%%%%%%%%%%%%%%%%%%%%%%%%%%%%%%%%%%%%5

\subsubsection{Structure and Long-Run Behavior of Automation Network} \label{network_structure_long_run}

 We next examine how the structure of spillovers across tasks shapes the long-run composition of task production and data. In the case where $\sigma = 1/\eta$ we have a complete characterization:

%Finally, we investigate the speed of automation. We find automation to be slow in the absence of spillovers: it roughly linearly in log time across a wide range of parameter choices. We also find that the degree of spillovers can dramatically affect the speed of automation in the case where tasks are sufficiently substitutable.

%\begin{comment}
%    The long-run data stock of task $i$ (and hence also the long-run quantity of task $i$ production) is proportional to a nonlinear centrality measure for the  graphon $W$ given the data stock $\pmb{D}^*$ that solves the following functional equation: 
%\[
%\underbrace{\pmb{D}_{it}}_{\propto \pmb{y}_{it}} \propto f_i^{\sigma} \cdot \Big(\int W(i,k) \pmb{D}_{kt} dk \Big)^{\eta \sigma} 
%\]
%
%The long-run accumulation of data for task $i$ hinges on the (i) degree to which task $i$ benefits from data as captured by $f_i$; and the (ii) amenability of task $i$ in utilizing the data from other tasks as captured by $\int W(i,j) \pmb{D}_{kt} dk$, noting that latter is an endogenous object since it features the data stock of other tasks. 
%\end{comment}

\begin{proposition}\label{prop:limit_eigenfunction}
If $\sigma = 1/\eta$ then for any initial stock of data $\pmb{D}_0$ and any vector $\pmb{f}$, the long-run data stock and output is as follows: 
\[
\lim_{t \to \infty} \pmb{D}_t \simeq A_t \cdot \mathcal{F}_1 \quad \text{for some $(A_t)_t$ s.t. $\lim_{t \to \infty} A_t = \infty$.}
\]
where $\mathcal{F}_1$ is the right principal eigenfunction of the linear operator $\mathcal{L}$:
    \[
    \mathcal{L} \mathcal{F}_1 = \lambda_1 \mathcal{F}_1 
    \quad \text{where} 
    \quad 
    \mathcal{L}_i \pmb{D}  = f_i^{\sigma} \cdot \int W(i,k) \pmb{D}_k dk.
    \]
\end{proposition}

\cref{prop:limit_eigenfunction} fully characterizes the \emph{composition} of data and long-run output as a function of (i) degree to which task $i$ benefits from data as captured by $f_i$; and the (ii) eigenfunction centrality of task $i$ as captured by the leading right eigenfunction $\mathcal{F}_1$. %The proof  relies on the Krein-Rutman theorem (a generalization of Perron-Fronbenius to Banach spaces) to guarantee the existence of a leading (right) eigenfunction $\mathcal{F}_1$ and a strictly positive spectral gap. 

The characterization in \cref{prop:limit_eigenfunction} allows us to show that in the long run, the composition of data must converge in proportion to a task's leading eigenfunction centrality.\footnote{We believe the result continues to hold in the case where $\sigma \neq 1/\eta$. Here, the operator $\mathcal{L}$ is no longer linear, and convergence requires pairing more recent results on positivity of eigenfunctions of nonlinear (convex when $\sigma > 1/\eta$; concave when $\sigma < 1/\eta$) operators with those of monotone dynamical systems.} Indeed, in the special case that $\pmb{f}$ is constant, the long-run composition of data and production is exactly proportional to a continuous analog of the eigenvector centrality measure that is used widely across economics and network science (see \cite{golub2025eigenvalues} for a survey). When $\sigma > 1/\eta$, the amenability of task in utilizing the data from others tasks is \emph{amplified} in equilibrium. Thus, long-run differences in the data stock should be more dispersed than is predicted by eigenfunction centrality of $W$. By contrast, when $\sigma < 1/\eta$, this is \emph{dampened} and long-run differences in the data stock should be less dispersed than is predicted by the eigenfunction centrality. \label{appendix:contagion_limit}

\begin{proof}[{\bf Proof of \cref{prop:limit_eigenfunction}}]

Notice 
{\allowdisplaybreaks \begin{align*}
        \lim_{t \to \infty}\cfrac{D_{it}}{D_{jt}}  
        &= \lim_{t \to \infty}\cfrac{y_{it}}{y_{jt}}  \tag{L'Hopital's} \\
        &= \lim_{t \to \infty}\cfrac{\psi^K_{it}K_{it}}{\psi^K_{jt}K_{jt}} \tag{tasks produced by capital asymptotically} \\
        &= \lim_{t \to \infty} \left(\cfrac{\psi^K_{it}}{\psi^K_{jt}}\right)^{\sigma}
        \tag{$K_{it}=\frac{\left(\psi_{it}^{K}\right)^{\sigma-1}}{\int_{0}^{\gamma_{t}}\left(\psi_{jt}^{K}\right)^{\sigma-1}dj}K$}
        \\
        &= \lim_{t \to \infty} \left(\cfrac{f_i}{f_j}\right)^{\sigma} \left(\cfrac{\mathcal{A}_i(\pmb{D}_t)}{\mathcal{A}_j(\pmb{D}_t)}\right)^{\eta \sigma} \\
        &= \lim_{t \to \infty} \left(\cfrac{f_i}{f_j}\right)^{\sigma} \left(\cfrac{\int W(i,k) D_{kt} dk}{\int W(j,k) D_{kt} dk}
        \right)^{\eta \sigma} 
\end{align*}}
Hence we are looking for a data vector $\pmb{D}$ that solves: 
\[
\pmb{D}_i = C_t \cdot f_i^{\sigma} \Big( \cdot \int_X W(i,k) \pmb{D}_k dk\Big)^{\eta \sigma}
\]
where $C$ is a constant independent of the task index $i \in X$. We can define the operator 
\[
\mathcal{L} : \mathcal{D} \to \mathcal{D} \quad \text{as} \quad 
\mathcal{L} \pmb{D}_i = f_i^{\sigma} \Big(\int^1_0  W(i,k) \pmb{D}_k dk\Big)^{\eta \sigma}
\]
so we have the Fredholm integral equation of the second kind: 
\[
\pmb{D} = C \cdot \mathcal{L} \pmb{D}
\]
and in the special case in which $\sigma = 1/\eta$, we have: 
\[
\pmb{D}_i = C \cdot f_i^{\sigma} \cdot \Big(  \int_X W(i,k) \pmb{D}_k dk\Big)
\]
and $\mathcal{L}$ is linear, compact, but not necessarily self-adjoint since spillovers are not necessarily symmetric. The spectral theorem for compact operators guarantees the existence of solutions. Substituting, we have the equation 
\[
\pmb{D} = C \cdot \mathcal{L}\pmb{D} \implies \mathcal{L} \pmb{D} = \underbrace{\lambda}_{=: 1/C} \pmb{D}.
\]
What is $C$? Equivalently, what is the corresponding eigenvalue of the graphon $W$? We will show that it converges to the \emph{principal eigenfunction} of the operator $\mathcal{L}$. We start with the following lemma that invokes a basic fact from spectral theory: 

\begin{lemma}\label{lem:krein_rutman}
    The operator $\mathcal{L}$ admits a unique principal eigenvalue $\lambda_1$ and associated eigenfunction $\mathcal{F}_1$ with strictly postiive spectral gap i.e., $\mathsf{Re}(\lambda_1 - \lambda_k) > 0$ for any $k > 1$. 
\end{lemma}
This follows by observing that $\mathcal{L}$ is a positive ($f_i$, $W$ are both positive), compact, and irreducible operator (since $W$ is connected by assumption) and applying the Krein-Rutman theorem which is an analog of the Perron-Frobenius theorem for infinite-dimensional Banach spaces.

Now consider our law of motion for the data \emph{ratio} we derived above:  
\begin{equation}
    \cfrac{\partial \pmb{D}_{it} / \partial t}{\partial \pmb{D}_{jt} / \partial t} = \cfrac{\mathcal{L}_i \pmb{D}_t}{\mathcal{L}_j \pmb{D}_t}.
\end{equation}

This can be rearranged to separate terms by the task index:
\[
\frac{1}{\mathcal{L}_i\pmb{D}_t} \dot{\pmb{D}_{it}} = \frac{1}{\mathcal{L}_j\pmb{D}_t} \dot{\pmb{D}_{jt}}
\]
Since this equality must hold for all pairs $i, j \in X$, the expression must be equal to a quantity that is independent of the task index $i$. Call this scalar function $\alpha(t)$ (that might depend on time): 
\[
\frac{1}{\mathcal{L}_i\pmb{D}_t} \dot{\pmb{D}}_{it} = \alpha(t)
\]
This gives the following infinite-dimensional ODE for our system: 
\begin{equation}
    \dot{\pmb{D}} = \alpha(t) \cdot \mathcal{L}\pmb{D}
    \label{eq:underlying_dynamic}
\end{equation}

The scalar function $\alpha(t)$ in \eqref{eq:underlying_dynamic} affects the \emph{speed} of evolution along the solution trajectory but does not alter the \emph{shape} of the trajectory i.e., the composition of the data. We thus perform a time change by writing $\tau(t) = \int^t \alpha(s) ds$, which simplifies the equation to:
\begin{equation}
    \frac{d \pmb{D}_{\tau}}{d\tau} = \mathcal{L}\pmb{D}_{\tau}
    \label{eq:rescaled_dynamic}
\end{equation}
We solve this equation by expanding the solution in the basis of the eigenfunctions $\{\mathcal{F}_k\}_{k=1}^\infty$ of $\mathcal{L}$. Fix some large time $T$ for which our above equation holds (taking nested limits if you prefer). Then given the data stock $\pmb{D}_{T}$, since $\mathcal{L}$ is a linear operator, we we have the Riesz projection onto the dominant subspace and the `residual': %\footnote{Recall that $\mathcal{L}$ is not self-adjoint because our graphon $W$ need not be symmetric. This will not in general be orthogonal. See \cite{fabian2011banach} for a textbook treatment.} 
\[\pmb{D}_T =  \underbrace{c_1 \mathcal{F}_1}_{\substack{\text{Riesz} \\ \text{projection} \\ \text{$P\pmb{D}_T$}}} + w
\]
where $w$ is the residual $(I - P)D_T$ and $P$ is the Riesz projection associated with the principal eigenvalue. Next, since $\mathcal{L}$ is time-invariant, the solution to our ODE \eqref{eq:rescaled_dynamic} is given by the operator exponential: 
\begin{align*}
\pmb{D}_{\tau - \tau(T)} &= e^{\tau\mathcal{L}} \pmb{D}_T       \\
&= c_1 e^{\tau \mathcal{L}} \mathcal{F}_1 + e^{\tau \mathcal{L}} w \\
&= c_1 e^{\lambda_1 \tau} \mathcal{F}_1 + e^{\tau \mathcal{L}} w \tag{$\mathcal{F}_1$ is an eigenfunction}
\end{align*}

Next notice that we can bind the second term: define $\gamma:= \sup \{ \mathsf{Re}(\lambda): \lambda \in \sigma(\mathcal{L}), \lambda \neq \lambda_1\}$ so for any $\epsilon < \lambda_1 - \gamma$ we have: 
\[
\|e^{\tau \mathcal{L}} w  \| \leq M e^{(\epsilon + \gamma)\tau} \|w\|.
\]
Now applying \cref{lem:krein_rutman}, we have 
\[
\pmb{D}_{\tau - \tau(T)} = e^{\lambda_1 \tau} \Big[c_1 \mathcal{F}_1 + \underbrace{\cfrac{e^{\tau \mathcal{L}}w}{e^{\lambda_1 \tau}}}_{\to 0}]
\]
so the shape of the data distribution $\pmb{D}_t$ converges to that of the principal eigenfunction $\mathcal{F}_1$ of the operator $\mathcal{L}$:
\[
\cfrac{D_{\tau - \tau(T)}}{\|D_{\tau - \tau(T)}\|} \to \cfrac{\mathcal{F}_1}{\|\mathcal{F}_1\|}
\]
and we indeed have 
\[
\pmb{D}_t \simeq A_t \cdot \mathcal{F}_1
\]
where $A_t$ is a time-dependent scalar tracking the aggregate data level.
\end{proof}

%%%%%%%%%%%%%%%%%%%%%%%%%%%%%%%%%%%%%%%

\subsection{\cref{sec:efficiency}: Planner's Problem}
\begin{proof}[{\bf Proof of \cref{prop:generic_inefficiency}}] \label{sec:pf_generic_inefficiency}

The proof proceeds as follows. Throughout the proof, we will compare the planner's optimized allocation $k_{it}$ with an auxiliary allocation that is myopically optimal given the same data stock as the planner's, $\bar{k}_{it}$. We will then compare $\bar{k}_{it}$ to the fully myopic allocation starting at $t=0$, which we denote by $k_{it}^{m}$. We will show that these paths exhibit a \emph{weak ordering that holds for all $t\in\left[0,T\right]$} for $\sigma\geq1$, but not generically when $\sigma<1$. The reason for comparing to the static myopic allocation is that the fully myopic allocation will in general feature a different data stock as the planner's for any $t>0$.

We will first show the planner's equilibrium path coincides with the market's if and only if $\sigma=1/\eta$. Then, we will sign the direction of bias to show that the planner leans into the market's allocation for capital. In other words, if the market allocates more capital towards unproductive tasks than productive ones, as in the complements case, the planner would do so even more aggressively.

\paragraph{Generic inefficiency.}

The planner's current value Hamiltonian is 
\[
\mathcal{H}\left(\boldsymbol{k}_{t},\boldsymbol{D}_{t},\boldsymbol{\lambda}_{t},\mu_{t}\right)=Y_{t}+\sum_{i=1}^{2}\lambda_{it}y_{it}+\mu_{t}\left(K-\sum_{i\in\{h,l\}}\frac{1}{2}k_{it}\right)
\]
We will take $i$ to mean the $h$ task and $j$ to mean the $l$ task. The FOCs are 
\begin{align*}
\frac{1}{2}Y_{t}^{\frac{1}{\sigma}}y_{it}^{-\frac{1}{\sigma}}\psi^{K}\left(D_{it}\right)+\lambda_{it}\psi^{K}\left(D_{it}\right)-\frac{1}{2}\mu_{t} & =0\\
\frac{1}{2}Y_{t}^{\frac{1}{\sigma}}y_{it}^{-\frac{1}{\sigma}}\left(\psi^{K}\left(D_{it}\right)\right)^{\prime}k_{it}+\lambda_{it}\left(\psi^{K}\left(D_{it}\right)\right)^{\prime}k_{it} & =\rho\lambda_{it}-\dot{\lambda}_{it}
\end{align*}
The planner's relative task production is thus 
\[
\frac{y_{it}}{y_{jt}}=\left(\frac{\frac{1}{2}\frac{\mu_{t}}{\psi^{K}\left(D_{jt}\right)}-\lambda_{jt}}{\frac{1}{2}\frac{\mu_{t}}{\psi^{K}\left(D_{it}\right)}-\lambda_{it}}\right)^{\sigma}
\]
We know in the myopic case, given the same data stock, the ratio would be 
\[
\frac{\bar{y}_{it}}{\bar{y}_{jt}}=\left(\frac{\psi^{K}\left(D_{it}\right)}{\psi^{K}\left(D_{jt}\right)}\right)^{\sigma}
\]
So the planner allocates more capital to task $i$ compared to the myopic allocation with the same data stock if and only if 
\[
\lambda_{jt}\psi^{K}\left(D_{jt}\right)\leq\lambda_{it}\psi^{K}\left(D_{it}\right)
\]
We know the relationship holds with equality when $t=T$ since $\lambda_{iT}=\lambda_{jT}=0$. We will evolve the system backwards in time. Differentiating the RHS, we have 
\begin{align*}
\left(\psi_{it}^{K}\right)^{\prime}y_{it}\lambda_{it}+\psi_{it}^{K}\dot{\lambda}_{it} & =\left(\psi_{it}^{K}\right)^{\prime}y_{it}\lambda_{it}+\psi_{it}^{K}\left(\rho\lambda_{it}-\frac{1}{2}Y_{t}^{\frac{1}{\sigma}}y_{it}^{-\frac{1}{\sigma}}\left(\psi_{it}^{K}\right)^{\prime}k_{it}-\lambda_{it}\left(\psi_{it}^{K}\right)^{\prime}k_{it}\right)\\
 & =\left(\psi_{it}^{K}\right)^{\prime}y_{it}\lambda_{it}+\psi_{it}^{K}\left(\rho\lambda_{it}-\left(\frac{1}{2}Y_{t}^{\frac{1}{\sigma}}y_{it}^{-\frac{1}{\sigma}}+\lambda_{it}\right)\left(\psi_{it}^{K}\right)^{\prime}k_{it}\right)\\
 & =\left(\psi_{it}^{K}\right)^{\prime}y_{it}\lambda_{it}+\psi_{it}^{K}\rho\lambda_{it}-\frac{1}{2}\mu_{t}\left(\psi_{it}^{K}\right)^{\prime}k_{it}
\end{align*}
The LHS is symmetric. The two sides are equal at $t=T$ if and only if 
\begin{align*}
\left(\psi_{jt}^{K}\right)^{\prime}y_{jt}\lambda_{jt}+\psi_{jt}^{K}\rho\lambda_{jt}-\frac{1}{2}\mu_{t}\left(\psi_{jt}^{K}\right)^{\prime}k_{jt} & =\left(\psi_{it}^{K}\right)^{\prime}y_{it}\lambda_{it}+\psi_{it}^{K}\rho\lambda_{it}-\frac{1}{2}\mu_{t}\left(\psi_{it}^{K}\right)^{\prime}k_{it}\\
\left(\psi_{jt}^{K}\right)^{\prime}k_{jt}\left(\psi_{jt}^{K}\lambda_{jt}-\frac{1}{2}\mu_{t}\right) & =\left(\psi_{it}^{K}\right)^{\prime}k_{it}\left(\psi_{it}^{K}\lambda_{it}-\frac{1}{2}\mu_{t}\right)\tag{\ensuremath{\iff}}\\
\frac{\left(\psi_{jt}^{K}\right)^{\prime}}{\left(\psi_{it}^{K}\right)^{\prime}} & =\frac{k_{it}}{k_{jt}}\tag{\ensuremath{\iff}}\\
\left(\frac{\psi_{it}^{K}}{\psi_{jt}^{K}}\right)^{1/\eta} & =\left(\frac{\psi_{it}^{K}}{\psi_{jt}^{K}}\right)^{\sigma}\tag{\ensuremath{\iff}}
\end{align*}
When $\sigma=1/\eta$, the guess holds, and the planner indeed allocates capital according to her myopic optimal allocation given her current data stock. Since this holds at all times, we conclude $k_{it}=\bar{k}_{it}=k_{it}^{m}$ for $i=l,h$ and $t$. When $\sigma\neq1/\eta$, however, this can only hold if $D_{it}=D_{jt}$, but if we evolve the system back to $t=0$, this contradicts the initial condition that one task begins with strictly higher data. As such, we conclude the paths will be \emph{inefficient} when $\sigma\neq1/\eta$. Moreover, from the derivation above, we note that the RHS is greater than the LHS for some $t=T-\epsilon$ with $\epsilon>0$ sufficiently small if and only if 
\[
\left(\frac{\psi_{it}^{K}}{\psi_{jt}^{K}}\right)^{1/\eta}\leq\left(\frac{\psi_{it}^{K}}{\psi_{jt}^{K}}\right)^{\sigma}
\]
where we use the fact that $\psi_{jt}^{K}\lambda_{jt}-\frac{1}{2}\mu_{t}<0$ and that we are running the system \emph{backward} in time. That is, provided $\psi_{it}^{K}>\psi_{jt}^{K}$, the planner produces more of $i$ relative to the myopic optimal allocation given her data stock if and only if $\sigma>1/\eta$. In the next part of this argument, we extend this local result near $T$ to a global result regarding the entire path, and we will compare with the fully myopic path instead of the myopic path given the planner's data stock.

\begin{figure}
    \centering
    \includegraphics[width=0.5\linewidth]{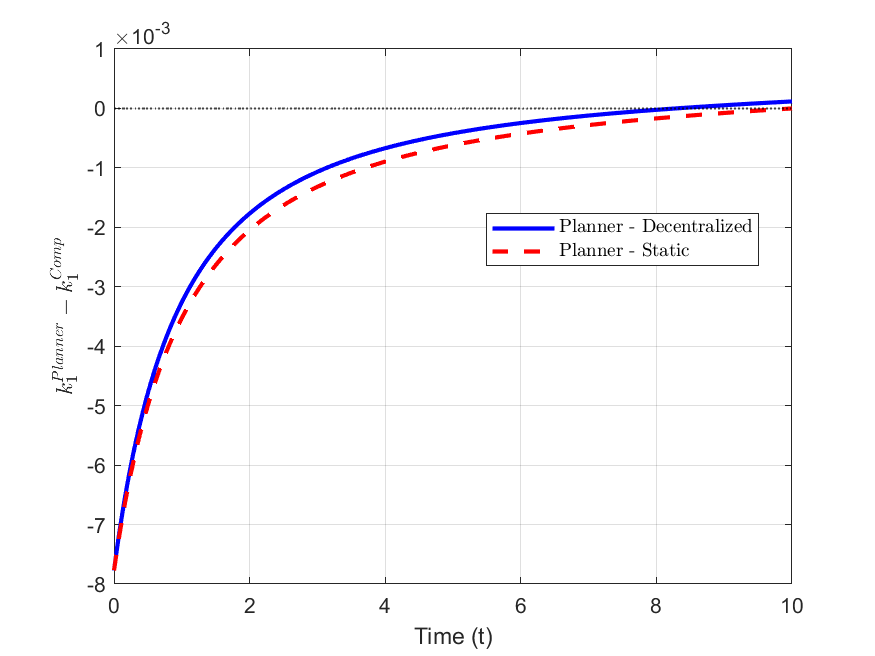}
    \caption{Comparison of capital allocated to the data-rich sector.}
    \label{fig:example_reversal}
\end{figure}

\paragraph*{Global comparison.}
Towards making a global comparison, we will first show \cref{lemma:no crossing}, which states that either $k_{ht}\geq\bar{k}_{ht}\text{ }\forall t$ or $k_{ht}\leq\bar{k}_{ht}\text{ }\forall t\in\left[0,T\right]$. Put differently, the planner's allocation and her myopically optimal allocation cannot cross. Using this result, we can now conclude the proof.
\begin{enumerate}
\item $\sigma<1$. In this case, we must have $k_{ht}\leq\bar{k}_{ht}\text{ }\forall t$ and strictly so for at least some interval by combining our earlier argument with \cref{lemma:no crossing}.\footnote{We cannot have the relation always hold with equality. Otherwise, we would contradict that one task has strictly higher data to begin with as argued previously.} Additionally, note that $\bar{k}_{h0}=k_{h0}^{m}$ since data stocks are $t=0$ are equivalent, so $k_{h0}^{m}\leq k_{h0}$. However, we do not necessarily have $k_{ht}\leq k_{ht}^{m}$ for all times as the productivity of the initially data poor task is growing faster under the planner's allocation. Given the planner's initial tilt towards the low data task, we may reach a stage where the planner balances her allocation of capital across tasks sooner compared to the decentralized equilibrium, so we may very well have $k_{ht}>k_{ht}^{m}$ at some point. \cref{fig:example_reversal} illustrates this reversal numerically: whereas the planner allocates less capital to the data-poor sector compared to the decentralized equilibrium to begin with, the allocation reverses around $t \approx 8$.
\item $1\leq\sigma<1/\eta$. As before, we have $k_{ht}\leq\bar{k}_{ht}\text{ }\forall t$. We now compare $\bar{k}_{ht}$ with $k_{ht}^{m}$. The two are initially equal. Then the planner allocates weakly less capital compared to the decentralized equilibrium, thereby accumulating less data, and so on and so forth. It is thus straightforward to verify that a lower ratio of data (of $h$ to $l$) implies less capital allocated to $h$ when $\sigma\geq1$ at all times. Therefore we have the following ordering: $k_{ht}\leq\bar{k}_{ht}\leq k_{ht}^{m}$.
\item $\sigma=1/\eta$. This is covered by the efficiency result from before. 
\item $\sigma>1/\eta$. In this case, we must have $k_{ht}\geq\bar{k}_{ht}\geq k_{ht}^{m}\text{ }\forall t$ and the first inequality holding strictly at least for some interval. The argument is analogous to the second case. 
\end{enumerate}
\end{proof}

\begin{lemma}[No crossing] \label{lemma:no crossing} Either $k_{ht}\geq\bar{k}_{ht}\text{ }\forall t$ or $k_{ht}\leq\bar{k}_{ht}\text{ }\forall t\in\left[0,T\right]$. 
\end{lemma}

\begin{proof}[{\bf Proof of \cref{lemma:no crossing}}]
Towards a contradiction, suppose there exists some $\tau\in\left(0,T\right)$ where $k_{ht}$ crosses from below to above of $\bar{k}_{ht}$. The other direction proceeds analogously. Intuitively, the planer first over invests in the low data task and then undoes her investment, which cannot be optimal. To establish this rigorously, we will construct a perturbation by shifting capital between $l$ and $h$ over a small interval around $\tau$. In doing so, we do not affect output prior to the perturbation and strictly improve the objective over this short interval. Our construction will also have the property that the perturbed data vector after the interval coincides with the data vector from the planner's original equilibrium, so the outputs are equal after the interval as well. Clearly, such a perturbation strictly improves the objective, contradicting optimality.

It remains to show that such a construction exists. Let $D_{i}^{p}\left(s\right)$ denote the planner's data of $i$ at time $s$; let $\tilde{D}_{i}\left(s\right)$ be the data of $i$ at time $s$ following the perturbed capital path. We make the time dependence more explicit for clarity. Fix some $\epsilon>0$. Formally, our feasibility requirements are 
\begin{align*}
D_{i}^{p}\left(\tau-\epsilon\right) & =\tilde{D}_{i}\left(\tau-\epsilon\right)\\
D_{i}^{p}\left(\tau+\epsilon\right) & =\tilde{D}_{i}\left(\tau+\epsilon\right)\\
\dot{\tilde{D}}_{1}\left(s\right) & =f\left(\tilde{D}_{1}\left(s\right)\right)^{\eta}\left(k_{1}\left(s\right)+\Delta_{1}^{k}\left(s\right)\right)\quad s\in\left[\tau-\epsilon,\tau+\epsilon\right]\\
\dot{\tilde{D}}_{2}\left(s\right) & =f\left(\tilde{D}_{2}\left(s\right)\right)^{\eta}\left(k_{2}\left(s\right)-\Delta_{1}^{k}\left(s\right)\right)\quad s\in\left[\tau-\epsilon,\tau+\epsilon\right]
\end{align*}
Here $f\equiv f_{1}=f_{2}$. From ODE theory, we recognize that $\tilde{D}_{i}\left(s\right)$ solve a Bernoulli differential equation which admits closed form solutions. Our requirement of reaching the same data vector at $\tau+\epsilon$ boils down to 
\begin{align*}
D_{1}^{p}\left(\tau+\epsilon\right) & =\left(\left(\tilde{D}_{1}\left(\tau-\epsilon\right)\right)^{1-\eta}+\left(1-\eta\right)f\int_{\tau-\epsilon}^{\tau+\epsilon}\left(k_{1}\left(s\right)+\Delta_{1}^{k}\left(s\right)\right)ds\right)^{\frac{1}{1-\eta}}\\
D_{2}^{p}\left(\tau+\epsilon\right) & =\left(\left(\tilde{D}_{2}\left(\tau-\epsilon\right)\right)^{1-\eta}+\left(1-\eta\right)f\int_{\tau-\epsilon}^{\tau+\epsilon}\left(k_{2}\left(s\right)-\Delta_{1}^{k}\left(s\right)\right)ds\right)^{\frac{1}{1-\eta}}
\end{align*}
Observe that for any $\Delta_{1}^{k}$ that is anti-symmetric around $\tau$, the above relations must hold. As such, simply choose a $\Delta_{1}^{K}:\left[\tau-\epsilon,\tau+\epsilon\right]\to\mathbb{R}_{>0}$ satisfying 
\[
\begin{cases}
k_{1}\left(\tau-s\right)<k_{1}\left(\tau-s\right)+\Delta_{1}^{k}\left(\tau-s\right)<\bar{k}_{1}\left(\tau-s\right) & \text{ }\forall s\in\left[0,\epsilon\right]\\
\bar{k}_{1}\left(\tau+s\right)<k_{1}\left(\tau+s\right)+\Delta_{1}^{k}\left(\tau+s\right)<k_{1}\left(\tau+s\right) & \text{ }\forall s\in\left[0,\epsilon\right]
\end{cases}
\]
For example, we can choose a linear segment anchored at $k_{1}\left(\tau\right)=\bar{k}_{1}\left(\tau\right)$ that is above the planner's allocation but sufficiently close to it between $\left[\tau-\epsilon,\tau\right]$. Call the difference $\Delta_{1}^{K}\left(\tau-s\right)$. We would simply take the negative of the difference for $\left[\tau,\tau+\epsilon\right]$, which would lie in the interval $\left(k_{1}^{M}\left(\tau+s\right),k_{1}^{M}\left(\tau+s\right)\right)$ provided the initial segment is close enough to the planner's allocation. 
\end{proof}

%%%%%%%%%%%%%%%%%%%%%%%%%

\subsection{\cref{section:capitalaccumualation}: Capital Accumulation} \label{appendix:capital_accum} 
In this section, we first derive the singularity result under a Solow capital accumulation rule; we then bound the savings appropriately in a neoclassical growth model and apply our Solow result to conclude.

With Solow capital accumulation, again, we have
\begin{align*}
1-\gamma_{t}\equiv\chi_{t} & =\frac{L}{K}D_{t}^{-\eta}\\
\dot{D}_{t} & =\frac{L}{\chi_{t}}
\end{align*}
The evolution of capital is given by
\begin{align*}
\dot{K}_{t} & =-\delta K_{t}+s_{t}Y_{t}\\
 & =-\delta K_{t}+s_{t}\left(\gamma_{t}^{\frac{1}{\sigma}}\left(D_{t}^{\eta}K_{t}\right)^{\frac{\sigma-1}{\sigma}}+\left(1-\gamma_{t}\right)^{\frac{1}{\sigma}}L^{\frac{\sigma-1}{\sigma}}\right)^{\frac{\sigma}{\sigma-1}}
\end{align*} 

%%%%%%%%%%%%%%%%%%%%%%%%%%%%%%%

\begin{lemma} \label{lemma:solow_singularity}
When sectors are symmetric, and the savings rate satisfies
\[
\inf_{t\geq T}s_{t}\geq s>0
\]
for some $T$ finite, then there is a finite time explosion.
\end{lemma}

\begin{proof}[{\bf Proof of \cref{lemma:solow_singularity}}] The proof proceeds as follows. First, we will run the system long enough so that $Y_{t}\geq D^{\eta}K_{t}\geq\delta K_{t}$. This ensures capital will be increasing at all times moving forward. Next, we will lowerbound capital accumulation, output, and automation by studying an economy where we ignore labor altogether. We finally establish a finite time explosion in this economy and thereby concluding the argument.

\paragraph{1. Finding a time such that capital always increases thereafter}

We will first find a $t_{0}$ such that capital is guaranteed to increase thereafter. We do this in two steps. First, notice that 
\begin{align*}
Y_{t} & =D_{t}^{\eta}\frac{K_{t}}{\gamma_{t}}\\
 & \geq D_{t}^{\eta}K_{t}
\end{align*}
The first equality follows from $\gamma_{t}$ solving the following relationship
\[
Y_{t}=Y_{i,t}=D_{t}^{\eta}\frac{K_{t}}{\gamma_{t}}=\frac{\psi^{L}L}{1-\gamma_{t}}\quad\forall i\in X
\]
Next, notice that 
\[
\dot{D}_{t}=Y_{t}=\frac{\psi^{L}L}{1-\gamma_{t}}=\psi^{L}L>0
\]
This implies $D_{t}$ tends to infinity, so we can find a $t_{0}$ such that $D_{t_{0}}^{\eta}K_{t_{0}}>\delta K_{t_{0}}\implies\dot{K}_{t_{0}}>0$. Then, by virtue of $D_{t}$ being strictly increasing in time, $\dot{K}_{t}>0$ for all $t\geq t_{0}$.

\paragraph{2. A lower-bounded economy}

We next consider $t\geq\max\left\{ T,t_{0}\right\} \equiv\uline{t}$ by studying a simplified counterfactual economy. The actual economy evolves according to
\[
\begin{cases}
\dot{K}_{t} & =-\delta K_{t}+s_{t}D_{t}^{\eta}\frac{K_{t}}{\gamma_{t}}\\
\dot{D}_{t} & =D_{t}^{\eta}\frac{K_{t}}{\gamma_{t}}
\end{cases}
\]
If we produce with capital only and a lower bounded savings rate $s$, the dynamics are
\[
\begin{cases}
\dot{\tilde{K}}_{t} & =-\delta\tilde{K}_{t}+s\tilde{D}_{t}^{\eta}\tilde{K}_{t}\\
\dot{\tilde{D}}_{t} & =\tilde{D}_{t}^{\eta}\tilde{K}_{t}
\end{cases}
\]
By a version of Kamke's Comparison Theorem \citep{smith1995monotone}, we can show $K_{t}\geq\tilde{K}_{t}$ and $D_{t}\geq\tilde{D}_{t}$ for all $t\geq t_{0}$ if we specialize the two economies to begin with the same initial conditions at $\uline{t}$ as the system is cooperative.

\paragraph{3. Dynamics of the lower-bounded economy}

Now we turn to the dynamics of the lower-bounded economy at $t\geq\uline{t}$. The capital evolution is given by
\[
\dot{\tilde{K}}_{t}=-\delta\tilde{K}_{t}+s\tilde{D}_{t}^{\eta}\tilde{K}_{t}
\]
Since $D_{t}$ is going to infinity in the limit, we can choose $t_{1}$ so that for some $0<\epsilon<s$
\[
\dot{\tilde{K}}_{t_{0}}=-\delta\tilde{K}_{t_{1}}+s\tilde{D}_{t_{1}}^{\eta}\tilde{K}_{t_{1}}\geq\left(s-\epsilon\right)\tilde{D}_{t_{1}}^{\eta}\tilde{K}_{t_{1}}
\]
Again, by virtue of $D_{t}$ being increasing in $t$, the above relation holds for all $t\geq t_{1}$. From the data LOM in this economy, we have
\[
\dot{\tilde{K}}_{t}\geq\left(s-\epsilon\right)\dot{\tilde{D}}_{t}
\]
Integrating from $t_{1}$ to $t$, the relationship is
\[
\tilde{K}_{t}\geq\left(s-\epsilon\right)\tilde{D}_{t}-\tilde{D}_{t_{1}}+\tilde{K}_{t_{1}}
\]
We can then find some $t_{2}$ large enough so that for some $\omega>0$
\[
\tilde{K}_{t_{2}}\geq\left(s-\epsilon\right)\tilde{D}_{t_{2}}-\tilde{D}_{t_{1}}+\tilde{K}_{t_{1}}\geq\omega\tilde{D}_{t_{2}}
\]
This relationship similarly holds for all $t\geq t_{2}$ as $\tilde{D}_{t}$ is monotonically increasing. Substituting this relation back into the dynamics for data accumulation, we have for all $t\geq t_{2}$
\[
\dot{\tilde{D}}_{t}\geq\omega\tilde{D}_{t}^{1+\eta}
\]
This super-linear growth implies a finite time explosion. In particular, we know that
\[
\tilde{D}_{t}\geq\left(\tilde{D}_{t_{2}}^{-\eta}-\omega\eta\left(t-t_{2}\right)\right)^{-\frac{1}{\eta}}
\]
which means full automation as well as the finite time explosion in data and output takes place weakly before $t_{2}+1/\left(\omega\eta\tilde{D}_{t_{2}}^{\eta}\right)$ in the original economy. \end{proof}

\begin{proof}[{\bf Proof of Proposition \ref{prop:singularity}}]

We now move to the NGM set-up. The representative household solves
\[
\max_{\{c_{t}\}}\int_{0}^{\infty}e^{-\rho t}\frac{c_{t}^{1-\theta}}{1-\theta}dt\quad\text{s.t.}\quad\dot{a}_{t}=r_{t}a_{t}+w_{t}-c_{t}
\]
yielding the Euler equation 
\[
\frac{\dot{c}_{t}}{c_{t}}=\frac{1}{\theta}\left(r_{t}-\delta-\rho\right)
\]
 and the aggregate resource constraint $\dot{K}_{t}=s_{t}Y_{t}-\delta K_{t}$, where $s_{t}=\left(Y_{t}-c_{t}\right)/Y_{t}$ is the equilibrium savings rate. To apply the previous result, we want to show that the savings rate $s_{t}\to s>0$, so the savings rate ought to be lower bounded for some $t=T$ sufficiently large. 

We proceed by guess-and-verify and let $t^{*}$ denote the time of singularity.\footnote{If helpful, this can be recast as a proof by contradiction, so $t^{*}=\infty$.} Assume $s_{t}\to s>0$, then for $t$ sufficiently large
\[
\dot{K}_{t}=s_{t}Y_{t}-\delta K_{t}\geq\omega_{t}Y_{t}
\]
Where $\omega_{t}$ is chosen as in the earlier proof, and $\lim_{t\to t^{*}}\omega_{t}=s$. We also have $\dot{D}_{t}=Y_{t}$. As such, $\lim_{t\to t^{*}}K_{t}/D_{t}=s_{t}$ as we approach the singularity, which implies the growth rates of these variables, $g_{K,t}$ and $g_{D,t}$, must also converge. From the resource constraint,
\[
g_{Y,t}=\eta g_{D,t}+g_{K,t}\approx\left(1+\eta\right)sD_{t}^{\eta}
\]
From the Euler equation,
\[
g_{c,t}\approx\frac{1}{\theta}D_{t}^{\eta}
\]
By virtue of the guess, consumption and output growth rate must equalize, so
\[
s=\frac{1}{\left(1+\eta\right)\theta}
\]
This completes the guess-and-verify step. Applying \cref{lemma:solow_singularity} establishes that output and consumption will become unbounded in finite time, and the entire economy will be automated in finite time as well. Finally, when $\theta \leq 1$, the candidate path is feasible and delivers infinite expected utility; when $\theta > 1$, the transversality condition is satisfied. This shows the candidate path is both feasible and optimal, concluding our proof.
\end{proof}

\section{Extensions} \label{appendix:discussion_extensions}

\paragraph{A multisector economy.} In our economy  tasks are aggregated directly into a final good. We might, instead, consider an economy comprised of  a set of intermediate goods ($M$) and a spectrum of tasks ($X$), in which a final good is produced with intermediate goods, and each intermediate good is produced with a bundle of tasks. Different intermediate goods value different tasks---the degree to which intermediate good $m \in M$ values task $i \in X$ is $G(m,i) \geq 0$. Intermediate goods overlap in the set of tasks that they value, so data of task $i$ is generated by  and used  in production of all the intermediate goods that use task $i$.\footnote{The idea that different sectors might be connected by sharing `overlapping tasks' that precipitates data spillovers is related in spirit to \cite*{chen2023capability} in which different `markets' share `overlapping capabilities' that precipitates merger waves.} 

\cref{appendix:multisector} outlines such an economy and sketches how relative elasticities for the final good vis-a-vis intermediate goods, as well as how the pattern of intermediate good-task dependance and how they overlap, shapes differential data accumulation. 

\paragraph{Non-substitutable spillovers.} Our model of data spillovers imposed the implicit assumption that although the importance of task $j$ on $i$'s productivity is weighted by $W(i,j)$, they still enter into $i$'s production function in a \emph{substitutable} fashion. While this is standard in economic analyses of networks, we might imagine richer data aggregation in which we require both data from task $i$ \emph{and} task $j$ to augment capital productivity for $i$: 
\[
\left(\int \Big(W(i,j) D_{jt}\Big)^{\frac{\phi-1}{\phi}} dj \right)^{\frac{\phi}{\phi-1}} \quad \text{where $\phi \in [0,+\infty)$.} 
\]
Our model of cross-task spillovers corresponds to $\phi \to +\infty$; in the other extreme $\phi \to 0$ data-driven automation is then held up by the data-poorest task.\footnote{For instance if $W(i,\cdot)$ is full-support on $X$ this slows down the transition path of automation substantially since the rate at which effective data grows is held-up by, e.g., the most data-poor sector. This has sharp implications for the planner's problem of how to allocate production across sectors.} Of course, this `CES-form' imposes a kind of symmetry in how task data aggregates; we think a richer model capturing higher-order interactions among task-specific data is an exciting line of future work. 

\paragraph{Expanding the set of tasks.} Although our main analysis holds the task space fixed, our model offers a novel perspective on why human labor might have a persistent advantage over capital at new tasks. Let $X_t = [0, x_t]$ denote the set of tasks used in the final good production where the \emph{task frontier} $x_t$ is continuous and increasing in $t$. For a new task $x_t$ that is introduced at time-$t$, it is, by definition, novel and was not previously produced so at the time $t$ at which it was introduced we have no task-specific data. How might we expect automation and data accumulation in this economy to unfold? 

The main trade-off stems from the degree of spillover among the new task and other exiting tasks. If the spillover is weak, labor will have a comparative advantage when the new task is introduced, although the advantage might not persist. In this economy, \emph{human labor facilitates automation} by producing frontier tasks---and attendant training data---precisely when data is scarce.\footnote{For instance, in \cite{korinek2024scenarios} it is assumed that humans can simply perform high complexity tasks.} Alternatively, if there is sufficient overlap with existing tasks (as given by our graphon $W$ that captures cross-task spillovers), task $x_t$  can be automated by training AI systems on nearby tasks. In this world, strong data spillovers allow automation to \emph{bootstrap itself} by continually using the endogenously generated data on tasks at the frontier to automate new tasks as soon as they emerge.

Which of these economies emerge will ultimately depend on the \emph{speed} at which new tasks emerge  relative to the \emph{strength} of spillovers from existing tasks to new tasks. When new tasks emerge quickly, tasks near the frontier, have, themselves, not been around for long---and so they are ``data poor". Then, in the absence of strong spillovers from the core tasks, labor will play an important dual role in both producing new tasks as well as generating data on them. Alternatively, When new tasks emerge slowly, tasks near the frontier  have been around for long enough that they are ``data rich", and  even local spillovers may well suffice for the new task to be rapidly automated. 

\subsection{A Multi-sector Economy} \label{appendix:multisector}

In the main text we focused on a single-sector economy that aggregated tasks into a final good. Here we extend our main model to one with multiple sectors that is aggregated via another CES into a final good. All sectors draw from the same continuum of tasks, but different sectors may use the same task at different intensities, or combine tasks differently via CES aggregators with different task substitutabilities.%\footnote{The idea that sectors might be connected by sharing `overlapping tasks', and that this might precipitate spillovers is related in spirit to \cite*{chen2023capability} in which different `markets' share `overlapping capabilities' that precipitates merger waves.} 

\paragraph{Tasks and the final good.} We now index intermediate products by $m \in M := [0,1]$. A final good is now produced competitively
\[Q = \Big(\int_M q_m^{\frac{\nu-1}{\nu}}\Big)^{\frac{\nu}{\nu - 1}}\]
where $q_m$ is the quantity of sector $m$ produced competitively. 

\paragraph{Intermediate products.} Intermediate products are produced by tasks: 
\[
q_m = \left(\int_{X} \left(G(m,i) \cdot y^m_i\right)^{\frac{\sigma_m - 1}{\sigma_m}} di\right)^{\frac{\sigma_m}{\sigma_m - 1}}.
\]
where, as before, $X = [0,1]$ is our task space and $y^m_i$ is the quantity of task $ i \in X$ used to produce sector $m \in M$. The function 
\[
G(m,\cdot): X \to \mathbb{R}_+
\]
governs the degree to which task $i \in X$ is used in the production of sector $m \in M$. 
\begin{remark}
The analysis in the main text is thus a special case where the economy was comprised of a single sector $m$, and $G(m,i) = 1$ for each $i \in X$. 
\end{remark}

\paragraph{Tasks.} Task $i$ can be performed with capital or labor: 
\[
y_i^m = \psi^L \cdot L^m_i + \psi^k({D}_i) \cdot K^m_i
\]
where $L^m_i, K^m_i$ are the quantities of labor and capital employed in sector $m$ to produce task $i$. We assume, as in the main text, that 
\[
\psi^K\left(\mathcal{A}_i\left(\pmb{D}\right)\right) = \left(\mathcal{A}_i\left(\pmb{D}\right)\right)^{\eta} \quad \text{for $\eta \in (0,1)$.}
\]
and to simplify our exposition we assume $\mathcal{A}_i(\pmb{D}) = \pmb{D}_i$ i.e., no cross-task spillovers. 

% \bxnote{This mapping doesn't make sense... do you mean $D^m_i$ instead of $\boldsymbol{D}_i$?}
% That is, \emph{task-specific} data drives the productivity of capital in producing task $i$.
% \bxnote{We should clarify that the D vector is two dimensional now (over $i$ and $m$); also wouldn't it make sense to have $\sigma_m$? That is, task substitutability can be different for different sectors. Right now the only heterogeneity is in $G$.}

\paragraph{Law of motion aggregated across sectors.} 
We have the law of motion 
\[\dot{{D}_i} = \int_{M} y^m_i dm \quad \text{for each task $i \in X$}
\]
i.e., the evolution of the data for task $i$ depends on the production of \emph{all} sectors. 

\paragraph{Task cost.} The resulting unit cost to produce using task $i$ given data stock $\pmb{D}$ is:
\begin{equation}
    c_i({D}_i) = \min \left( \frac{w}{\psi^L}, \frac{r}{\psi^K\left({D}_i\right)} \right) = \min \left(\frac{w}{\psi^L}, \frac{r}{{D}_i^\eta} \right)
\end{equation}

\paragraph{Intermediate product cost.} The price index for intermediate product $m$ is its unit cost function:
\begin{equation}
    p_m(\pmb{D}) = \left( \int_{j \in X} \left( \frac{c_j({D}_j)}{G(m,j)} \right)^{1-\sigma_m} dj \right)^{\frac{1}{1-\sigma_m}}.
\end{equation}
where $\sigma_m > 0$ is the CES parameter for sector $m \in M$. 

\paragraph{Final good production and demand for intermediate products.} The demand for sector $m$'s output is given by the usual FOC condition for the final good producer:  
\begin{equation}
    q_m(D) = Q \left( \frac{p_m(\pmb{D})}{P} \right)^{-\nu}
\end{equation}
where $P$ is the aggregate price index of the final good $Q$. We will set the final good as the numeraire in this economy so $P=1$.

\paragraph{Task demand.} Task demand for intermediate product $m$ can be written: 
\[
y_i^m(\pmb{D}) = q_m \cdot (G(m,i))^{{\sigma_m}-1} \cdot \left(\frac{c_i}{p_m}\right)^{-{\sigma_m}}
\]
and so aggregate task demand can be written:
\begin{align*}
y_i(\pmb{D}) &= \int_{m \in M} y_i^m(\pmb{D}) dm \\
&= \int_{m \in M} \Big[ q_m \cdot (G(m,i))^{{\sigma_m}-1} \cdot \left(\frac{c_i}{p_m}\right)^{-{\sigma_m}} \Big] dm\\
&= \int_{m \in M} \left[ Q \left(p_m(\pmb{D})\right)^{-\nu} (G(m,i))^{{\sigma_m}-1} \cdot \left(\frac{c_i}{p_m}\right)^{-{\sigma_m}}  \right] dm \\
&= Q \int_{m \in M} \Big[\left(p_m(\pmb{D})\right)^{\sigma_m-\nu} (G(m,i))^{{\sigma_m}-1} \cdot \left(c_i\right)^{-{\sigma_m}}  \Big] dm \\
&= Q \int_{m \in M} \underbrace{\left( \int_{j \in X} \left( \frac{c_j(\pmb{D}_j)}{G(m,j)} \right)^{1-\sigma_m} dj \right)^{\frac{\sigma_m-\nu}{1-\sigma_m}}}_{= p_m^{\sigma_m - \nu}} (G(m,i))^{{\sigma_m}-1} \cdot \left(c_i\right)^{-{\sigma_m}}  dm 
\end{align*}

We make three key observations about how the quantity of task $i$ is determined. This is the key variable that will, in turn, drive the corresponding task-specific data accumulation of $\pmb{D}_i$. 

First, and most straightforwardly, $y_i$ is decreasing in $c_i$, albeit weighted by sectoral elasticities $(\sigma_m)_m$ inside the integral. If $\sigma_m = \sigma$ then we have a constant elasticity of the task exactly given by $\sigma$. 

Second, how sectoral prices $(p_m)_m$ propagate to task quantities hinge on the sizes of $\sigma_m \gtrless \nu$. If $\sigma_m > \nu$ i.e., sectoral production via tasks is more substitutable than the final good production---this perhaps the economically relevant case---then an (exogenous) increase in $p_m$ that holds $c_i$ fixed e.g., perturbing the cost of other tasks leads to \emph{more} of task $i$ being produced. This is driven by two effects. First, the final goods producer demands less of sector $m$ since 
\[
q_m = Q p_m^{-\nu}.
\]
But as the sector price increases, there is also a task reallocation effect within sector $m$: producers reallocate production toward task $i$ as long as the sector's production function is sufficiently substitutable, and this is exactly the condition $\sigma_m > \nu$. 

Third, $G(m,i)$ measures the `importance' of task $i$ for sector $m$. If production in sector $m$ are gross substitutes i.e., $\sigma_m \geq 1$ then when demand for sector $m$ increases, more important tasks---those with high $G(m,i)$---are \emph{overproduced}, and the extent of this overproduction is given by $\sigma_m \geq 1$. Conversely, if $\sigma < 1$ then the opposite obtains: when demand for sector $m$ increases, the less important tasks are overproduced. 

%%%%%%%%%%%%%%%%%%%%%%%%%%%%%%%%
%%%%%%%%%%%%%%%%%%%%%%%%%%%%%%%%

\section{Simulation Method}

\label{appendix:simulation_details}
We solve a discretized version of our model with a continuum of tasks. \cref{discretized-model} introduces the discretized model and \cref{discretized-evolution} outlines the numerical procedure. Time subscripts are omitted for brevity. Replication codes are available upon request.\footnote{Please email bryantxia2435@gmail.com}

\subsection{Discretized Model}
\label{discretized-model}

We approximate the continuum of tasks by $N$ large. Define $\epsilon:=1/N$, output is given by\footnote{We set $N=1000$ in our codes.} 
\[
Y:=\left(\sum_{i=1}^{N}\left(y_{i}\right)^{\frac{\sigma-1}{\sigma}}\epsilon\right)^{\frac{\sigma}{\sigma-1}}
\]
Costs are now $\sum_{i=1}^{N}p_{i}y_{i}\epsilon$,
so the FOC for the final good producer is 
\[
Y^{\frac{1}{\sigma}}y_{i}^{-\frac{1}{\sigma}}=p_{i}=\min\left\{ w/\psi_{L},r/\psi_{i}^{K}\right\} 
\]
Intermediate producers still solve 
\[\max_{L_{i},K_{i}}p_{i}\left(\psi_{i}^{K}K_{i}+\psi^{L}L_{i}\right)-r_{i}K_{i}-wL_{i}\] 
To determine the equilibrium allocation, we define an iterative procedure. Let $\tilde{\gamma}$ be the largest sector index where producing with capital is \textit{strictly} cheaper than labor: 
\[
\tilde{\gamma}:=\max\left\{ i\in\left\{ 0,\dots,N\right\} :\frac{r}{\psi_{i}^{K}}<\frac{w}{\psi^{L}}\right\} 
\]
Consequently, tasks $\leq\tilde{\gamma}$ are produced entirely by capital. We then examine the \textit{next} sector, $\tilde{\gamma}+1$. If production costs equalize at this sector, it becomes the ``split'' sector where the intermediate good producer is indifferent between producing using labor versus capital. Let $\alpha\in[0,1]$ be the fraction of task $\tilde{\gamma}+1$ produced using capital; below we detail the procedure for solving this ratio.

\textbf{Factor prices and allocations given $\tilde{\gamma}$ and $\alpha$:} For labor, tasks $\tilde{\gamma}+2,\dots,N$ and the labor-portion of task $\tilde{\gamma}+1$ must clear: 
\[
\epsilon\left(\sum_{i=\tilde{\gamma}+2}^{N}L_{i}+(1-\alpha)L_{\tilde{\gamma}+1}\right)=L
\]
From the final good producer's FOC and the fact that labor productivity is equalized across all tasks, $L_{i}$ must be constant:
\[
L_{i}=\frac{L}{\epsilon\left(N-(\tilde{\gamma}+1)+1-\alpha\right)}=\frac{L}{\epsilon\left(N-\tilde{\gamma}-\alpha\right)}
\]
For capital, tasks $1,\dots,\tilde{\gamma}$ and the capital-portion of task $\tilde{\gamma}+1$ must clear: 
\[
\epsilon\left(\sum_{i=1}^{\tilde{\gamma}}K_{i}+\alpha K_{\tilde{\gamma}+1}\right)=K
\]
Using the optimal allocation rule $K_{i}\propto(\psi_{i}^{K})^{\sigma-1}$: 
\[
K_{i}=\frac{\left(\psi_{i}^{K}\right)^{\sigma-1}}{\sum_{m=1}^{\tilde{\gamma}}\left(\psi_{m}^{K}\right)^{\sigma-1}+\alpha\left(\psi_{\tilde{\gamma}+1}^{K}\right)^{\sigma-1}}\frac{K}{\epsilon}
\]
Plugging these allocations into the final good producer's FOC, we have that factor prices conditional on $\tilde{\gamma}$ and $\alpha$ are given by: 
\begin{align*}
r\left(\tilde{\gamma},\alpha\right) & =Y^{\frac{1}{\sigma}}\left(\frac{\epsilon\left(\sum_{m=1}^{\tilde{\gamma}}\left(\psi_{m}^{K}\right)^{\sigma-1}+\alpha\left(\psi_{\tilde{\gamma}+1}^{K}\right)^{\sigma-1}\right)}{K}\right)^{\frac{1}{\sigma}}\\
w\left(\tilde{\gamma},\alpha\right) & =Y^{\frac{1}{\sigma}}\left(\frac{\epsilon\left(N-\tilde{\gamma}-\alpha\right)\left(\psi^{L}\right)^{\sigma-1}}{L}\right)^{\frac{1}{\sigma}}
\end{align*}
\textbf{Solving Procedure:} First, consider the set
\[
\mathcal{C}:=\left\{ j:\frac{\psi_{j}^{K}}{\psi_{L}}>\frac{r\left(j,0\right)}{w\left(j,0\right)}\right\} 
\]
If this set is empty, set $\tilde{\gamma}=0$. Else, set $\tilde{\gamma}:=\max C$. Next, solve numerically for $\alpha\in\left[0,1\right]$ such that
\[
\frac{\psi_{\tilde{\gamma}+1}^{K}}{\psi^{L}}=\frac{r\left(\tilde{\gamma},\alpha\right)}{w\left(\tilde{\gamma},\alpha\right)}
\]

\subsection{Evolution Equations}

\label{discretized-evolution}

At every $t$, recall the state variable is the vector of data, $\boldsymbol{D}$. We first construct $\psi$ from data,  
\[
\psi_{i}^{K}\left(D\right)=f\left(\frac{i}{N}\right)\left(\sum_{m=1}^{N}g\left(\big|m-i\big|\epsilon\right)D_{m}\right)^{\eta}
\]
Note that setting $g\left(0\right)=1$ and 0 everywhere else corresponds to the no spillover case. Next, we determine $\tilde{\gamma}$ and $\alpha$ using the procedure above. The evolution of data depends on task output $y_{i}$. 
\begin{itemize}
\item For $i\le\tilde{\gamma}$ (Strictly Capital): $y_{i}=\psi_{i}^{K}K_{i}$. 
\item For $i>\tilde{\gamma}+1$ (Strictly Labor): $y_{i}=\psi^{L}L_{i}$. 
\item For the split sector $i=\tilde{\gamma}+1$: 
\begin{itemize}
\item If $\alpha=0$: $y_{i}=\psi^{L}L_{i}$ (Transitions to labor). 
\item If $\alpha>0$: Costs equalize, so $p_{i}=w/\psi^{L}$. Demand implies $y_{i}=\psi^{L}L_{labor}$, where $L_{labor}$ is the intensity $L_{i}$ calculated for fully labor sectors.\footnote{Since labor productivity for all sectors are the same, the combined capital and labor output for this sector must be equal to the output in labor only sectors.} 
\end{itemize}
\end{itemize}
Combining these, the data evolution is: 
\[
\frac{dD_{i}}{dt}=y_{i}-\rho D_{i}
\]
Finally, for Solow-style capital accumulation devoid of household optimization, we can augment the simulation by including the usual capital accumulation rule and keeping track of $K(t)$ 
\[
\dot{K}=sY-\delta K
\]
We simulate this system in Matlab. 

% \subsection{Discussion} In the long run, the discretized system is not a faithful representation of our model consisting of a continuum of tasks. This is because at most $N-1$ tasks can be automated in the discretized economy given labor market clearing. This feature has implications for wages, rental rate, and output in very long horizons. As such, we mainly rely on our simulations to highlight patterns that arise during the transition as opposed to investigating the limiting behavior of the discretized system. 

%%%%%%%%%%%%%%%%%%%%%%%%%%%%%%%%
%%%%%%%%%%%%%%%%%%%%%%%%%%%%%%%%

%\newpage
\section{Additional Numerical Results} \label{appendix:additional_numerical_results}
We collect additional numerical results in this appendix.
\subsection{Time Paths Wages and Outputs without Capital Accumulation}

\begin{figure}[h]
    \centering
    % --- Top Row (Wages) ---
    \subfloat[Evolution of $w$ for $\sigma = 0.5$ \label{fig:comp_wage}]{%
        \includegraphics[width=0.48\textwidth]{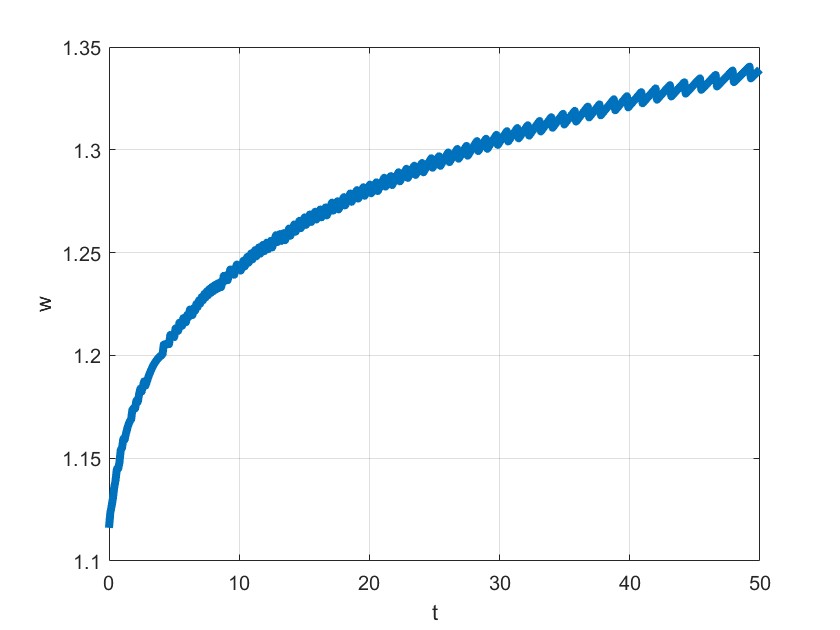}
    }
    \hfill % Adds horizontal space
    \subfloat[Evolution of $w$ for $\sigma = 5.5$ \label{fig:sub_wage}]{%
        \includegraphics[width=0.48\textwidth]{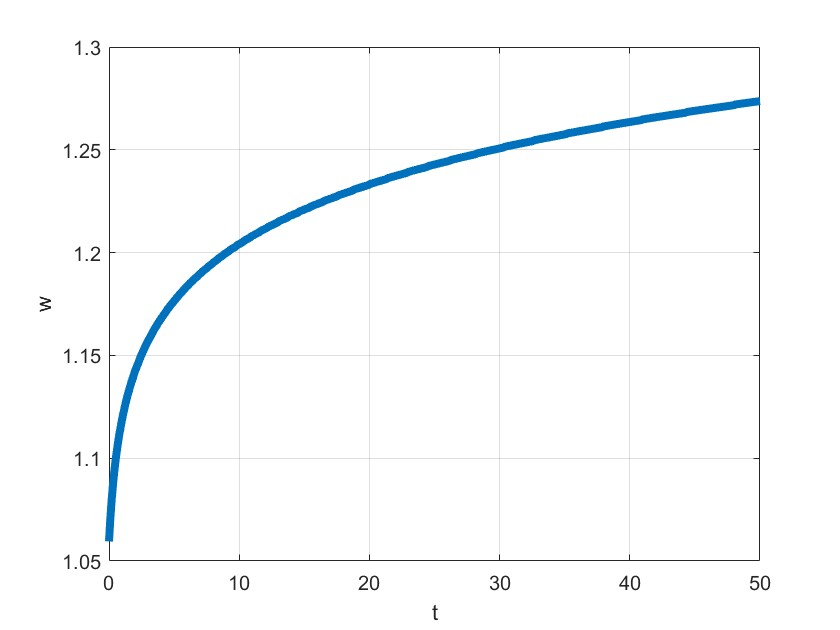}
    }
    \par\medskip % Adds vertical space between rows
    % --- Bottom Row (Output) ---
    \subfloat[Evolution of $Y$ for $\sigma = 0.5$ \label{fig:comp_output}]{%
        \includegraphics[width=0.48\textwidth]{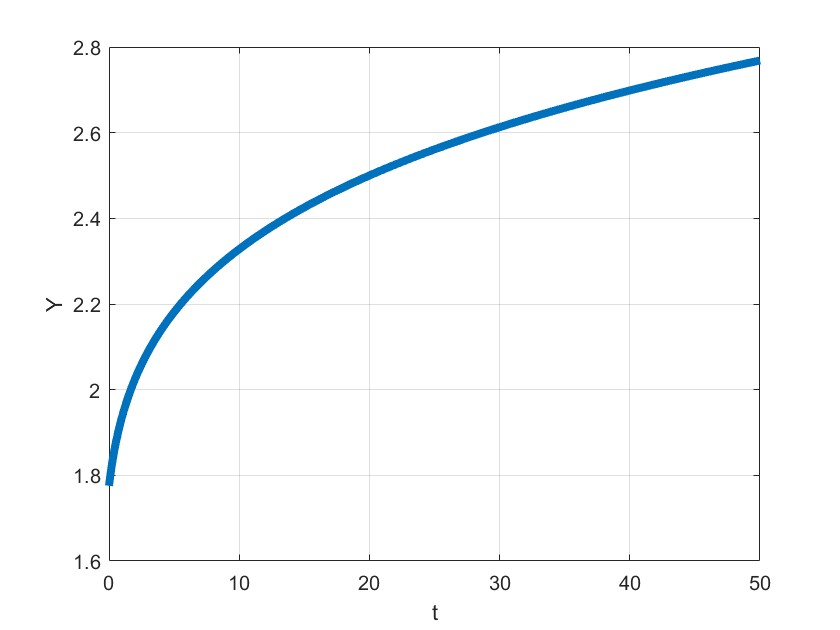}
    }
    \hfill % Adds horizontal space
    \subfloat[Evolution of $Y$ for $\sigma = 5.5$ \label{fig:sub_output}]{%
        \includegraphics[width=0.48\textwidth]{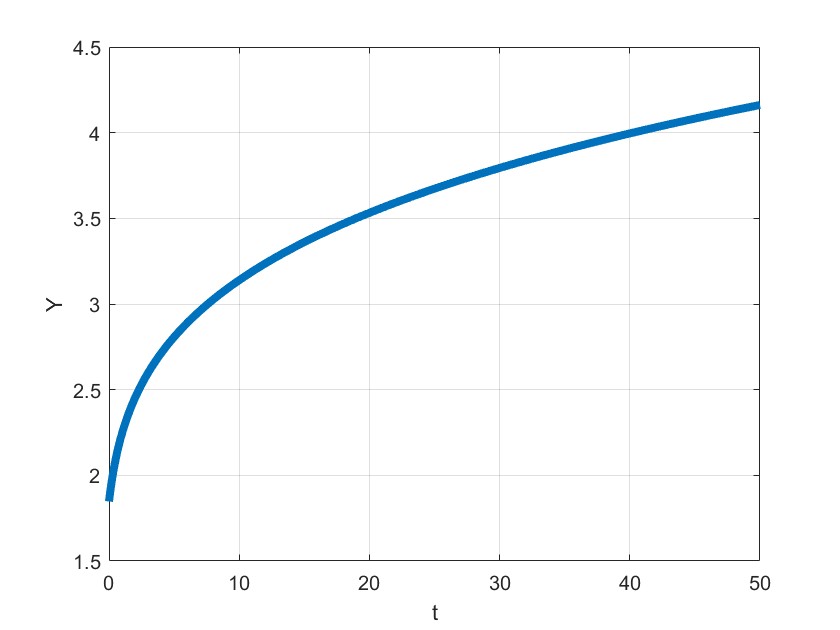}
    }
    % --- Combined Caption ---
    \caption{The time-path of wages and output. In the complements case, wages rise due to increased capital productivity, but this effect will eventually be fully offset by the displacement effect as labor is restricted endogenously to a shrinking set of tasks.}
    \label{fig:output_and_wages_no_acc}
\end{figure}

\clearpage
\newpage

% \newpage
% \subsection{Spillovers and the speed of automaton} \label{appendix:speed_spillovers}
% Below we explore how spillovers affect the speed of automation. All parameters are the same as in \cref{fig:complements} except for task substitutability. We use a uniform kernel to aggregate data in the case with spillovers. The kernel for any given task $i$ is supported on $[\max\{i-0.1,0\}, \min\{i+0.1, 1\}]$.

% \begin{figure}[h]
%     \centering
%     \caption{Automation with and without spillovers}
%     \subfloat[$\sigma = 0.5$ \label{fig:left}]{%
%         \includegraphics[width=0.5\textwidth]{speed_spillovers/sigma0p5_spd_spillovers.jpg}%
%     }
%     \hfill 
%     \subfloat[$\sigma = 5.5$ \label{fig:right}]{%
%         \includegraphics[width=0.5\textwidth]{speed_spillovers/sigma5p5_spd_spillovers.jpg}%
%     }
%     \label{fig:speed_spillovers}
% \end{figure}

\subsection{Peripheral Automation Across Substitutability and Bridge Strength} \label{appendix:peripheral automation sim}

\begin{figure}[H]
    \centering
    \includegraphics[width=\linewidth]{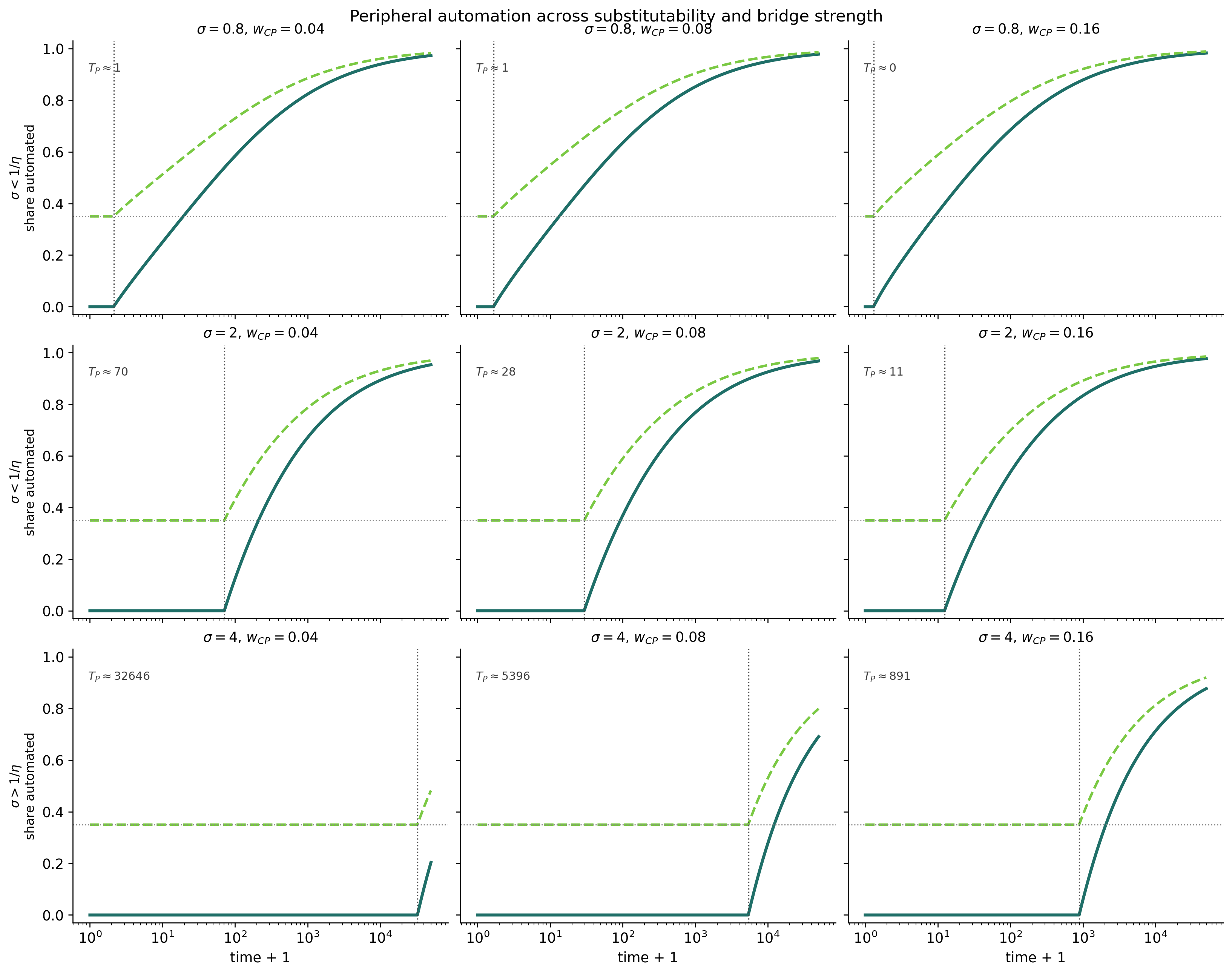}
    \caption{}
    \label{fig:core-periphery-grid}
\end{figure}

\subsection{Capital Accumulation \`a la Solow}\label{sec:solow_numerics}
Below we collect additional simulation results with capital accumulation. All parameters are the same as in the baseline, except we assume an exogenous savings rate of $s=0.02$ and $\delta=0.01$.

\begin{figure}[ht!]
    \centering % Center the entire figure
    % The main caption for the entire figure
    \caption{$\sigma=0.5$ with Solow-style capital accumulation.}
    % --- Top Row ---
    \subfloat[\label{fig:sub1}]{\includegraphics[width=0.33\textwidth]{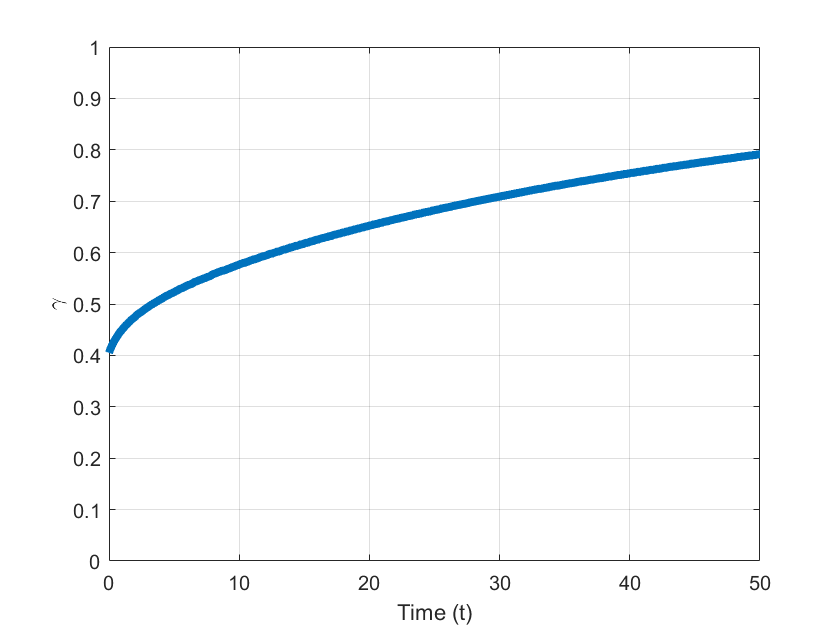}}
    \hfill % Adds horizontal space between the subfigures
    \subfloat[\label{fig:sub2}]{\includegraphics[width=0.33\textwidth]{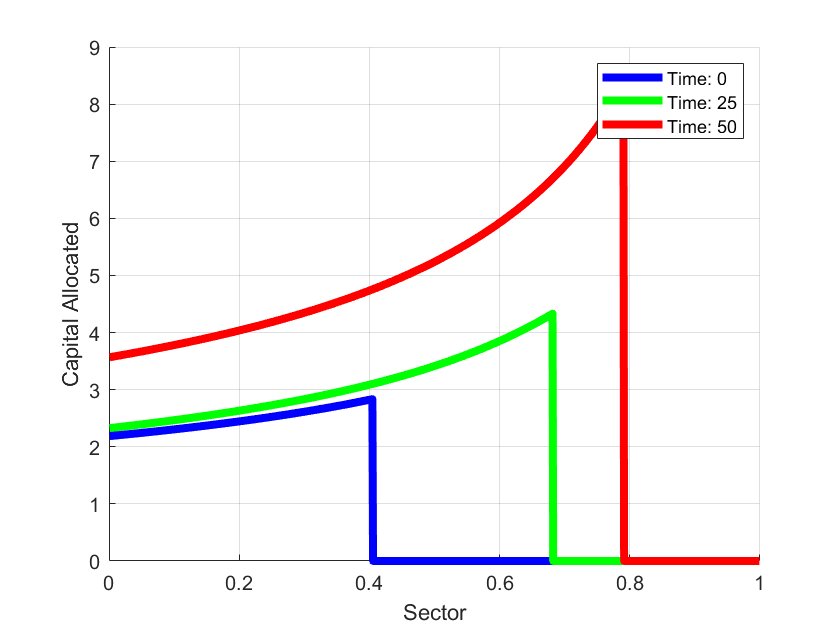}}
    \hfill
    \subfloat[\label{fig:sub3}]{\includegraphics[width=0.33\textwidth]{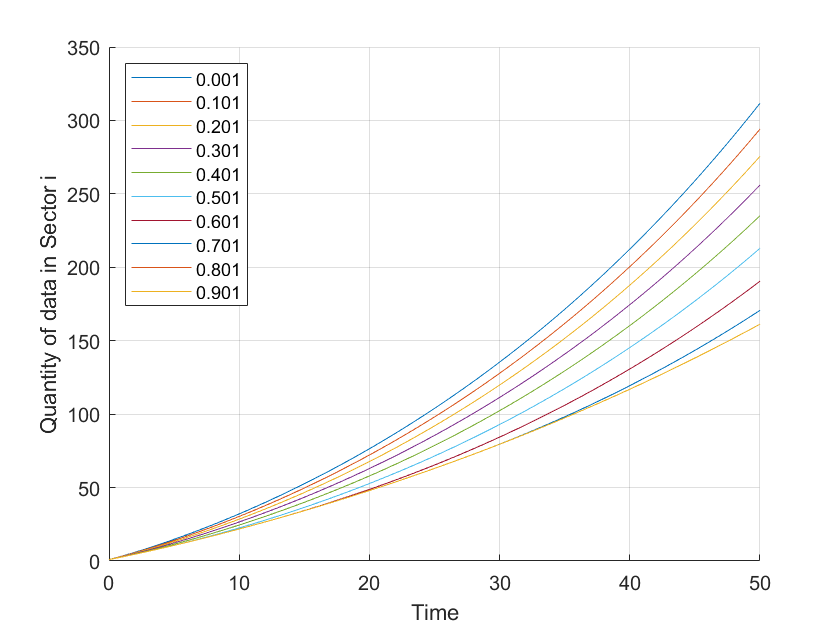}}

    \vspace{0.5cm} % Adds vertical space between rows

    % --- Middle Row ---
    \subfloat[\label{fig:sub4}]{\includegraphics[width=0.33\textwidth]{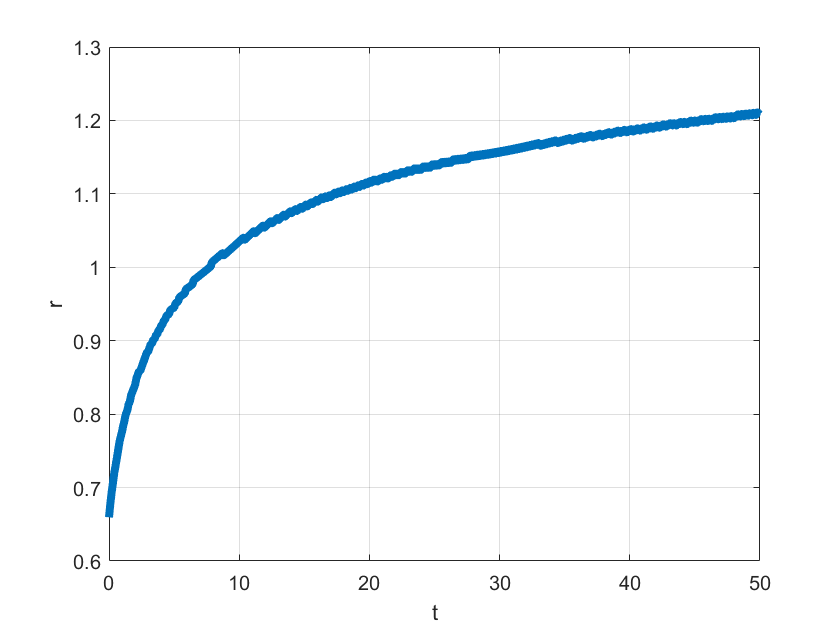}}
    \hfill
    \subfloat[\label{fig:sub5}]{\includegraphics[width=0.33\textwidth]{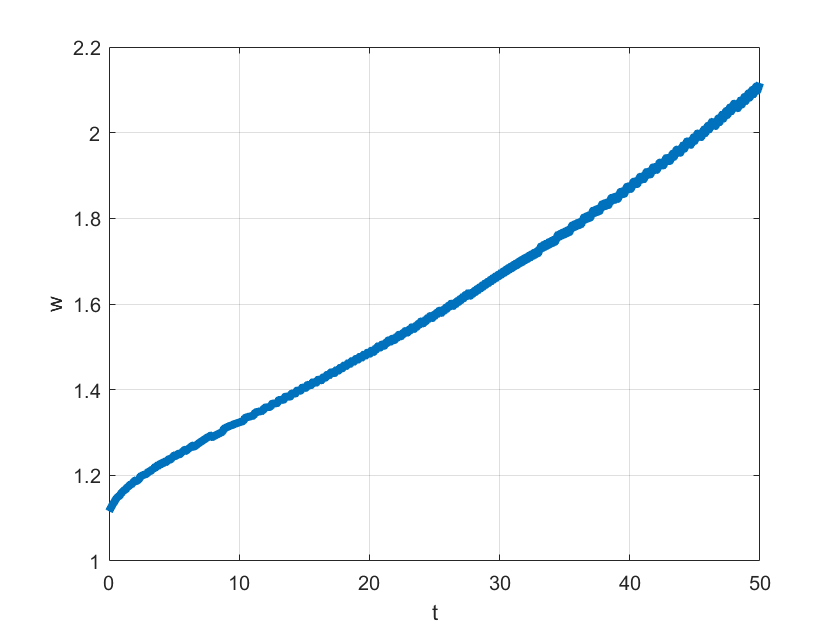}}
    \hfill
    \subfloat[\label{fig:sub6}]{\includegraphics[width=0.33\textwidth]{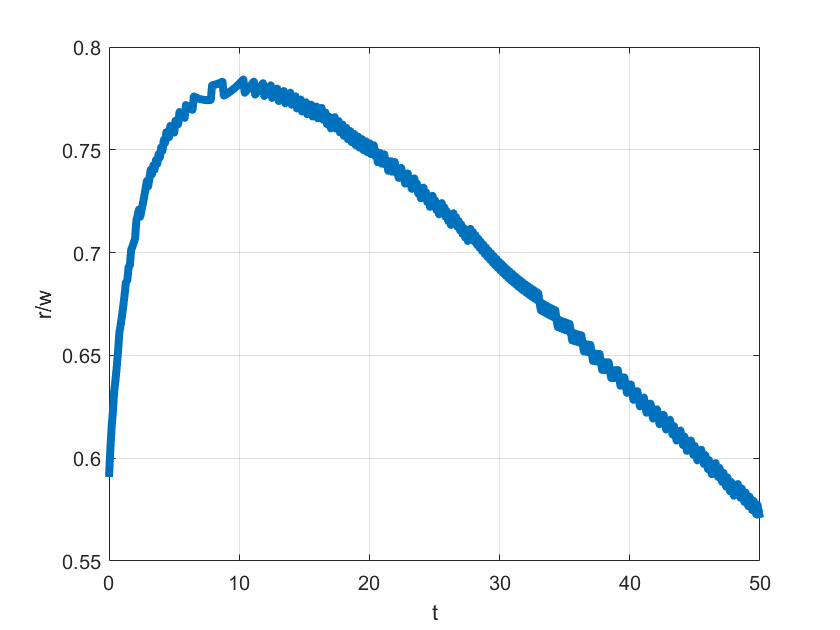}}
    
    \vspace{0.5cm} % Adds vertical space between rows

    % --- Bottom Row ---
    \subfloat[\label{fig:sub7}]{\includegraphics[width=0.33\textwidth]{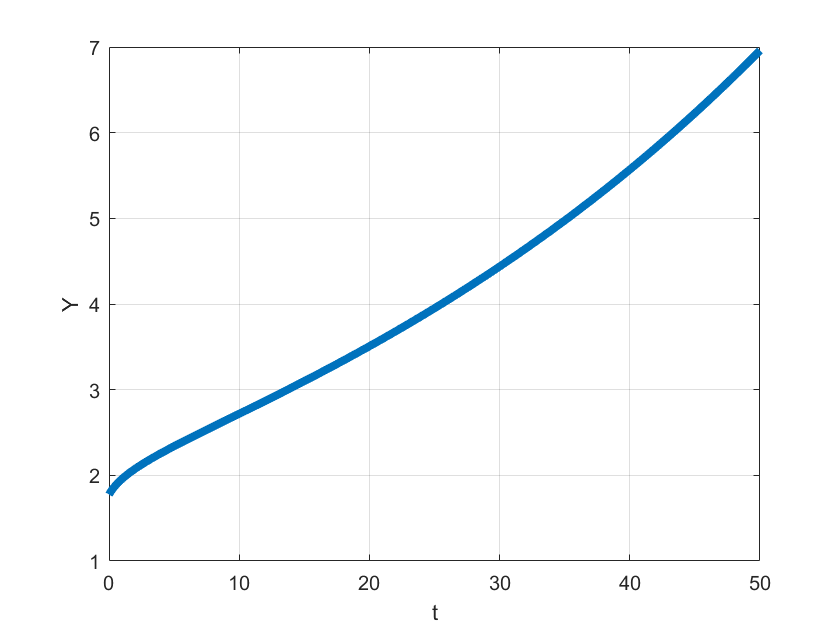}}
    \hfill
    \subfloat[\label{fig:sub8}]{\includegraphics[width=0.33\textwidth]{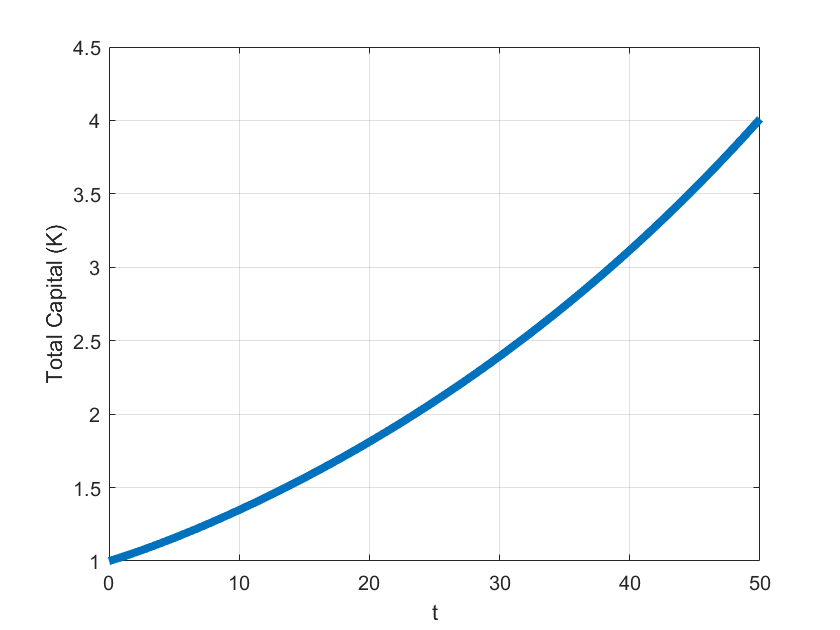}}
    \hfill
    \subfloat[\label{fig:sub9}]{\includegraphics[width=0.33\textwidth]{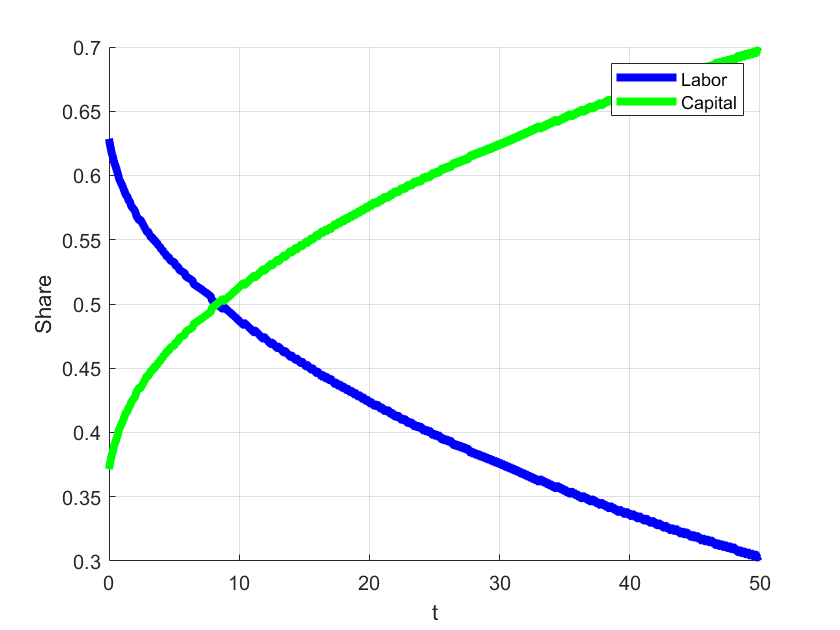}}
    
    \label{fig:complements_with_acc}
\end{figure}

\begin{figure}[ht!]
    \centering % Center the entire figure
    % The main caption for the entire figure
    \caption{$\sigma=5.5$ with Solow-style capital accumulation.}
    % --- Top Row ---
    \subfloat[\label{fig:sub1}]{\includegraphics[width=0.33\textwidth]{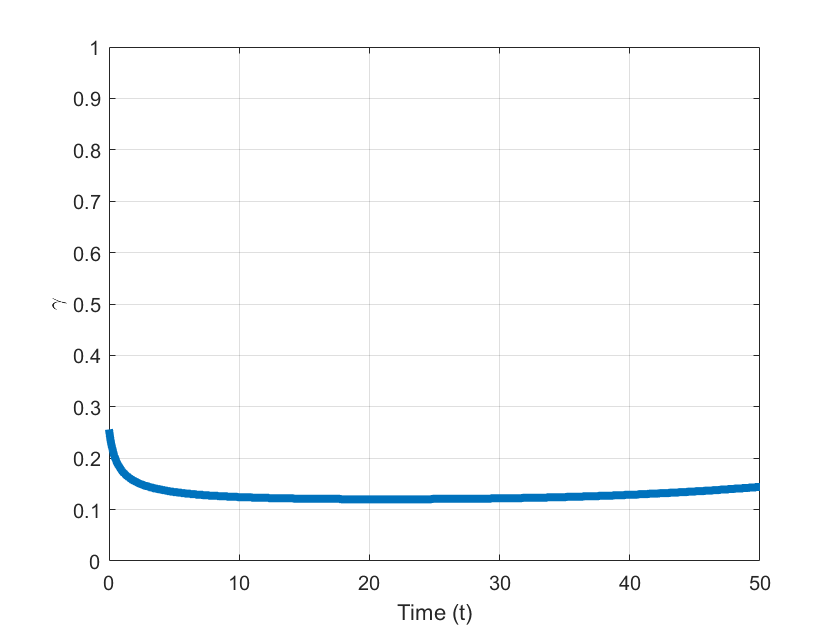}}
    \hfill % Adds horizontal space between the subfigures
    \subfloat[\label{fig:sub2}]{\includegraphics[width=0.33\textwidth]{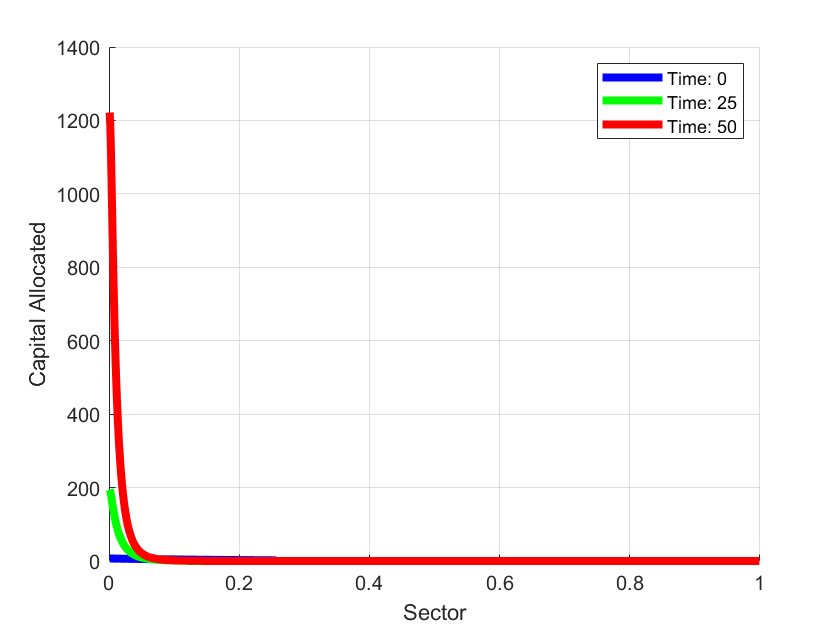}}
    \hfill
    \subfloat[\label{fig:sub3}]{\includegraphics[width=0.33\textwidth]{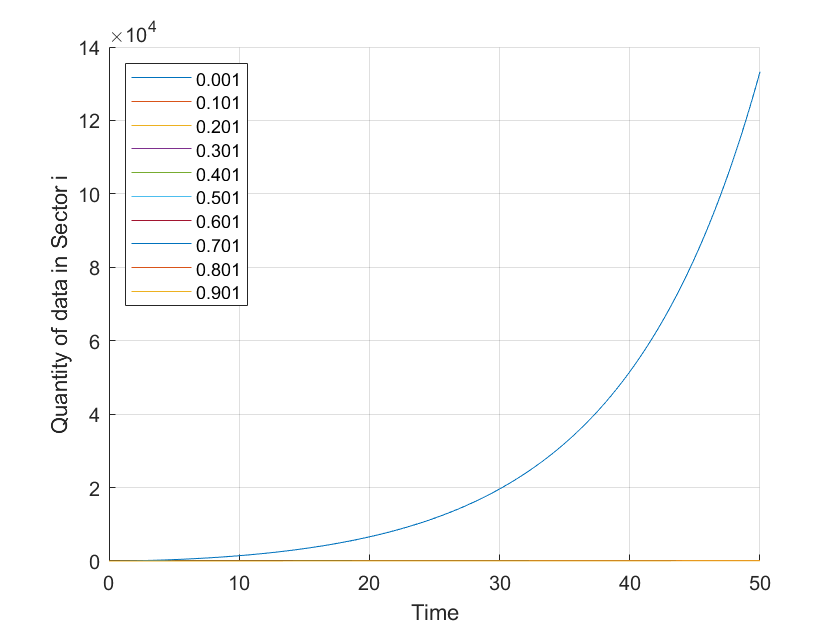}}

    \vspace{0.5cm} % Adds vertical space between rows

    % --- Middle Row ---
    \subfloat[\label{fig:sub4}]{\includegraphics[width=0.33\textwidth]{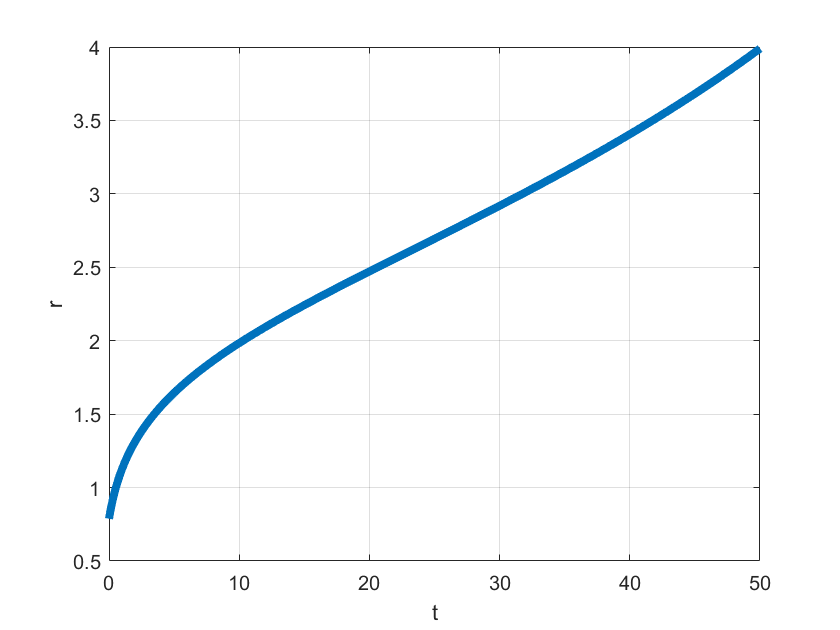}}
    \hfill
    \subfloat[\label{fig:sub5}]{\includegraphics[width=0.33\textwidth]{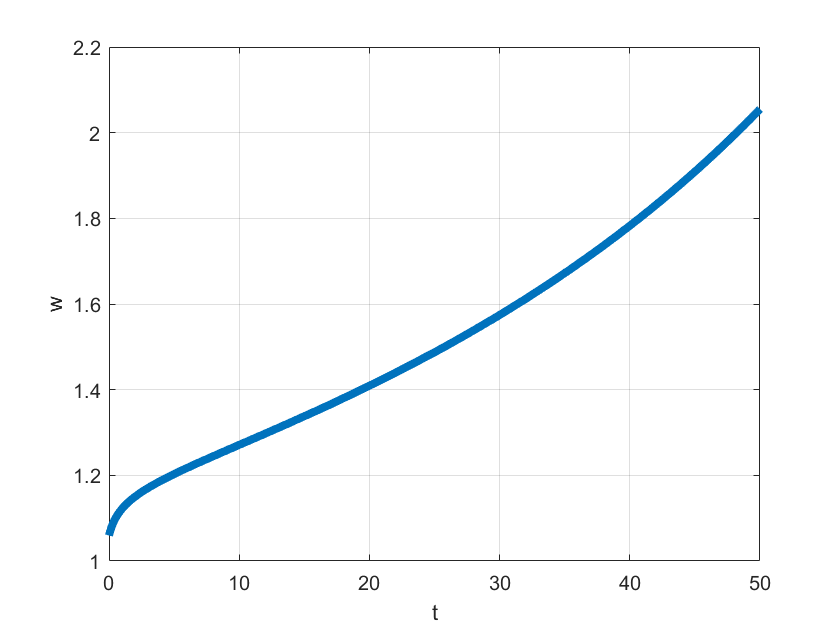}}
    \hfill
    \subfloat[\label{fig:sub6}]{\includegraphics[width=0.33\textwidth]{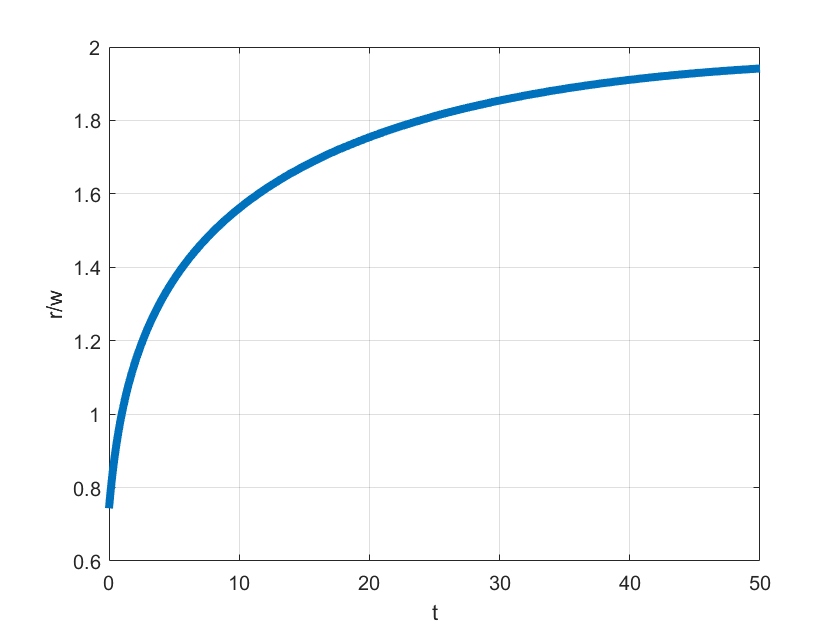}}
    
    \vspace{0.5cm} % Adds vertical space between rows

    % --- Bottom Row ---
    \subfloat[\label{fig:sub7}]{\includegraphics[width=0.33\textwidth]{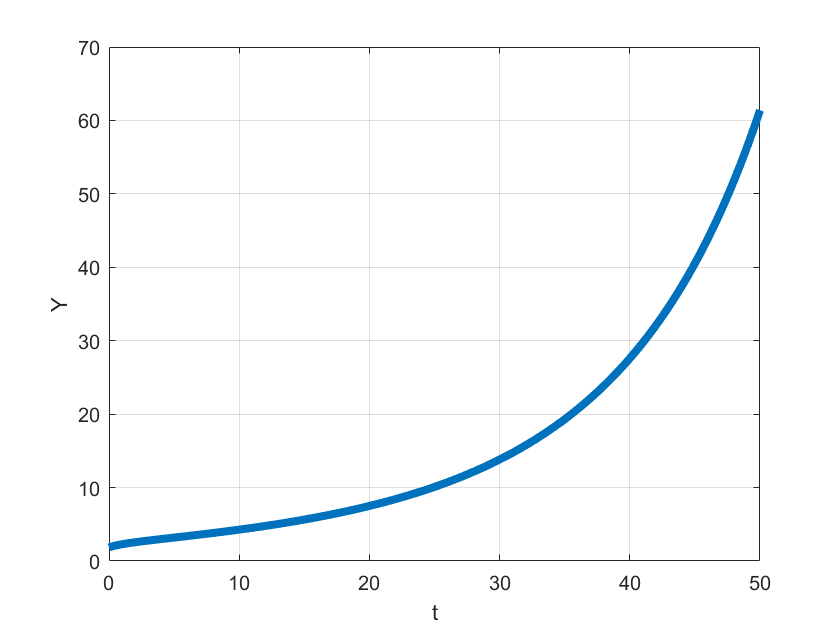}}
    \hfill
    \subfloat[\label{fig:sub8}]{\includegraphics[width=0.33\textwidth]{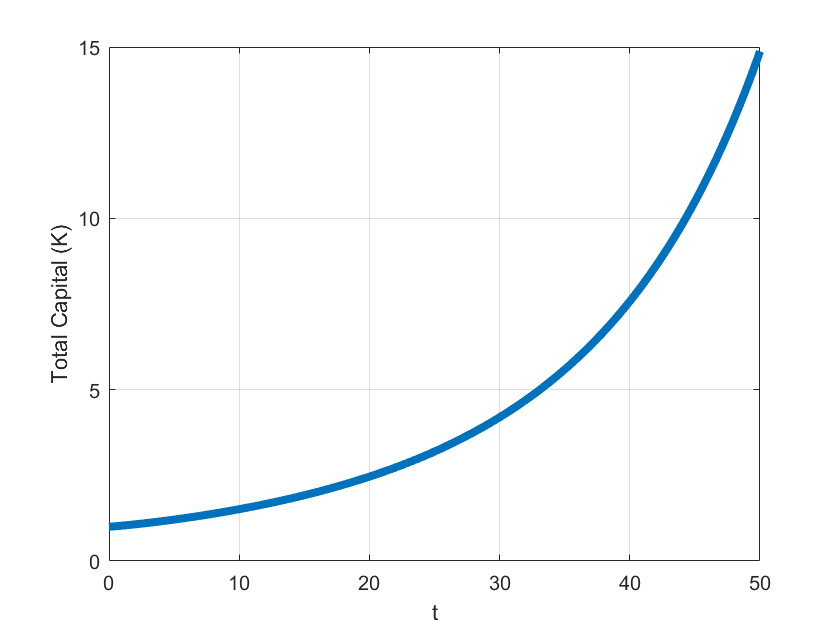}}
    \hfill
    \subfloat[\label{fig:sub9}]{\includegraphics[width=0.33\textwidth]{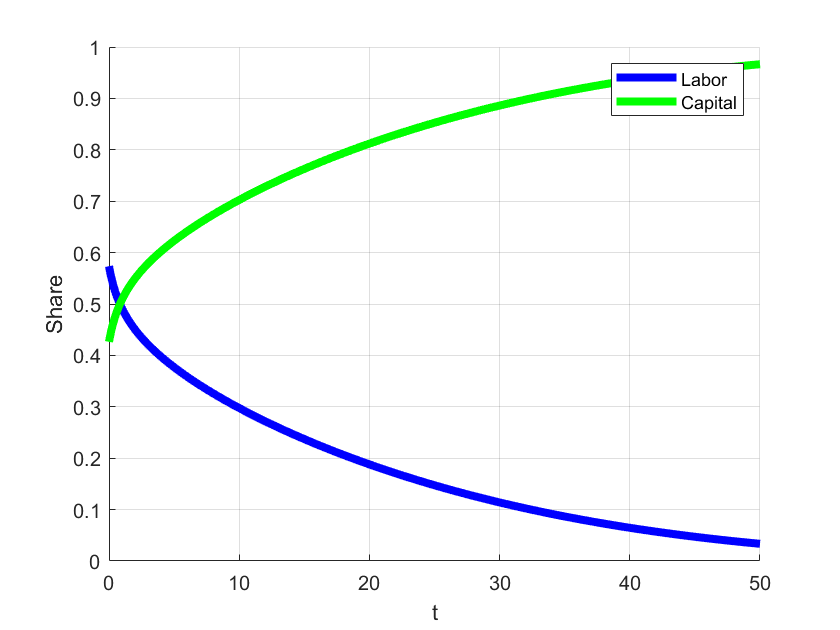}}
    \label{fig:substitutes_with_acc}
\end{figure}
\clearpage
\newpage

\subsection{Planner's Optimal Allocation} \label{appendix:planner_numerics}
Below we collect simulation results from the two-task economy. The first two figures are optimal paths chosen by the planner under the `$AK$' simplification. We set $f_h = f_l = 1$, $D_{h0} = 10$, $D_{l0}=1$, $\eta = 0.2$, $K = 1$, and $\rho = 0.05$.
\begin{figure}[ht!]
    \centering
    \caption{Planner's allocation at $\sigma = 0.5$}
    \includegraphics[width=\linewidth]{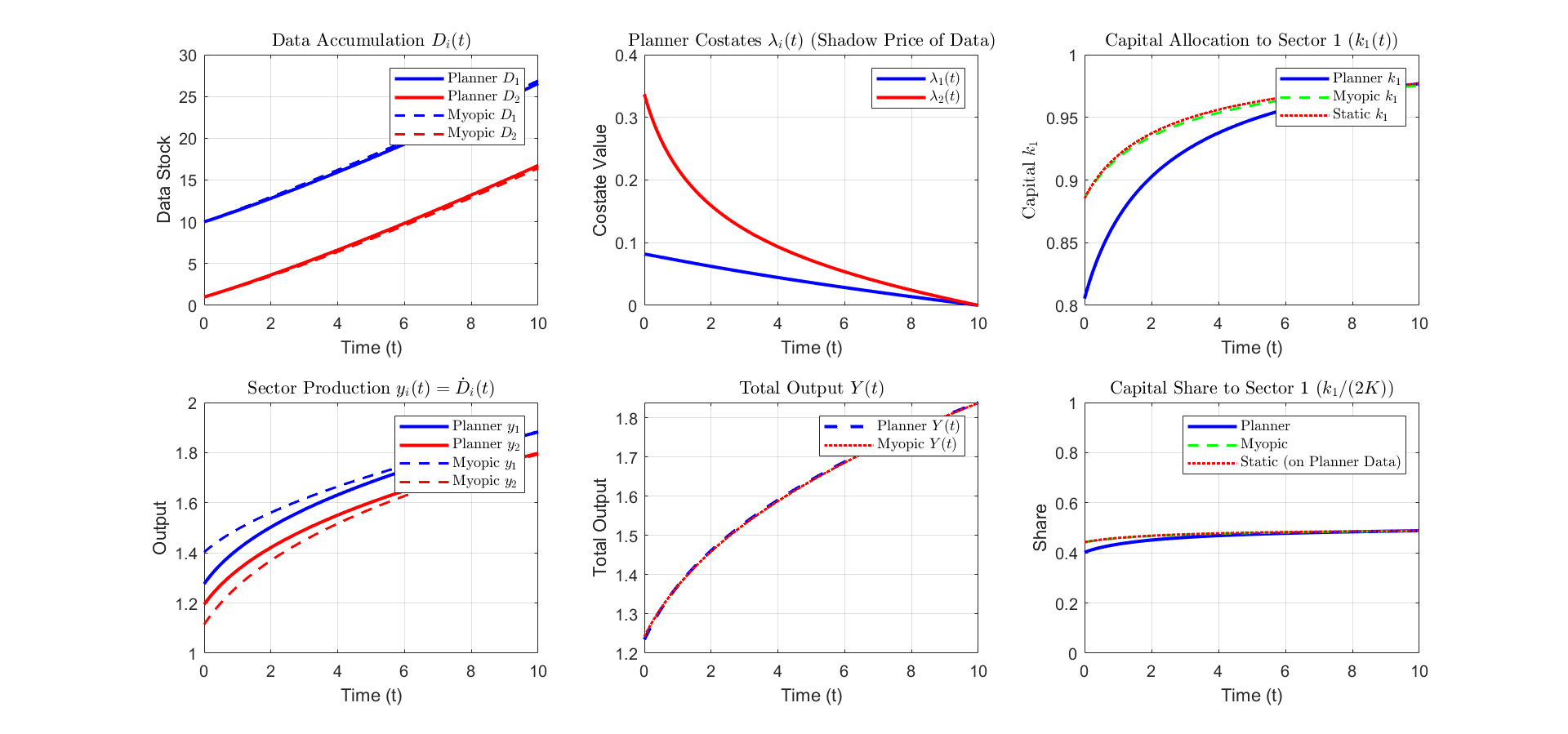}
    \label{fig:two_task_comp}
\end{figure}
\begin{figure}[h!]
    \caption{Planner's allocation at $\sigma = 5.5$}
    \centering
    \includegraphics[width=\linewidth]{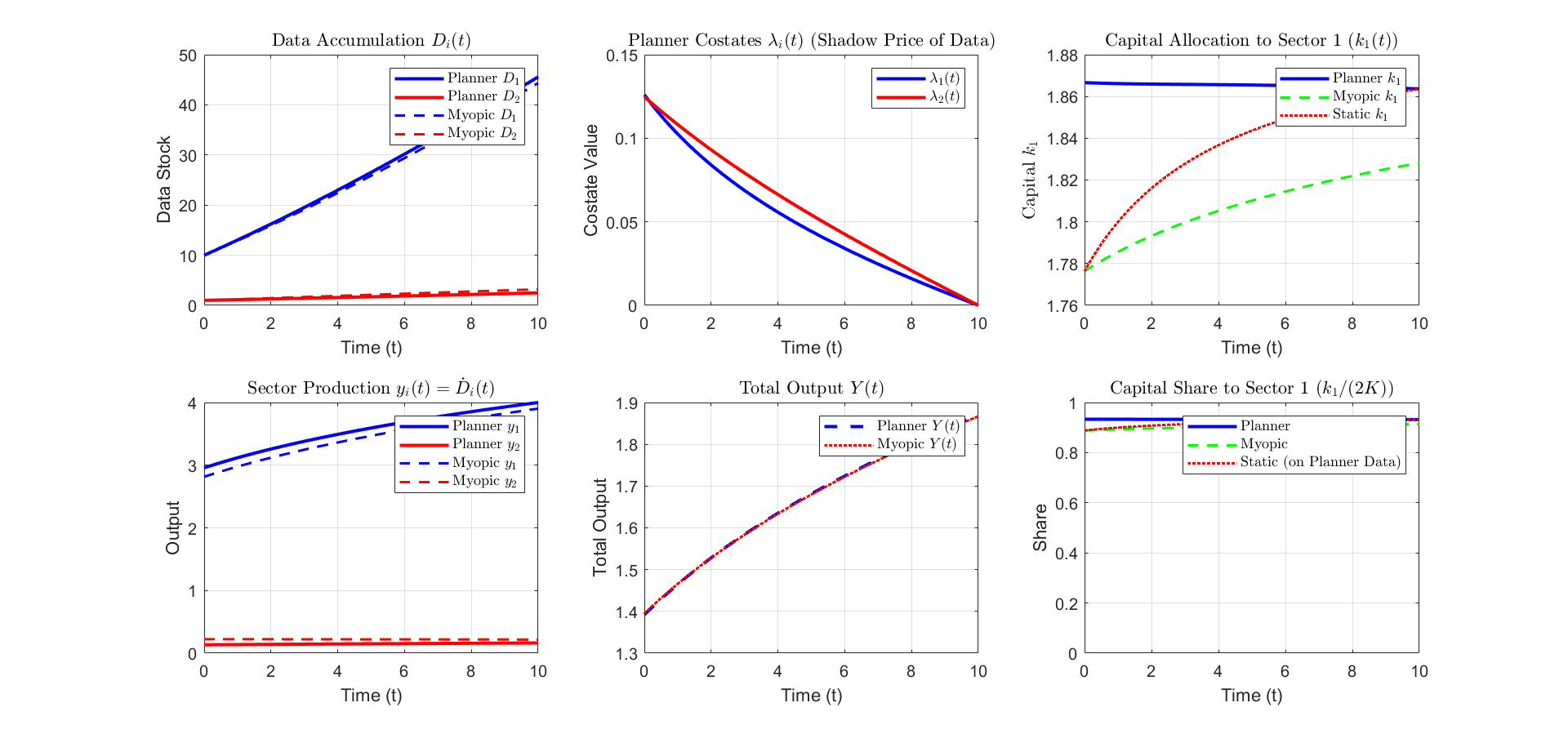}
    \label{fig:two_task_sub}
\end{figure}
\clearpage
Next, we simulate the complement case with labor included. We set $\sigma = 0.5$, $f_h = f_l = 1$, $D_{h0} = 100$, $D_{l0}=1$, $\eta = 0.2$, $K=L=1$, and $\rho = 0.05$.
\begin{figure}[h!]
    \caption{Planner's allocation at $\sigma = 0.5$ with labor included}
    \centering
    \includegraphics[width=\linewidth]{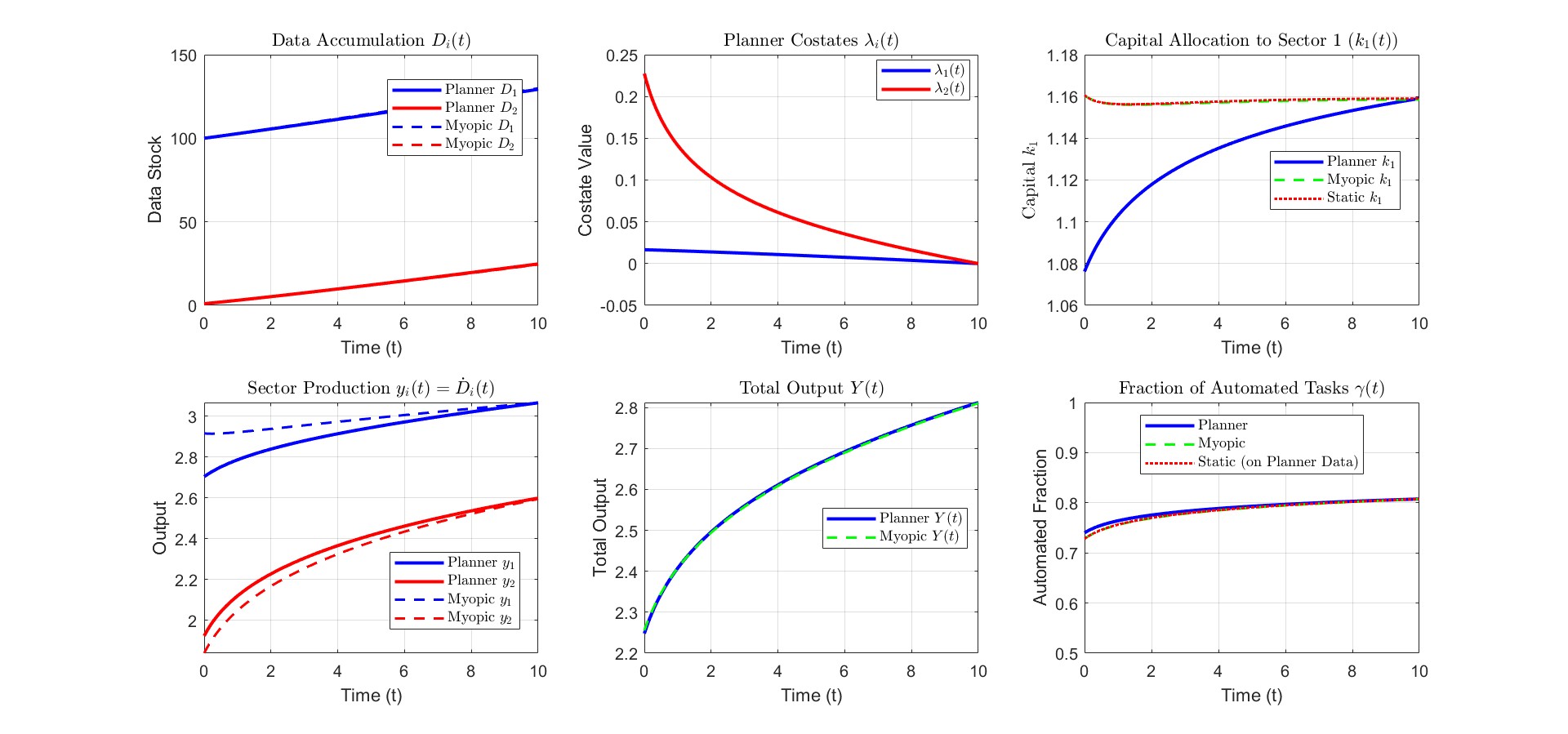}
    \label{fig:two_task_comp_labor}
\end{figure}

\end{document}